\pdfoutput=1
\documentclass[11pt]{article}
\usepackage[letterpaper,margin=1in]{geometry}

\usepackage[utf8]{inputenc} 
\usepackage[T1]{fontenc}    
\usepackage{url}            
\usepackage{booktabs}       
\usepackage{amsfonts}       
\usepackage{nicefrac}       
\usepackage{microtype}      
\usepackage{xcolor}         

\usepackage{nicefrac,bm,bbm}
\usepackage{amsmath,amsthm,color,colortbl,amssymb}

\usepackage{hyperref}

\usepackage{caption}
\usepackage{natbib}
\renewcommand{\bibsection}{\section*{References}}
\usepackage{wrapfig}
\usepackage{comment}

\usepackage{graphicx}
\usepackage{multirow}
\usepackage{xspace}
\usepackage{mathtools}
\usepackage{euscript}
\usepackage{thm-restate}
\usepackage{tikz}
\usepackage{subcaption}
\usepackage{tablefootnote,tabularx}
\usepackage[T1]{fontenc}    

\usepackage{pifont}

\usepackage{euscript}
\usepackage{mleftright}

\usepackage[capitalize]{cleveref}

\crefalias{AlgoLine}{line}%
\crefname{algocf}{Algorithm}{Algorithms}

\usepackage{cancel}
\usepackage{newpxtext}
\usepackage{newpxmath}
\definecolor{niceRed}{RGB}{190,38,38}
\definecolor{Red2}{RGB}{219, 50, 54}
\definecolor{mgreen}{RGB}{160, 200, 140}
\definecolor{blueGrotto}{RGB}{5,157,192}
\definecolor{limeGreen}{HTML}{81B622}
\definecolor{myellow}{rgb}{0.88,0.61,0.14}
\definecolor{darkGreen}{HTML}{2E8B57}
\definecolor{navyBlueP}{HTML}{03468F}
\definecolor{Sepia}{HTML}{7F462C}
\definecolor{red2}{HTML}{1F462C}
\definecolor{orange2}{HTML}{FF8000}
\definecolor{mgray}{HTML}{ABB3B8}
\definecolor{myPurple}{RGB}{175,0,124}
\definecolor{mypurple2}{rgb}{0.8,0.62,1}
\definecolor{royalBlue}{HTML}{057DCD}
\definecolor{mpink}{HTML}{FC6C85}
\theoremstyle{plain}
\newtheorem{theorem}{Theorem}[section]

\theoremstyle{definition}
\newtheorem{definition}[theorem]{Definition}

\theoremstyle{remark}

\hypersetup{
	colorlinks = true,
	linkcolor = niceRed,
	citecolor = royalBlue,
	linktocpage = true,
	urlcolor = mpink
}

\usepackage{amsmath,amsfonts,bm}

\def\eqref#1{equation~\ref{#1}}

\def\1{\bm{1}}

\DeclareMathAlphabet{\mathsfit}{\encodingdefault}{\sfdefault}{m}{sl}
\SetMathAlphabet{\mathsfit}{bold}{\encodingdefault}{\sfdefault}{bx}{n}

\DeclareMathOperator*{\argmax}{arg\,max}

\usepackage{url}
\allowdisplaybreaks

\usepackage{multirow}
\usepackage{tikz}
\usepackage[utf8]{inputenc} 
\usepackage[T1]{fontenc}    
\usepackage{hyperref}       
\usepackage{url}            
\usepackage{booktabs}       
\usepackage{amsfonts}       
\usepackage{nicefrac}       
\usepackage{microtype}      
\usepackage{xcolor}         
\usepackage{amsmath}
\usepackage{amsthm,enumerate,amssymb}
\usepackage{thm-restate}
\usepackage{thmtools} 
\usepackage{mathtools}
\usepackage{wrapfig}
\usepackage{amsmath}
\usepackage{pifont}

\usepackage{comment}
\usepackage{graphics}
\usepackage{graphicx}
\usepackage{appendix}
\usepackage{algorithm}
\usepackage{enumitem}
\usepackage[noend]{algpseudocode}
\algdef{SE}[DOWHILE]{Do}{doWhile}{\algorithmicdo}[1]{\algorithmicwhile\ #1}%
\usepackage{bbm}
\algnewcommand{\IfThen}[2]{
	\State \algorithmicif\ #1\ \algorithmicthen\ #2}

\usepackage{tikzscale} 

\usepackage[short, nocomma]{optidef} 

\newcommand{\sset}{\mathcal{S}}
\newcommand{\hset}{\mathcal{H}}

\newcommand{\A}{\mathcal{A}}

\newcommand{\vprincipal}{V^{\textnormal{P},\pi}}
\newcommand{\vagent}{V^{\textnormal{A},\pi}}
\newcommand{\vagentdev}{\widehat{V}^{\textnormal{A},\pi}}

\newcommand{\qprincipal}{Q^{\textnormal{P},\pi}}
\newcommand{\qagent}{Q^{\textnormal{A},\pi}}
\newcommand{\qagentdev}{\widehat{Q}^{\textnormal{A},\pi}}

\newcommand{\vprincipalpr}{V^{\textnormal{P},\sigma}}
\newcommand{\vagentpr}{V^{\textnormal{A},\sigma}}
\newcommand{\vagentdevpr}{\widehat{V}^{\textnormal{A},\sigma}}

\newcommand{\qprincipalpr}{Q^{\textnormal{P},\sigma}}
\newcommand{\qagentpr}{Q^{\textnormal{A},\sigma}}
\newcommand{\qagentdevpr}{\widehat{Q}^{\textnormal{A},\sigma}}

\newcommand{\vprincipalmark}{V^{\textnormal{P},\rho}}
\newcommand{\vagentmark}{V^{\textnormal{A},\rho}}

\newcommand{\X}{\mathcal{X}}

\usepackage[textsize=tiny]{todonotes}

\algblockdefx[Execute]{Execute}{EndExecute}[1]{\textbf{until} #1 \textbf{run}}{\textbf{end execute}}
\makeatletter
\ifthenelse{\equal{\ALG@noend}{t}}%
{\algtext*{EndExecute}}
{}%
\makeatother

\algblockdefx[IfInline]{IfInline}{EndIfInline}[1]{\textbf{if} #1 \textbf{then} }{\textbf{end if}}
\makeatletter
\ifthenelse{\equal{\ALG@noend}{t}}%
{\algtext*{EndIfInline}}
{}%
\makeatother

\definecolor{mygreen}{rgb}{0.0, 0.5, 0.0}
\definecolor{myorange}{rgb}{0.55, 0.62, 1}

\title{Contracting With a Reinforcement Learning Agent by Playing Trick or Treat}

\author{}
\renewcommand{\eqref}[1]{(\ref{#1})}

\author{
	Matteo Bollini \and
	Francesco Bacchiocchi\and
	Matteo Castiglioni \and
	Alberto Marchesi  \and
	Nicola Gatti\\
	 Politecnico di Milano, Milan, Italy\\
	\texttt{name.surname@polimi.it}
}

\begin{document}
\maketitle

\begin{abstract}
	We study \emph{principal-agent} problems where a \emph{farsighted} agent takes costly actions in an MDP.
	The core challenge in these settings is that agent's actions are \emph{hidden} to the principal, who can only observe their outcomes, namely state transitions and their associated rewards.
	Thus, the principal's goal is to devise a policy that incentives the agent to take actions leading to desirable outcomes.
	This is accomplished by committing to a payment scheme (a.k.a.~\emph{contract}) at each step, specifying a monetary transfer from the principal to the agent for every possible outcome.
	Interestingly, we show that \emph{Markovian} policies are unfit in these settings, as they do \emph{not} allow to achieve the optimal principal's utility and are constitutionally intractable.
	Thus, accounting for \emph{history} in unavoidable, and this begets considerable additional challenges compared to standard MDPs.
	Nevertheless, we design an efficient algorithm to compute an optimal policy, leveraging a compact way of representing histories for this purpose.
	Unfortunately, the policy produced by such an algorithm cannot be readily implemented, as it is only approximately \emph{incentive compatible}, meaning that the agent is incentivized to take the desired actions only approximately.
	To fix this, we design an efficient method to make such a policy incentive compatible, by only introducing a negligible loss in principal's utility.
	This method can be generally applied to any approximately-incentive-compatible policy, and it generalized a related approach that has already been discovered for classical principal-agent problems to more general settings in MDPs.
\end{abstract}

\section{Introduction}

Over the last few years, \emph{principal-agent problems} have received a growing attention from algorithmic game theory community~\citep{dutting2019simple,alon2021contracts,guruganesh2021contracts,dutting2021complexity,dutting2022combinatorial,castiglioni2022bayesian,alon2023bayesian,castiglioni2023designing,castiglioni2023multi,GuruganeshPower23}.
At their core, these problems capture scenarios where a principal tries to steer the behavior of a self-interested agent towards favorable outcomes.
In the simplest version of the problem, the agent has to take a costly action resulting in a stochastic outcome that determines a reward for the principal.
The key challenge stems from the fact that the agent's action is \emph{hidden} to the principal, who can only observe its resulting outcome.
In order to incentivize the agent to take actions leading to favorable outcomes, the principal has the ability commit to a \emph{contract}, which is payment scheme specifying a monetary transfer from the principal to the agent for every possible outcome.

Thanks to their generality and flexibility, principal-agent problems find application in a terrific number of settings, ranging from crowdsourcing platforms~\citep{ho2015adaptive}, blockchain-based smart contracts~\citep{cong2019blockchain}, and healthcare~\citep{bastani2016analysis} to delegation of online search problems and machine learning tasks~\citep{bechtel2022delegated,kleinberg2018delegated}.

Despite the recent surge in popularity of principal-agent problems, most of the works on the topic focus on ``one-shot'' problems in which the principal and the agent interact only one time.
Nevertheless, in most of the real-world problems of interest, it is natural to assume that the principal and the agent can interact multiple times in a \emph{sequential} manner.
Recently, some works~\citep{ho2015adaptive,cohen2022learning,zhu2023sample,bacchiocchilearning,chen2024bounded} tackled diverse learning tasks in repeated principal-agent problems.
Moreover, very recently, some concurrent works~\citep{ivanov2024principal,wu2024contractual} initiated the study of settings where the principal and the agent sequentially interact in a \emph{Markov decision process} (MDP).
However, these kind of interactions are still largely under-investigated and bring in several open questions that need to be addressed.

In this paper, we study principal-agent problems in which the agent takes \emph{hidden} actions in a time-inhomogeneous finite-horizon MDP.
At each step of the MDP, based on the current state and the performed action, the agent incurs in a cost and the environment transitions to a new state.
Moreover, the principal gets a reward that only depends on the current state and the next one (\emph{i.e.}, the outcome of the action taken by the agent).
Notice that principal's rewards are naturally independent of agent's actions, since the principal cannot observe them.
%
%
We study settings in which the principal can commit to a (randomized) \emph{policy} specifying contracts to be proposed to the agent at each step of the MDP. 
%
%
In particular, we study the most general and relevant case in which both the principal and the agent are \emph{farsighted}, namely, they both aim at maximizing their long-term cumulative expected utilities.

\subsection{Original contributions}

The focus of this paper is on the problem of efficiently computing an \emph{optimal} policy to commit to for the principal, \emph{i.e.}, one maximizing principal's cumulative expected utility given that the agent is \emph{farsighted} and best responds to the commitment so as to maximize their cumulative expected utility.
The main computational challenges of this problem beget from the fact that, when the agent is farsighted, \emph{Markovian} policies are \emph{not} sufficient to achieve optimality.
Indeed, the principal can increase their utility by ``threatening'' the farsighted agent to provide them little future payments in the case in which they do \emph{not} undertake the action plan devised by the principal.
To do so, the principal needs policies capable on conditioning their behavior on the whole history of interaction with the agent.
We call such policies \emph{history-dependent}, as they are clearly \emph{non-Markovian} in nature.

As a first preliminary result, we show that \emph{Markovian} policies are the wrong tool to deal with a farsighted agent, \emph{not} only because the preclude the principal from obtaining their optimal utility, but also because they are computationally intractable.
This is perhaps unexpected, since in most of the MDP settings of interest \emph{Markovian} policies are usually amenable to polynomial-time computations.
In particular, we show that the problem of computing a (suboptimal in general) utility-maximizing \emph{Markovian} policy is $\mathsf{NP}$-hard, even if one only seeks for an approximate solution to the problem.
Nevertheless, in this paper we show that, surprisingly, an approximately-optimal history-dependent policy to commit to can be computed in polynomial time, despite such a task may seem in principle much more demanding than its relative one in the realm of \emph{Markovian} policies.

Dealing with history-dependent policies raises considerable computational challenges, as in general these do \emph{not} even admit an efficient (\emph{i.e.}, polynomially-sized) representation.
However, as a first step we show that it is possible to focus w.l.o.g.~on policies that admit a compact representation, which are called \emph{promise-form} policies.
Intuitively, such policies provide a tool to efficiently encode all the needed information contained in the history of interaction in a single \emph{promise} value, which represents how much utility the agent is expecting to get by following the principal's intended course of action in future.
We prove that such information is the only one needed, as there always exists a promise-form policy that attains the principal's optimal utility.
Let us remark that a similar idea has been exploited in related \emph{Bayesian} persuasion settings by~\citet{bernasconi2023persuading}, and, in this paper, we carry it over to principal-agent problems by making the required adaptations needed to make it work there.

The rest of the paper revolves around computing an approximately-optimal promise-form policy in polynomial time.
This is conceptually split into two main building blocks, described next.

The \textbf{first} building block is a efficient algorithm to efficiently compute a suitable promise-form policy.
Specifically, we design an algorithm capable of finding a promise-form policy achieving at least the optimal principal's cumulative expected, though such a policy is only approximately \emph{incentive compatible} (IC).
This means that the agent is only approximately incentivized to follow the actions that the principal is willing to incentivize.
Thus, such a policy cannot be effectively put in place.

The \textbf{second} building block is a polynomial-time algorithm to convert any approximately IC promise-form policy into an (exactly) IC one, by only incurring in a bounded loss in principal's cumulative expected utility.
This applied to the policy output by the first block allows achieves the desired result.

\subsection{Related works}
\paragraph{Principal-Agent MDPs}
Our setting is related to the ones analyzed by~\cite{wu2024contractual} and~\cite{ivanov2024principal}.
In particular, their model consists of an agent playing an MDP while a principal commits to a contract at each step.
Unlike us, they look for a \emph{subgame perfect equilibrium}.
This equilibrium concept requires that the principal commits to optimal contract at each state and time step, thus disallowing \emph{threats}.
As a result, these works can focus on Markovian policies.

On the other hand, we focus on a framework where the principal can commit to a policy beforehand.
By means of this commitment, the principal can achieve a larger expected utility by threatening to reduce payments if the agent doesn't follow some trajectories.
For further details, see the proof of Proposition~\ref{thm:history_necessary} and the discussion of~\cite{ivanov2024principal}.

Finally, other works analyze similar problems in the context of policy teaching~\cite{zhang2008value,zhang2009policy}, where  the principal provides incentives to follow some trajectories, and environment design~\cite{yu2022environment,ben2024principal} (\emph{i.e.}, modifying the agent's reward functions).
\paragraph{Robust Contracts} 
Our work extends a known result in the context of ``one-shot'' principal-agent problems, namely that an almost incentive compatible contract can be converted into a incentive compatible one.
More formally, \cite{dutting2021complexity} show that given a contract $p$ that $\epsilon$-incentivizes an action $a$ and provides the principal's utility $u$ when $a$ is played, one can build another contract $p'$ that perfectly incentivizes some action $a'$ (possibile different from $a$) while achieving a principal's expected utility slightly smaller than $u$.
This result has also been extend to more complex principal-agent problems involving, \emph{e.g.}, menu of contracts~\cite{castiglioni2023designing} or repeated interactions~\cite{bernasconi2024regret}. 

\paragraph{Promises}
We devise a subclass of efficiently representable policies based on the concept of promises introduced by~\cite{bernasconi2024persuading} in the context of Bayesian persuasion in an MDP.
Intuitively, their policies promise a minimum expected utility to the agent at each step.
The main challenge in adapting their idea to contract design is that  in our framework the principal \emph{cannot} observe the agent's action.
This prompted us to use a stricter notion of promises and to change the relative incentive-compatibility constraints in a more complex formulation.
Furthermore, we observe that, being defined as lower bounds, the promises of~\cite{bernasconi2024persuading} can be arbitrarily smaller than the expected agent's utility.
Instead, our promises provide better insights on the inner working of our policies, as they represent the exact expected utility that is due to the agent.
%
 
\section{Preliminaries}
\label{sec:preliminaries}

\subsection{Model}

We study principal-agent problems where a \emph{farsighted} agent takes actions in a \emph{time-inhomogeneous finite-horizon} MDP, with such actions being \emph{hidden} to the principal, who can only observe their outcomes.
Formally, a problem instance is defined as a tuple:
\begin{equation*}
	\left( \sset, \A, \hset, \{P_h\}_{h \in \hset}, \mu, \{r_h\}_{h \in \hset}, \{c_h\}_{h \in \hset} \right), \,\, \textnormal{where:}
\end{equation*}
\begin{itemize}
	\item $\sset$ is a finite set of states.
	\item $\A$ is a finite set of agent's actions.
	\item $\hset = \{1,\dots,H\}$ is a set of time steps, with $H$ being the time horizon.
	\item $P_h : \sset \times \A \rightarrow \Delta(\sset)$ is a transition function, with $P_h(s' | s,a)$ denoting the probability of reaching state $s' \in \sset$ at step $h+1$ when taking action $a \in \A$ in state $s \in \sset$ at step $h \in \hset$.\footnote{In this paper, we denote by $\Delta(\mathcal{Y})$ the set of all the probability distributions defined over a set $\mathcal{Y}$.} 
	\item $\mu \in \Delta(\sset)$ is an initial distribution over states, with $\mu(s)$ denoting the probability of $s \in \sset$.
	\item $r_h : \sset \times \sset \rightarrow [0,1]$ is a principal's reward function, with $r_h(s,s')$ being the reward obtained when transitioning from state $s \in \sset$ at step $h \in \hset$ to state $s' \in \sset$ at step $h+1$.
	\item $c_h : \sset \times \A \rightarrow [0,1]$ is an agent's cost function, with $c_h(s,a)$ being the cost suffered by the agent when taking action $a \in \A$ in state $s \in \sset$ at step $h \in \hset$.\footnote{In this paper, for ease of presentation, we assume w.l.o.g.~that there exists an action $a \in \A$ whose cost is $c_h(s,a)=0$ for every step $h \in \mathcal{H}$ and state $s \in \sset$.}
\end{itemize}

At each step of the MDP, the principal has the ability to commit to a contract and an associated action recommendation for the agent.
A \emph{contract} is a payment scheme that prescribes a monetary transfer from the principal to the agent for every possible future state.
Formally, this is defined as a function $p : \sset \to [0,B]$, with $p(s)$ encoding how much the principal commits to pay the agent in the event that state $s \in \sset$ is reached at the next step of the MDP.\footnote{As it is customary in the literature, in this paper we assume that the agent has \emph{limited liability}, \emph{i.e.}, payments are only from the principal to the agent and \emph{not viceversa} (see, \emph{e.g.},~\citep{dutting2019simple}), and that payments are bounded above by some ``sufficiently-big'' constant $B \in \mathbb{R}_{>0}$ (see, \emph{e.g.},~\citep{zhu2023sample}).}
In the following, we denote by $\mathcal{P}$ the set of all the possible contracts, while, for ease of presentation, we let $\X \coloneqq \mathcal{P} \times \A$ be the set of all the possible pairs of contracts and action recommendations the principal can choose from.

In the most general case, the principal can commit beforehand to a \emph{non-stationary} and \emph{non-Markovian} policy (henceforth called \emph{history-dependent} policy for short).
This defines a randomization over contract-recommendation pairs for every possible \emph{history} realized while interacting with the agent in the MDP.
Formally, a history-dependent policy is a function $\pi : \mathcal{T} \to \Delta(\mathcal{X})$, where $\mathcal{T}$ is the set of all the possible histories (up to any step of the MDP).
Specifically, $\mathcal{T} \coloneqq \bigcup_{h \in \hset} \mathcal{T}_h$, where $\mathcal{T}_h$ is the set of all histories up to step $h \in \hset$, defined as follows:
\begin{equation*}
	\mathcal{T}_h \coloneqq \{\tau \mid \tau = (s_1,p_1,a_1,\dots,s_{h-1},p_{h-1},a_{h-1},s_h) : s_i \in \sset, p_i \in \mathcal{P}, a_i \in \A \},
\end{equation*}
with $s_i$ denoting the state reached at the $i$-th step and $(p_i,a_i)$ being the contract-recommendation pair selected by the principal at that step.
Notice that histories do \emph{not} include the actions that are actuality played by the agent, but only those recommended by the principal, as the former are \emph{not} observable.

		\begin{algorithm}[!htp]
			\caption{\texttt{Interaction}}\label{alg:nteraction}
			\begin{algorithmic}[1]
				\State Principal commits to $\pi: \mathcal{T} \to \Delta(\mathcal{X})$
				\State $s_1 \sim \mu$, $\tau \gets (s_1)$
				\ForAll{$h = 1, \ldots, H$}
				\State Commitment $(p_h,a_h) \sim \pi(\tau)$
				\State Agent plays an action $\widehat{a}_h \in \A$ maximizing future expected utility
				%
				%
				\State $s_{h+1} \sim P_h(s_h,\widehat{a}_h)$
				\State Principal gets $r^\text{P}_h(s_h,p_h,s_{h+1}) $
				\State Agent gets $r^\text{A}_h(s_h,p_h,\widehat{a}_h,s_{h+1}) $
				\State $\tau \gets \tau \oplus (p_h,a_h, s_{h+1})$
				\EndFor
			\end{algorithmic}
		\end{algorithm}

The interaction between the principal and the agent during a given step $h \in \hset$ when the history up to $h$ is $\tau = (s_1,p_1,a_1,\dots,s_{h-1},p_{h-1},a_{h-1},s_h)\in \mathcal{T}_h$ goes on as follows (see also Algorithm~\ref{alg:nteraction}):
\begin{enumerate}
	%
	%
	%
	\item The principal publicly commits to a contract-recommendation pair $(p_h, a_h) \sim \pi(\tau)$.
	%
	%
	\item The agent plays an action $\widehat{a}_h \in \A$ (possibly different from recommendation $a_h$) maximizing their (future) cumulative expected utility.
	\item Next state $s_{h+1} \sim P_h(s_h,\widehat{a}_h)$ is sampled.
	%
	%
	\item The principal and the agent collect their utilities $r^\text{P}_h(s_h,p_h,s_{h+1}) \coloneqq r_h(s_h,s_{h+1}) - p_h(s_{h+1})$ and $r^\text{A}_h(s_h,p_h,\widehat{a}_h,s_{h+1}) \coloneqq p_h(s_{h+1}) -c_h(s_h,\widehat{a}_h)$, respectively.
	%
\end{enumerate}


\subsection{Value functions}

Next, we introduce value functions for the principal and the agent.
%
The principal’s value function $\vprincipal_h : \mathcal{T}_h \to \mathbb{R}$ encodes the principal's cumulative expected utility from step $h \in \hset$ onwards under policy $\pi: \mathcal{T} \to \Delta(\mathcal{X})$, when assuming that the \emph{agent sticks to recommendations}.
Formally, for every history $\tau = (s_1,p_1,a_1,\dots,s_{h-1},p_{h-1},a_{h-1},s_h) \in \mathcal{T}_h$, the following holds:\footnote{In this paper, we always work with policies $\pi : \mathcal{T} \to \Delta(\X)$ that specify finitely-supported distributions. Thus, we can write expectations with respect to $\pi$ as sums over the pairs $(p,a) \in \mathcal{X}_\pi \coloneqq \bigcup_{\tau \in \mathcal{T}} \text{supp}(\pi(\tau))$, where $\text{supp}(\pi(\tau))$ denotes the (finite) support of the probability distribution $\pi(\tau) $}
\begin{equation*}
	\vprincipal_h (\tau) \coloneqq \sum_{(p,a) \in \mathcal{X}_\pi} \pi(p,a | \tau)  \qprincipal_h(\tau,p,a),
\end{equation*}
%
where $\qprincipal_h: \mathcal{T}_h \times \X_\pi \to \mathbb{R}$ is such that
\begin{equation*}
	\qprincipal_h(\tau,p,a) \coloneqq \sum_{s' \in \sset}P_h(s'|s_h,a) \left( r^\text{P}_h(s_h,p,s') + \vprincipal_{h+1}(\tau \oplus (p,a,s')) \right).
\end{equation*}

Similarly, we introduce the agent's value function $\vagent_h : \mathcal{T}_h \to \mathbb{R}$, encoding the agent's cumulative expected utility by following recommendations from $h \in \hset$ onwards under policy $\pi: \mathcal{T} \to \Delta(\mathcal{X})$.
%
For every history $\tau = (s_1,p_1,a_1,\dots,s_{h-1},p_{h-1},a_{h-1},s_h) \in \mathcal{T}_h$, it holds:
\begin{equation*}
	\vagent_h (\tau) \coloneqq \sum_{(p,a) \in \mathcal{X}_\pi} \pi(p,a | \tau) \qagent_h(\tau,p,a),
\end{equation*}
where $\qagent_h: \mathcal{T}_h \times \X_\pi \to \mathbb{R}$ is such that
\begin{equation*}
	\qagent_h(\tau,p,a) \coloneqq \sum_{s' \in \sset} P_h(s'|s_h,a) \left(r^\text{A}_h(s_h,p,a,s') + \vagent_{h+1}(\tau \oplus (p,a,s')) \right).
\end{equation*}

We also introduce the agent's deviation-value function $\qagentdev_h : \mathcal{T}_h \times \mathcal{X}_\pi \rightarrow \mathbb{R}$ to encode the agent's cumulative expected utility from step $h \in \hset$ onwards under policy $\pi: \mathcal{T} \to \Delta(\mathcal{X})$, when the \emph{agent deviates from recommended actions}.
%
%
For every $\tau = (s_1,p_1,a_1,\dots,s_{h-1},p_{h-1},a_{h-1},s_h)\in \mathcal{T}_h$, contract $p \in \mathcal{P}$, and recommendation $a \in \mathcal{A}$, this is formally defined as:
%
%
\begin{equation*}
	\qagentdev_h(\tau,p,a) \coloneqq \max_{\widehat{a} \in A} \sum_{s' \in \sset} P_h(s'|s_h,\widehat{a}) \left(r^\text{A}_h(s_h,p,\widehat{a},s') +\vagentdev_{h+1}(\tau \oplus (p,a,s')) \right),
\end{equation*}
where $\vagentdev_h: \mathcal{T}_h \to \mathbb{R}$ is defined as follows for every $h \in \hset$ and $\tau \in \mathcal{T}_h$:
\begin{equation*}
\vagentdev_h(\tau) \coloneqq \sum_{(p,a) \in \mathcal{X}_\pi} \pi(p,a | \tau) \qagentdev_h(\tau,p,a).
\end{equation*}
Intuitively, $\qagentdev_h(\tau,p,a)$ is the maximum cumulative expected utility that the agent can achieve by deviating from recommendations, once the principal has committed to the contract-recommendation pair $(p,a)$ at step $h$ given history $\tau$.
Notice that the cumulative expected utility achieved by deviating from step $h+1$ onwards is computed by using $\tau \oplus (p,a,s')$, \emph{i.e.}, accounting for the recommendation $a$ rather than the action $\widehat{a}$ actually played by the agent, as the principal cannot observe the latter.
%
%


Finally, given a policy $\pi : \mathcal{T} \to \Delta(\mathcal{X})$, we let $\vprincipal \coloneqq \sum_{s \in \sset} \mu(s) \vprincipal_1(s)$ be the principal's cumulative expected utility from the initial step onwards, when the agent sticks to recommendations.

\subsection{Optimal policies}

Next, we formally introduce the optimization problem faced by the principal.
First, we observe that, by well-known revelation-principle-style arguments (see, \emph{e.g.},~\citep{gan2022optimal}), it is possible to focus w.l.o.g.~on policies that are \emph{incentive compatible} and \emph{direct}.
The first property intuitively characterizes the policies under which the agent is ``correctly'' incentivized to follow recommendations, while the second one identifies policies that specify at most one contract for every possible action recommendation.
These two properties are formally defined in the following.
%
%
%
\begin{definition}[$\epsilon$-IC policies]
	Let $\epsilon \geq 0$. A history-dependent policy $\pi : \mathcal{T} \to \Delta(\mathcal{X})$ is said to be \emph{$\epsilon$-incentive compatible ($\epsilon$-IC)} if, for every step $h \in \mathcal{H}$ and history $\tau \in \mathcal{T}_h$, the following condition holds for every contract-recommendation pair $(p,a) \in \mathcal{X}$ such that $\pi(p,a | \tau)>0$:
	%
	%
	\begin{equation*}
		\qagent_h(\tau,p,a) \ge  \qagentdev_h(\tau,p,a) - \epsilon.
	\end{equation*}
	%
	Moreover, we say that $\pi$ is \emph{incentive compatible (IC)} if the condition holds for $\epsilon=0$.
\end{definition}
%
Intuitively, a policy is $\epsilon$-IC if the agent always loses at most $\epsilon$ by playing the action recommended by the principal rather than one maximizing their cumulative expected utility.
In this paper, we denote by $\Pi$ the set of all the IC policies.\footnote{In the rest of this paper, we assume w.l.o.g.~that, under IC policies, the agent always sticks to the principal's recommendations $a_h$ when there are multiple actions $\widehat{a}_h \in \A$ maximizing their cumulative expected utility.}
Moreover, let us also notice that, if a policy $\pi: \mathcal{T} \to \Delta(\mathcal{X})$ is IC, then it is the case that $\vagent_h(\tau) = \vagentdev_h(\tau)$ for every step $h \in \mathcal{H}$ and history $\tau \in \mathcal{T}_h$.
%


\begin{definition}[Direct policies]
	A policy $\pi : \mathcal{T} \to \Delta(\mathcal{X})$ is said to be \emph{direct} if, for every history $\tau \in \mathcal{T}$ and action $a \in \A$, there exists at most one contract $p \in \mathcal{P}$ such that $\pi(p,a | \tau) > 0$.
\end{definition}

Thus, the principal's goal can be formulated as the problem of finding an \emph{optimal} IC policy $\pi^\star \in \Pi$, which is one that achieves value $\text{OPT}$.
This is formally defined as $\text{OPT} \coloneqq \max_{\pi \in \Pi} \vprincipal$.

\section{Why do we need history-dependent policies?}\label{sec:history_necessary}

As a preliminary step, we discuss why history-dependent policies are necessary.

First, we show that history-dependent policies allow the principal to achieve a cumulative expected utility strictly larger than what is attainable by means of (non-stationary) Markovian policies.
%
%
Intuitively, whenever there are multiple trajectories reaching a particular state, a history-dependent policy can induce the agent to follow one providing high utility to the principal.
This is possible by ``threatening'' to do \emph{not} pay the agent in the event that they deviate from the desired trajectory.
In contrast, a Markovian policy cannot implement such a ``threat'', since it is \emph{not} able to distinguish how a state has been reached.
In Appendix~\ref{sec:app_markov}, we provide an example in which the principal can benefit from implementing the ``threats''  described above, thus showing that the following proposition holds.
\begin{restatable}{proposition}{historyNecessary}\label{thm:history_necessary}
	There exists a problem instance in which a history-dependent policy provides a principal's cumulative expected utility larger than that of any Markovian policy.
\end{restatable}
Furthermore, we also show that computing an approximately-optimal \emph{non-stationary Markovian} policy is $\mathsf{NP}$-hard, while, as we show in the rest of this paper, it is possible to implement an approximately-optimal history-dependent policy by means of a polynomial-time algorithm.
%
%
\begin{restatable}{theorem}{Hard}
	There exists a constant $\delta>0$ such that computing a non-stationary Markovian policy with value at least $(1-\delta)\textnormal{OPT}_\textnormal{Mkv}$ is $\mathsf{NP}$-hard, where $\textnormal{OPT}_\textnormal{Mkv}$ is the principal's cumulative expected utility of an optimal non-stationary Markovian policy.
\end{restatable}

\section{Promise-form policies}\label{sec:promise_form}
As observed in Section~\ref{sec:history_necessary}, Markovian policies are both suboptimal and computationally intractable.
Therefore, in this work, we focus on the problem of computing a history-dependent policy.
However, history-dependent policies pose a computational challenge, as representing them requires an exponential number of bits with respect to the time horizon $H$.
To address this issue, we employ a class of policies—called \emph{promise-form} policies—that can be efficiently represented and are guaranteed to perform as well as history-dependent ones.
In this section, we first formally define a \emph{promise-form} policy, and we show that there always exists a \emph{promise-form} policy that is both IC and provides the principal's with a cumulative expected utility equal to $\textnormal{OPT}$.\footnote{A similar class of policies has been previously adopted in Bayesian persuasion settings~\citep{bernasconi2024persuading}, providing effective results when the agent is farsighted.}
\subsection{Promise-form policies}
A promise-form policy is defined as a tuple $\sigma \coloneqq \{(I_h,J_h,\varphi_h,g_h)\}_{h \in \mathcal{H}}$, where:
\begin{itemize}
	\item $I_h : \mathcal{S} \rightarrow 2^{[0,H]}$ is a function specifying, for every state $s \in \mathcal{S}$, a finite set $I(s) \subseteq [0,H]$ of promises for step $h \in \mathcal{H}$.
	We assume that $|I_1(s)|=1$ for every $s \in \mathcal{S}$ and we denote with $i(s)$ the single element belonging to $|I_1(s)|$.
	For ease of notation, we set $I_{H+1}(s) \coloneqq \{0\}$ for every $s \in \mathcal{S}$, and we let $\mathcal{I} \coloneqq \bigcup_{h \in \mathcal{H}} \bigcup_{s \in \mathcal{S}} I_h(s)$.
	%
	\item $J_h:\sset \times \mathcal{I} \times \mathcal{A} \rightarrow \mathcal{B}$ is a function defining, for every state $s \in \mathcal{S}$, promise $\iota \in I_h(s)$, and action $a \in \mathcal{A}$, a finite set $J_h(s,\iota,a) \subseteq \mathcal{P}$ of contracts that the principal can commit to when recommending action $a \in \mathcal{A} $ in state $s \in \mathcal{S}$ with promise $\iota \in \mathcal{I}$, at step $h  \in \mathcal{H}$.
	Furthermore, we let $\mathcal{J} \coloneqq \bigcup_{h \in \mathcal{H}} \bigcup_{s \in \mathcal{S}} \bigcup_{\iota \in I_h(s)} \bigcup_{a \in \mathcal{A}} J_h(s,\iota,a)$.
	%
	\item $\varphi_h : \mathcal{S} \times \mathcal{I} \rightarrow \Delta(\mathcal{J} \times \mathcal{A})$ 
	encodes the policy at step $h \in \mathcal{H}$, with $\varphi_h(p,a | s, \iota)$ denoting the probability with which the principal commits to contract $p \in J(s,\iota,a)$ and recommends action $a \in \mathcal{A}$, when in a state $s \in \mathcal{S}$ with promise $\iota \in \mathcal{I}_h(s)$.
	%
	\item $g_h : \mathcal{S} \times \mathcal{I} \times \mathcal{J} \times \mathcal{A} \times \mathcal{S} \rightarrow \mathcal{I}$ is a promise function for step $h \in \mathcal{H}$, with $g_h(s,\iota,p,a,s') \in \mathcal{I}_{h+1}(s')$ being the promise for step $h+1$ when, at step $h \in \mathcal{H}$: the state is $s \in \mathcal{S}$, the promise is $\iota \in \mathcal{I}_h(s)$, the principal recommends action $a \in \mathcal{A}$, they commit to contract $p \in \mathcal{J}_h(s,\iota,a)$, and the environment transitions to state $s' \in \mathcal{S}$.
\end{itemize}
In a \emph{promise-form} policy, when reaching a state $s \in \mathcal{S}$ at step $h \in \mathcal{H}$, the principal promises a value $\iota \in I_h(s)$ to the agent that represents how much the agent will gain by following the principal’s recommendations. 
Furthermore, at each time step $h \in \mathcal{H}$, we notice that the principal’s contract-recommendations depend only on the current promise and the current state through the function $\varphi_h$.
%
%
However, both the cardinality of the set of promises $\mathcal{I}$ and the cardinality of the set of recommended contracts $\mathcal{J}$ may, in principle, be particularly large, thus potentially requiring a large number of bits to represent a promise-form policy.
In the rest of the paper, we show how to sidestep this issues by working with ``sufficiently small'' sets of promises and recommended contracts.
%

%
%

\subsection{Implementing a promise-form policy}
We now discuss how the principal can implement a promise-form policy $\sigma$.
Specifically, we show how the principal, given a promise-form policy $\sigma$, can provide the agent with a contract-recommendation pair at each step $h \in \mathcal{H}$ based on the current history $\tau_h$.
To do so, we introduce Algorithm~\ref{alg:implement}, which takes as input the promise-form policy $\sigma$ that the principal has committed to, the time step $h \in \mathcal{H}$, and the history up to that time step $\tau_h$.
Then, in order to determine the current promise $\iota \in I_h(s_h)$, Algorithm~\ref{alg:implement} recursively composes the functions $g_h$ for the preceding steps $h'<h$ using the components of history $\tau_h$ it receives as input (See Line~\ref{line:implement_compute_next_iota}).
Once Algorithm~\ref{alg:implement} has computed the current promise $\iota \in I_h(s_h)$, it returns the contract-recommendation pair sampled from the function $\varphi_h$ evaluated in the current promise $\iota$ and in the current states $s_h$.
%
%
%
\begin{algorithm}[!htp]
	\caption{\texttt{Implement promise-form}}\label{alg:implement}
	\begin{algorithmic}[1]
		\Require $h \in \mathcal{H}$, \Statex $\tau = (s_1,p_1,a_1, \dots, s_{h-1},p_{h-1},a_{h-1},s_h)$, \Statex $\sigma = \{(I_h,J_h,\varphi_h,g_h)\}_{h \in \mathcal{H}}$.
		\State $\iota \gets i(s_1)$ \Comment{$I_1(s_1)$ contains only $i(s_1)$}
		\ForAll{$h'=1,\dots,h-1$}
		\State $\iota \gets g_{h'}(s_{h'},\iota,p_{h'},a_{h'},s_{h'+1})$ \label{line:implement_compute_next_iota}
		\EndFor
		\State $(p,a) \sim \varphi_h(s_h,\iota)$
		\State \textbf{Return} $(p,a)$
	\end{algorithmic}
\end{algorithm}

In the following, given a promise-form policy $\sigma$, we denote with $\pi^\sigma$ the policy induced by $\sigma$ by following the procedure described above.
Formally, for each step $h \in \mathcal{H}$ and history $\tau_h$, we let $\pi^\sigma(\cdot|\tau) \coloneqq \varphi(\cdot|s_h, \iota^\sigma_\tau )$, where $\iota^\sigma_\tau$ is the promise computed according to Algorithm~\ref{alg:implement}, given in input the history $\tau$ and the promise-form policy $\sigma$.

Furthermore, given a promise-form policy $\sigma$, the value and action-value functions of its corresponding history-dependent policy $\pi^\sigma$ can be computed from the components that define the promise-form policy itself.
Indeed, when the principal commits to a promise-form policy, their behavior does not depend on the full history, but rather on the current state and promise.
Thus, we can define some equivalent functions that take in input, instead of a history $\tau \in \mathcal{T}$, a state $s \in \sset$ and a promise $\iota \in I_h(s)$, while the probability of prescribing a pair $(p,a) \in \mathcal{X}$ at step $h \in \mathcal{H}$ is computed as $\varphi_h(p,a|s,\iota)$.
The formal definition of these promise-form value functions are provided in Appendix~\ref{appendix:promise_form}, while their equivalence to the history-dependent value functions is stated in the following lemma:
\begin{restatable}{lemma}{promiseFunctionsEq}
	\label{lem:promise_functions_eq}
	Given a promise-form policy $\sigma =\{(I_h,J_h,\varphi_h,g_h)\}_{h \in \mathcal{H}}$ and the corresponding history-dependent policy $\pi$, for any step $h \in \mathcal{H}$, history $\tau$ ending in state $s$, action $a \in \mathcal{A}$ and contract $p \in \mathcal{B}$, the following holds:
	\begin{align*}
		\vprincipalpr_h (s,\iota^\sigma_\tau) &= \vprincipal_h (\tau), \quad
		\qprincipalpr_h(s,\iota^\sigma_\tau,p,a) = \qprincipal_h(\tau,p,a), \\
		\vagentpr_h (s,\iota^\sigma_\tau) &= \vagent_h (\tau), \quad
		\qagentpr_h(s,\iota^\sigma_\tau,p,a) = \qagent_h(\tau,p,a), \\
		\vagentdevpr_h (s,\iota^\sigma_\tau) &= \vagentdev_h (\tau), \quad
		\qagentdevpr_h(s,\iota^\sigma_\tau,p,a) = \qagentdev_h(\tau,p,a),
	\end{align*}
	where $\iota^\sigma_\tau$ is computed according to Algorithm~\ref{alg:implement}.
\end{restatable}

\subsection{Honesty}
We now introduce a subclass of promise-form policies in which the principal honestly (approximately) fulfils the promises made to the agent.
Formally, we have:
\begin{definition}[$\eta$-honesty]
	A promise-form policy $\sigma =\{(I_h,J_h,\varphi_h,g_h)\}_{h \in \mathcal{H}}$ is $\eta$-honest, with $\eta \ge 0$, if, for every step $h \in \mathcal{H}$, state $s \in \mathcal{S}$, and promise $\iota \in I_h(s)$, the following holds:
	
	\begin{equation*}
		\left| \sum_{(p,a) \in \mathcal{X}_{\sigma}} \varphi_h(p,a|s,\iota) \sum_{s' \in \sset} P_h(s'|s,a) \left(r^\textnormal{A}_h(s,p,a,s') +g_h(s,\iota,p,a,s') \right) -\iota \right| \le \eta.
	\end{equation*} 
\end{definition}
Intuitively, a $\eta$-honest policy ensures that at the last time step $H$, the difference between the promise and the actual agent's expected utility is at most $\eta>0$.
Thus, by employing a recursive argument, it is possible to show that if $\sigma$ is $\eta$-honest, then the difference between the actual agent's utility and the promised one increases by at most $\eta > 0$ at each previous time step, as stated in the following lemma.
\begin{restatable}{lemma}{honestyValue}\label{lem:honesty_value}
	Let $\sigma =\{(I_h,J_h,\varphi_h,g_h)\}_{h \in \mathcal{H}}$ be a $\eta$-honest promise-form policy.
	Then, for every step $h \in \mathcal{H}$ and state $s \in \mathcal{S}$ it holds that $\left|\vagentpr_h(s,\iota)-\iota \right| \le \eta(H-h+1)$ for every $\iota \in I_h(s)$.
\end{restatable}
%
%
We also introduce a crucial lemma showing that if an $\eta$-honest promise-form policy satisfies a suitable a set of ``local''constraints, then such a policy is also approximately IC.
%
\begin{restatable}{lemma}{localICConstr}
	\label{lem:local_ic_constr}
	Let $\sigma =\{(I_h,J_h,\varphi_h,g_h)\}_{h \in \mathcal{H}}$ be a $\eta$-honest promise-form policy such that for every step $h \in \mathcal{H}$, state $s \in \mathcal{S}$, promise $\iota \in I_h(s)$, actions $a,\widehat{a} \in \mathcal{A}$ and contract $p \in \mathcal{B}$ such that $\varphi_h(p,a|s,\iota)>0$, the following constraint holds:
	\begin{equation}
		\label{eq:local_ic_constr}
		\begin{split} 
			\sum_{s' \in \sset} &P_h(s'|s,a) \left(r^\textnormal{A}_h(s,p,a,s') +g_h(s,\iota,p,a,s') \right) \ge \\
			&\sum_{s' \in \sset} P_h(s'|s,\widehat{a}) \left(r^\textnormal{A}_h(s,p,\widehat{a},s') +g_h(s,\iota,p,a,s') \right).
		\end{split}	
	\end{equation}
	Then, $\sigma$ is $2\eta H^2$-IC. 
\end{restatable}
%
Thanks to Equation~\ref{eq:local_ic_constr}, we can require a policy to be approximately IC by only imposing a local constraint.
Indeed, for every state $s \in \sset$ and time step $h \in \mathcal{H}$, Equation~\ref{eq:local_ic_constr} depends only on the current promise $\iota$ and the components that define the promise-form policy $\sigma$.
%
Finally, given the previous results, we can prove that promise-form policies represent a subclass of history-dependent policies that is powerful enough to find the optimum.
\begin{restatable}{theorem}{promisingOptimal}
	\label{th:promising_optimal}
	There always exists a honest IC promising-form policy that achieves a principal's expected cumulative reward of $\textnormal{OPT}$.
\end{restatable}

\section{From $\epsilon$-IC policies to an approximate optimum}
\label{sec:epsilon_ic_to_ic}
Our goal is to compute an IC policy that achieves at least $\text{OPT}-\gamma$ for some given $\gamma$.
In order to do this, we show that an $\epsilon$-IC $\pi$ can be converted into an IC policy $\pi'$ with value $V^{\text{P},\pi'} \ge \vprincipal - (H+1)\sqrt{\epsilon}$.
This way, we can relax our problem, computing an $\epsilon$-IC policy with value at least OPT, and then convert it into an approximate optimum $\pi'$ with value $\text{OPT} - (H+1)\sqrt{\epsilon}$.

To show how to perform this conversion, we first need to introduce some additional notation.
As a first step, we introduce the set of histories $\mathcal{T}' \coloneqq \mathcal{T}_2 \cup \mathcal{T}_3 \cup \dots \cup \mathcal{T}_{H+1}$.
We recall that a history $\tau = (s_1,p_1,a_1,\dots,s_{h},p_{h},a_{h},s_{h+1}) \in \mathcal{T}_{h+1}$ contains $h$ transitions. 
Intuitively, the set $\mathcal{T}'$ include all the histories with at least a transition in them, including the histories in $\mathcal{T}_{H+1}$ containing the last transition of the MDP.

Given this set, we introduce the notion of \emph{agent's function}. 
\begin{definition}[Agent's function]
	An agent's function $\alpha : \mathcal{T}' \rightarrow \bigcup_{h \in \mathcal{H}} \mathcal{A}^{h}$ is a function mapping each history $\tau_h \in \mathcal{T}_h \subseteq \mathcal{T}'$ to a sequence of $h-1$ actions satisfying the following conditions:
	\begin{enumerate}
		\item For every pair of histories $\tau,\tau' \in \mathcal{T}'$ such that $\tau \oplus \tau' \in \mathcal{T}'$, it holds that $\alpha(\tau) \oplus \alpha(\tau') = \alpha(\tau \oplus \tau')$.
		\item Given any history $\tau = (s_1,p_1,a_1,\dots,s_{h},p_h,a_{h},s_{h+1}) \in \mathcal{T}'$, it holds that 
		\[\alpha(s_1,p_1,a_1,\dots,s_h,p_h,a_h,s_{h+1}) = \alpha(s_1,p_1,a_1,\dots,s_h,p_h,a_h,s)\] for every state $s \in \mathcal{S}$.
	\end{enumerate}
\end{definition}
Intuitively, an agent's function provides, given any history $\tau = (s_1,p_1,a_1,\dots,s_h,p_h,a_h,s_{h+1}) \in \mathcal{T}'$, the sequence of actions $\alpha(\tau_h) = (\widehat{a}_1,\widehat{a}_2,\dots,\widehat{a}_h)$ effectively played by a (possibly irrational) agent.
With an abuse of notation, we denote with $\alpha_k(\tau)$ the $k$-th element of the sequence.

We will employ different agent's functions to analyze how the cumulative rewards change depending on both the principal's policy and the strategy used by the agent to respond.
In particular, we are interest in cumulative rewards when the agent follows the recommendations.
The \emph{recommended agent's function} $\alpha$ is formally defined as:
\begin{equation*}
	\alpha(\tau) \coloneqq (a_1,a_2,\dots,a_h) \quad \forall \tau = (s_1,p_1,a_1,\dots,s_h,p_h,a_h,s_{h+1}) \in \mathcal{T}'.
\end{equation*}

Furthermore, given a policy $\pi$, we consider the \emph{incentivized agent's functions}, which correspond to an agent that always plays rationally.
Formally, an agent's function $\alpha$ is incentivized by $\pi$ if:\footnote{Given a history $\tau = (s_1,a_1,\dots,s_h,a_h,s_{h+1}) \in \mathcal{T}'$, for any $h' \in \{1,2,\dots,h+1\}$ we denote with $\tau(s_{h'})$ the history $\tau(s_{h'})=(s_1,a_1,\dots,s_{h'-1},a_{h'-1},s_{h'})$, that is the portion of $\tau$ up to state $s_{h'}$.}
\begin{equation*}
	\alpha_{h'}(\tau) \in \argmax_{\widehat{a} \in \mathcal{A}} \sum_{s' \in \sset}P_h(s'|s_{h'},\widehat{a}) \left( p_{h'}(s') - c_{h'}(s_{h'},\widehat{a}) +\widehat{V}^{\text{A},\pi}_{h+1}(\tau(s_{h'}) \oplus (p_{h'},a_{h'},s')) \right).
\end{equation*}
We observe that there may be different incentivized policies, depending on how ties are broken.

Finally, we define the value of a policy $\pi$ and agent's function $\alpha$ as:
\begin{equation*}
	V^{\text{P},\pi,\alpha} \coloneqq \hspace{-1.5cm} \sum_{\substack{\tau \in \mathcal{T}'_{\pi,\alpha}, \\ \tau \coloneqq (s_1,p_1,a_1,\dots,s_h,p_h,a_h,s_{h+1})}} \hspace{-1.5cm} \mathbb{P}(\tau|\pi,\alpha) \left(r_h(s_{h},s_{h+1}) - p_h(s_{h+1})\right),
\end{equation*}
where the probability of history $\tau = (s_1,p_1,a_1,\dots,s_h,p_h,a_h,s_{h+1}) \in \mathcal{T}'$ is:
\begin{equation*}
	\mathbb{P}(\tau|\pi,\alpha) = \mu(s_1)\prod_{h'=1}^{h}  \pi(p_{h'},a_{h'}|\tau(s_{h'})) P_{h'}(s_{h'+1}|s_{h'},\alpha_{h'}(\tau)),
\end{equation*}
while $\mathcal{T}'_{\pi,\alpha}$ is the set of all histories $\tau$ in $\mathcal{T}'$ such that $\mathbb{P}(\tau|\pi,\alpha)>0$.
Similarly, we define the value of the policy \emph{for the agent} as:
\begin{equation*}
	V^{\text{A},\pi,\alpha} \coloneqq \hspace{-1.5cm} \sum_{\substack{\tau \in \mathcal{T}'_{\pi,\alpha}, \\ \tau \coloneqq (s_1,p_1,a_1,\dots,s_h,p_h,a_h,s_{h+1})}} \hspace{-1.5cm} \mathbb{P}(\tau|\pi,\alpha) \left(p_h(s_{h+1}) - c_h(s_{h},\alpha_h(\tau)) \right).
\end{equation*}

To build an IC policy $\pi'$ given an $\epsilon$-IC policy $\pi$, we extend the technique proposed for the single-state case by~\cite{dutting2021complexity}.
In particular, we change every contract $p$ proposed by $\pi$ at any time step $h$ and state $s$ to a new contract $p'$ defined as:
\begin{equation*}
	p'(s') \coloneqq (1-\sqrt{\epsilon})p(s') + \sqrt{\epsilon}r_h(s') \quad \forall s' \in \sset.
\end{equation*}
The new policy $\pi'$ achieves a ``good'' value when the farsighted agent best responds to it, as stated by the following lemma:
\begin{restatable}{lemma}{fromepictopseudoic}
	\label{lem:from_ep_ic_to_pseudo_ic}
	Given an $\epsilon$-IC direct policy $\pi$, let $\pi'$ be the policy defined changing the contracts proposed by $pi$, formally $\pi'(p',a|\tau) \coloneqq \pi(p,a|\tau)$ for $p'(s') \coloneqq (1-\sqrt{\epsilon})p(s') + \sqrt{\epsilon}r_h(s')$, with $\tau \in \mathcal{T}_h$ and $(p,a) \in \mathcal{X}_\pi$.
	Then it holds that $V^{\textnormal{P},\pi',\alpha'} \ge V^{\textnormal{P},\pi} -(H+1)\sqrt{\epsilon}$ for any agent's function $\alpha'$ incentivized by $\pi'$.
\end{restatable}

The result above holds for any possible history-dependent policies.
In particular, it holds even when we consider promise-form policies, which admit polynomial-sized representations.
Furthermore, we observe that the result holds for any incentivized agent's function, regardless of how the agent breaks ties.
Thus, while the optimal IC policy $\pi^\star$ requires the agent to beak ties accordingly to the recommendations, our approximately optimal policy $\pi'$ is robust with respect to how the agent breaks such ties, as long as they are farsighted and play rationally.

\begin{algorithm}[!htp]
	\caption{\texttt{From $\epsilon$-IC to IC}}\label{alg:from_ep_ic_to_ic}
	\begin{algorithmic}[1]
		\Require direct $\epsilon$-IC $\sigma = \{(I_h,J_h,\varphi_h,g_h)\}_{h \in \mathcal{H}}$ \Statex that achieves $\vprincipalpr \ge OPT$.
		\State $\sigma^1 \gets \texttt{Change-Contracts}(\sigma)$
		\State $\sigma^2 \gets \texttt{Realign-Actions}(\sigma^1)$
		\State $\sigma^3 \gets \texttt{Realign-Promises}(\sigma^2)$
		\State \textbf{Return} $\sigma^3$
	\end{algorithmic}
\end{algorithm}


Let $\sigma$ be an $\epsilon$-IC promise-form policy with value at least OPT.
Algorithm~\ref{alg:from_ep_ic_to_ic} computes a honest and IC promise-form policy $\sigma^3$ with value at least $\text{OPT} - (H+1)\sqrt{\epsilon}$.
The algorithm is divided in three subprocedures, whose details are provided in Appendix~\ref{appendix:subprocedures_ep_ic_to_ic}.
The first procedure, $\texttt{Change-Contracts}$, computes a policy $\sigma^1$ by applying Lemma~\ref{lem:from_ep_ic_to_pseudo_ic}.
However, the policy $\sigma^1$ is not IC, as it still recommends the same actions of $\sigma$, albeit it changed the contracts.
To address this issue, the procedure $\texttt{Realign-Actions}$ computes an IC policy $\sigma^2$ by changing the recommended actions to the actual best responses.
For completeness, the third procedure, $\texttt{Realign-Promises}$, computes a policy $\sigma^3$ that is both IC and honest, by recomputing the promises that had lost their semantic value when the algorithm changed the contracts.
The properties of Algorithm~\ref{alg:from_ep_ic_to_ic} are summarized in the following lemma, where the size of $\sigma$ is the size of the sets of promises and contracts $\mathcal{I}$ and $\mathcal{J}$.
\begin{restatable}{theorem}{fromepictoic}
	\label{th:from_ep_ic_to_ic}
	Given a direct and $\epsilon$-IC promise-form policy $\sigma = \{(I_h,J_h,\varphi_h,g_h)\}_{h \in \mathcal{H}}$ that achieves at least $\vprincipalpr \ge OPT$, Algorithm~\ref{alg:from_ep_ic_to_ic} computes an IC and honest promise-form policy $\sigma^3$ in time polynomial in the size of $\sigma$ and the instance. 
	Furthermore, the value of $\sigma^3$ is at least $V^{\textnormal{P},\sigma^3} \ge \textnormal{OPT} - (H+1)\sqrt{\epsilon}$,
	while its size is such that $|\mathcal{I}^3| \le |\mathcal{I}|$ and $|\mathcal{J}^3| \le |\mathcal{J}|$.
\end{restatable}

\section{Approximation Policy}\label{sec:approx_policy}


We provide an algorithm to efficiently compute a promise-form policy $\sigma$ that is $\epsilon$-IC for a given $\epsilon$ and provide the principal with an expected cumulative utility of at least OPT.
Intuitively, the algorithm works by considering a discrete set of promises $\mathcal{I} \subseteq \mathcal{D}_\delta \coloneqq \{k\delta \mid 0 \le k \le \lfloor \nicefrac{HB}{\delta} \rfloor\}$, where $\delta$ is a discretization step defined depending on $\epsilon$.
Furthermore, there always exists $\pi^*$ that is both \emph{direct} and IC.
This fact prompts us to compute a direct promise-for policy $\sigma$, \emph{i.e.}, such that $|J_h(s,\iota,a)| \le 1$ for every step $h \in \mathcal{H}$, state $s \in \sset$, promise $\iota \in \mathcal{D}_\delta$, and action $a \in \mathcal{A}$.

		\begin{algorithm}[!htp]
			\caption{\texttt{Approximation policy}}\label{alg:approx_policy}
			\begin{algorithmic}[1]
				\Require $\delta \in (0,1)$.
				\State $M^\delta_{H+1}(s,0) \gets 0 \quad \forall s \in \sset$
				\State $M^\delta_{H+1}(s,\iota) \gets -\infty \quad \forall s \in \sset, \iota \in \mathcal{D}_\delta \setminus \{0\}$
				\ForAll{$h \in \mathcal{H}$}
				\ForAll{$s \in \mathcal{S}$}
				\State $I_h(s) \gets 0$
				\ForAll{$\iota \in \mathcal{D}_\delta$}
				\State $(\{\alpha_a\}_{a \in \A},\{p^a\}_{a \in \A},z,v) \gets O^{h,s}_{\iota,\delta}(M^\delta_{h+1})$ \label{line:use_oracle}
				\State $M^\delta_{h}(s,\iota) \gets v$
				\ForAll{$a \in \A : \alpha_a > 0$}
				\State $\varphi_h(p^a,a|s,\iota) \gets \alpha_a$
				\State $J_h(s,\iota,p^a,a) \gets p^a$
				\State $g_h(s,\iota,p^a,a,s') \gets z(a,s') \quad \forall s' \in \sset$.
				\EndFor
				\If{$v>-\infty$}
				\State $I_h(s) \gets I_h(s) \cup \{\iota\}$
				\EndIf
				\EndFor
				\EndFor
				\EndFor
				\ForAll{$s \in \sset$}
				\State $I_1(s) \gets \argmax_{\iota \in I_1(s)} M^\delta_1(s,\iota)$ \label{line:initial_promise}
				\EndFor
				\State \textbf{Return} $\{(I_h,\varphi_h,g_h)\}_{h \in \mathcal{H}}$
			\end{algorithmic}
		\end{algorithm}

Algorithm~\ref{alg:approx_policy} keeps a table for every step $h \in \mathcal{H}$ represented as a function $M^\delta_h : \sset \times \mathcal{D}_\delta \rightarrow \mathbb{R} \cup \{-\infty\}$.
Intuitively, $M^\delta_h(s,\iota)$ is an approximation of the maximum cumulative reward that the principal can achieve from step $h$, starting from state $s \in \sset$ with promise $\iota \in \mathcal{D}_\delta$.
$M^\delta_h(s,\iota)$ is set to $-\infty$ when it is not possible to promise $\iota$ while being ``honest enough'' to be $\epsilon$-IC.

At high level, Algorithm~\ref{alg:approx_policy} works by employing dynamic programming, filling the tables $M^{\delta}_h$ from step $h=H$ to $h=1$.
At any step $h \in \mathcal{H}$, for every state $s \in \sset$ and promise $\iota \in \mathcal{D}_\delta$, it imposes the policy $\sigma$ to respect the promise $\iota$ up to an error of $\lambda \ge \nicefrac{\delta}{2}$.
Hence, we can work with the discrete set of promises $\mathcal{D}_\delta$ instead of considering the whole set of possible honest promises, as a single promise $\iota \in \mathcal{D}_\delta$ kept ``honestly enough'' encompasses all the promises $\widetilde{\iota} \in [\iota-\lambda,\iota+\lambda]$ that could be kept honestly. 
Furthermore, Algorithm~\ref{alg:approx_oracle} imposes the constraint represented by Equation~\ref{eq:local_ic_constr}, so that the policy $\sigma$ results $\epsilon$-IC.
Satisfying these constraints, it computes the (approximately) optimal policy-functions at step $h$, state $s$ and with promise $\iota$, then fills $M^\delta_h(s,\iota)$ with the value it achieves for the principal.
When the tables $M^\delta_h$ are filled for every $h \in \mathcal{H}$, it chooses for every state the initial promise in $\mathcal{D}_\delta$ that maximizes the principal's utility. 

\subsection{The optimization problem}
At Line~\ref{line:use_oracle} Algorithm~\ref{alg:approx_policy}, ideally we would like to compute the components of the promise-form policy in order to maximize the principal's expected utility from state $s \in \sset$ and step $h \in \mathcal{H}$ while promising $\iota \in \mathcal{D}_\delta$ ``honestly enough''.
However, this requires us to solve a mixed integer quadratic program, thus we will develop an approximation method.

We start by defining the problem $\mathcal{P}_{h,s,\iota}(M)$, whose optimal solution maximizes the principal's cumulative reward while honestly keeping the promise $\iota$ at step $h$ and in state $s$, assuming that the future rewards are the ones specified by table $M$.
Given a table $M$ represented by the function $M : \sset \times \mathcal{D}_\delta \rightarrow \mathbb{R} \cup \{-\infty\}$, which at Line~\ref{line:use_oracle} Algorithm~\ref{alg:approx_oracle} is $M \coloneqq M^\delta_{h+1}$, we let $\mathcal{P}_{h,s,\iota}(M)$ be the optimization problem defined as:
\begin{equation*}
	\widehat{F}_{h,s,\iota}(M) \coloneqq \max_{\substack{\alpha \in \Delta(\A)\\ p \in \mathcal{P}^{|\A|}\\ z:\A \times \sset \rightarrow \mathcal{D}_\delta}}
	F_{h,s,M}(\alpha,p,z) \quad \text{s.t.} \quad (\alpha, p, z) \in \Psi^{h,s}_{\iota,0}.
\end{equation*}
Variable $\alpha_a$ corresponds to the probability of recommending action $a \in \A$, variable $p^a$ is the contract that incentives $a$, and $z$ corresponds to the function $g_h$ to compute the next promise. 

The objective function is defined as:
\begin{equation*}
	F_{h,s,M}(\alpha,p,z) = \sum_{a \in \A} \alpha_a \sum_{s' \in \sset} P_h(s'|s,a) \left(r^P_h(s,p^a,s') + M(s,z(a,s'))\right).
\end{equation*}

The set $\Psi^{h,s}_{\iota,0}$ includes all the variables $\alpha \in \Delta(\A), p \in \mathcal{P}^{|\A|}, z: \A \times \sset \rightarrow \mathcal{D}_\delta$ that satisfy the following constraints:
\begin{align}
	\sum_{a \in A}\alpha_a \sum_{s' \in \sset} P_h(s'|s,a) \left(r^A_h(s,p^a,a,s') +z(a,s')\right) &= \iota \label{eq:opt_problem_constr_honest} \\
	\begin{split}
		\alpha_a \sum_{s' \in \sset} P_h(s'|s,a) \left(r^A_h(s,p^a,a,s') + z(a,s')\right) &\ge \\
		\alpha_a \sum_{s' \in \sset} P_h(s'|s,\widehat{a}) \big(r^A_h(s,p^a,\widehat{a},s') + &z(a,s')\big) \quad \forall a,\widehat{a} \in \A \label{eq:opt_problem_constr_ic}		
	\end{split}
\end{align}
Intuitively, Equation~\ref{eq:opt_problem_constr_honest} and Equation~\ref{eq:opt_problem_constr_ic} correspond to honesty and incentive-compatibility respectively.
The term $\alpha_a$ in Equation~\ref{eq:opt_problem_constr_ic} makes sure that the constraint is trivially satisfied when action $a$ is not recommended.

We observe that the problem can be written as a mixed integer quadratic program (observe problem $\mathcal{P}_{h,s,\iota}(M)$ optimizes over functions $z:\A \times \sset \rightarrow \mathcal{D}_\delta$ with discrete values). 
In Appendix~\ref{appendix:approx_orcale} we provide a detailed solution to efficiently compute an approximate solution to it.
In particular, such an approximate solution belongs to a ``relaxed'' feasible space $\Psi^{h,s}_{\iota,\lambda}$, for a given $\lambda > 0$, defined by the set of variables $\alpha \in \Delta(\A), p \in \mathcal{P}^{|\A|}, z: \A \times \sset \rightarrow \mathcal{D}_\delta$ that satisfy the following constraints:
\begin{align}
	\sum_{a \in A}\alpha_a \sum_{s' \in \sset} P_h(s'|s,a) \left(r^A_h(s,p^a,a,s') +z(a,s')\right) &\ge \iota -\lambda \label{eq:opt_rel_prob_constr_honest_ge} \\
	\sum_{a \in A}\alpha_a \sum_{s' \in \sset} P_h(s'|s,a) \left(r^A_h(s,p^a,a,s') +z(a,s')\right) &\le \iota +\lambda \label{eq:opt_rel_prob_constr_honest_le} \\
	\begin{split}
		\alpha_a \sum_{s' \in \sset} P_h(s'|s,a) \left(r^A_h(s,p^a,a,s') + z(a,s')\right) &\ge \\
		\alpha_a \sum_{s' \in \sset} P_h(s'|s,\widehat{a}) \big(r^A_h(s,p^a,\widehat{a},s') + &z(a,s')\big) \quad \forall a,\widehat{a} \in \A \label{eq:opt_rel_prob_constr_ic}
	\end{split}
\end{align}
Intuitively, Equation~\ref{eq:opt_rel_prob_constr_honest_ge} and Equation~\ref{eq:opt_rel_prob_constr_honest_le} relax Equation~\ref{eq:opt_problem_constr_honest}, requiring the policy to be $\lambda$-honest at step $h$.
We observe that $\Psi^{h,s}_{\iota,\lambda} = \bigcup_{\widetilde{\iota} \in [\iota-\lambda,\iota+\lambda]} \Psi^{h,s}_{\iota,0}$.
This property let us work with the discrete set of promises $\mathcal{D}_\delta$.

Given these sets, we can provide the definition of \emph{approximation oracle} as follows.
\begin{definition}[Approximation oracle]
	\label{def:approx_oracle}
	An algorithm $O^{h,s}_{\iota,\delta}$ is an approximation oracle for $\mathcal{P}_{h,s,\iota}(M)$ if it returns a tuple $(\alpha,p,z,v)$, with $\alpha \in \Delta(\mathcal{A}),p \in \mathcal{P}^{|\mathcal{A}|}, z:\mathcal{A} \times \sset \rightarrow \mathcal{D}_\delta, v \in \mathbb{R} \cup \{-\infty\}$, such that:
	\begin{enumerate}
		\item $\widehat{F}_{h,s,\widetilde{\iota}}(M) \le v \le F_{h,s,M}(\alpha,p,z)$ for every $\widetilde{\iota} \in [\iota-\delta,\iota+\delta]$.
		\item if $v>-\infty$, then $(\alpha,p,z) \in \Psi^{h,s}_{\iota,\lambda}$ with $\lambda = 2\delta$.
		\item $M(s',z(a,s')) > -\infty$ for every $a \in \A$ and $s' \in \sset$.
	\end{enumerate}
\end{definition}
An approximation oracle relaxes the problem $\mathcal{P}_{h,s,\iota}(M)$ in two ways: (1) it uses real numbers, $q(a,s') \in \mathbb{R}$, in place of the discrete variables $z(a,s') \in \mathcal{D}_\delta$, and (2) it relaxes the constraints, so that the solution belongs to $\Psi^{h,s}_{\iota,\delta}$.
The resulting problem is equivalent to an LP, which can be solved in polynomial time.
We observe that considering the set $\Psi^{h,s}_{\iota,\delta}$ is ``equivalent'' to consider $\Psi^{h,s}_{\widetilde{\iota},0}$, with $\widetilde{\iota} \in [\iota-\delta,\iota+\delta]$.
This allows us to use a discrete set of promises $\mathcal{D}_\delta$, as we do not force the principal to be completely honest.
As a final step, the oracle approximates the values of $q(a,s')$ to discrete values in $\mathcal{D}_\delta$, so that the future promises $z(a,s')$ belong to $I_{h+1}(s')$.
Due to this final approximation step, the solution $(\alpha,p,z)$ returned by the approximation oracle belongs to $\Psi^{h,s}_{\iota,2\delta}$.
Finally, the oracle returns also the value $v$ of the LP, which is used by Algorithm~\ref{alg:approx_oracle} to fill the value $M^\delta_{h}(s,\iota)$ of table $M^\delta_h$.
This value is different from the one of the function $F$ evaluated in $(\alpha,p,z)$.
We fill the table with it in order to recover some concavity properties, as we discuss in detail in Appendix~\ref{appendix:approx_orcale}.

\begin{restatable}{theorem}{MainTh}
	\label{th:main_theorem}
	For any $\epsilon>0$, Algorithm~\ref{alg:approx_policy} instantiated with $\delta=\nicefrac{\epsilon}{4H^2}$ computes an $\epsilon$-IC promise-form policy $\sigma=\{(I_h,J_h,\varphi_h,g_h)\}_{h \in \mathcal{H}}$ in time polynomial in the instance size and $\nicefrac{1}{\epsilon}$.
	Furthermore, $\sigma$ is such that $V^{\text{P},\pi^\sigma} \ge \text{OPT}$, while $\mathcal{I}=\mathcal{O}(H^4B/\epsilon)$ and $\mathcal{J}=\mathcal{O}(H^4B|\sset||\mathcal{A}|/\epsilon)$.
\end{restatable}


\bibliographystyle{plainnat}

\bibliography{references}

\begin{thebibliography}{31}
\providecommand{\natexlab}[1]{#1}
\providecommand{\url}[1]{\texttt{#1}}
\expandafter\ifx\csname urlstyle\endcsname\relax
  \providecommand{\doi}[1]{doi: #1}\else
  \providecommand{\doi}{doi: \begingroup \urlstyle{rm}\Url}\fi

\bibitem[Alimonti and Kann(2000)]{alimonti2000some}
Paola Alimonti and Viggo Kann.
\newblock Some apx-completeness results for cubic graphs.
\newblock \emph{Theoretical Computer Science}, 237\penalty0 (1-2):\penalty0
  123--134, 2000.

\bibitem[Alon et~al.(2021)Alon, D{\"u}tting, and
  Talgam-Cohen]{alon2021contracts}
Tal Alon, Paul D{\"u}tting, and Inbal Talgam-Cohen.
\newblock Contracts with private cost per unit-of-effort.
\newblock In \emph{Proceedings of the 22nd ACM Conference on Economics and
  Computation}, pages 52--69, 2021.

\bibitem[Alon et~al.(2023)Alon, Duetting, Li, and
  Talgam-Cohen]{alon2023bayesian}
Tal Alon, Paul Duetting, Yingkai Li, and Inbal Talgam-Cohen.
\newblock Bayesian analysis of linear contracts.
\newblock In \emph{Proceedings of the 24th ACM Conference on Economics and
  Computation}, EC '23, page~66, 2023.

\bibitem[Bacchiocchi et~al.(2024)Bacchiocchi, Castiglioni, Marchesi, and
  Gatti]{bacchiocchilearning}
Francesco Bacchiocchi, Matteo Castiglioni, Alberto Marchesi, and Nicola Gatti.
\newblock Learning optimal contracts: How to exploit small action spaces.
\newblock In \emph{The Twelfth International Conference on Learning
  Representations}, 2024.

\bibitem[Bastani et~al.(2016)Bastani, Bayati, Braverman, Gummadi, and
  Johari]{bastani2016analysis}
Hamsa Bastani, Mohsen Bayati, Mark Braverman, Ramki Gummadi, and Ramesh Johari.
\newblock Analysis of medicare pay-for-performance contracts.
\newblock \emph{Available at SSRN 2839143}, 2016.

\bibitem[Bechtel et~al.(2022)Bechtel, Dughmi, and Patel]{bechtel2022delegated}
Curtis Bechtel, Shaddin Dughmi, and Neel Patel.
\newblock Delegated pandora's box.
\newblock In \emph{Proceedings of the 23rd ACM Conference on Economics and
  Computation}, pages 666--693, 2022.

\bibitem[Ben-Porat et~al.(2024)Ben-Porat, Mansour, Moshkovitz, and
  Taitler]{ben2024principal}
Omer Ben-Porat, Yishay Mansour, Michal Moshkovitz, and Boaz Taitler.
\newblock Principal-agent reward shaping in mdps.
\newblock In \emph{Proceedings of the AAAI Conference on Artificial
  Intelligence}, volume~38, pages 9502--9510, 2024.

\bibitem[Bernasconi et~al.(2023)Bernasconi, Castiglioni, Marchesi, and
  Mutti]{bernasconi2023persuading}
Martino Bernasconi, Matteo Castiglioni, Alberto Marchesi, and Mirco Mutti.
\newblock Persuading farsighted receivers in mdps: the power of honesty.
\newblock \emph{Advances in Neural Information Processing Systems}, 36, 2023.

\bibitem[Bernasconi et~al.(2024{\natexlab{a}})Bernasconi, Castiglioni, and
  Marchesi]{bernasconi2024regret}
Martino Bernasconi, Matteo Castiglioni, and Alberto Marchesi.
\newblock Regret-minimizing contracts: Agency under uncertainty.
\newblock \emph{arXiv preprint arXiv:2402.13156}, 2024{\natexlab{a}}.

\bibitem[Bernasconi et~al.(2024{\natexlab{b}})Bernasconi, Castiglioni,
  Marchesi, and Mutti]{bernasconi2024persuading}
Martino Bernasconi, Matteo Castiglioni, Alberto Marchesi, and Mirco Mutti.
\newblock Persuading farsighted receivers in mdps: the power of honesty.
\newblock \emph{Advances in Neural Information Processing Systems}, 36,
  2024{\natexlab{b}}.

\bibitem[Bertsimas and Tsitsiklis(1997)]{Bertsimas}
Dimitris Bertsimas and John Tsitsiklis.
\newblock \emph{Introduction to Linear Optimization}.
\newblock Athena Scientific, 1st edition, 1997.
\newblock ISBN 1886529191.

\bibitem[Castiglioni et~al.(2022)Castiglioni, Marchesi, and
  Gatti]{castiglioni2022bayesian}
Matteo Castiglioni, Alberto Marchesi, and Nicola Gatti.
\newblock Bayesian agency: Linear versus tractable contracts.
\newblock \emph{Artificial Intelligence}, 307:\penalty0 103684, 2022.

\bibitem[Castiglioni et~al.(2023{\natexlab{a}})Castiglioni, Marchesi, and
  Gatti]{castiglioni2023designing}
Matteo Castiglioni, Alberto Marchesi, and Nicola Gatti.
\newblock Designing menus of contracts efficiently: The power of randomization.
\newblock \emph{Artificial Intelligence}, 318:\penalty0 103881,
  2023{\natexlab{a}}.

\bibitem[Castiglioni et~al.(2023{\natexlab{b}})Castiglioni, Marchesi, and
  Gatti]{castiglioni2023multi}
Matteo Castiglioni, Alberto Marchesi, and Nicola Gatti.
\newblock Multi-agent contract design: How to commission multiple agents with
  individual outcomes.
\newblock In \emph{Proceedings of the 24th ACM Conference on Economics and
  Computation}, page 412–448, 2023{\natexlab{b}}.

\bibitem[Chen et~al.(2024)Chen, Chen, Deng, and Huang]{chen2024bounded}
Yurong Chen, Zhaohua Chen, Xiaotie Deng, and Zhiyi Huang.
\newblock Are bounded contracts learnable and approximately optimal?
\newblock \emph{arXiv preprint arXiv:2402.14486}, 2024.

\bibitem[Cohen et~al.(2022)Cohen, Deligkas, and Koren]{cohen2022learning}
Alon Cohen, Argyrios Deligkas, and Moran Koren.
\newblock Learning approximately optimal contracts.
\newblock In \emph{Algorithmic Game Theory: 15th International Symposium, SAGT
  2022, Colchester, UK, September 12–15, 2022, Proceedings}, page 331–346,
  Berlin, Heidelberg, 2022. Springer-Verlag.

\bibitem[Cong and He(2019)]{cong2019blockchain}
Lin~William Cong and Zhiguo He.
\newblock Blockchain disruption and smart contracts.
\newblock \emph{The Review of Financial Studies}, 32\penalty0 (5):\penalty0
  1754--1797, 2019.

\bibitem[D{\"u}tting et~al.(2019)D{\"u}tting, Roughgarden, and
  Talgam-Cohen]{dutting2019simple}
Paul D{\"u}tting, Tim Roughgarden, and Inbal Talgam-Cohen.
\newblock Simple versus optimal contracts.
\newblock In \emph{Proceedings of the 2019 ACM Conference on Economics and
  Computation}, pages 369--387, 2019.

\bibitem[Dutting et~al.(2021)Dutting, Roughgarden, and
  Talgam-Cohen]{dutting2021complexity}
Paul Dutting, Tim Roughgarden, and Inbal Talgam-Cohen.
\newblock The complexity of contracts.
\newblock \emph{SIAM Journal on Computing}, 50\penalty0 (1):\penalty0 211--254,
  2021.

\bibitem[D{\"u}tting et~al.(2022)D{\"u}tting, Ezra, Feldman, and
  Kesselheim]{dutting2022combinatorial}
Paul D{\"u}tting, Tomer Ezra, Michal Feldman, and Thomas Kesselheim.
\newblock Combinatorial contracts.
\newblock In \emph{2021 IEEE 62nd Annual Symposium on Foundations of Computer
  Science (FOCS)}, pages 815--826. IEEE, 2022.

\bibitem[Gan et~al.(2022)Gan, Han, Wu, and Xu]{gan2022optimal}
Jiarui Gan, Minbiao Han, Jibang Wu, and Haifeng Xu.
\newblock Optimal coordination in generalized principal-agent problems: A
  revisit and extensions.
\newblock \emph{arXiv preprint arXiv:2209.01146}, 2022.

\bibitem[Guruganesh et~al.(2021)Guruganesh, Schneider, and
  Wang]{guruganesh2021contracts}
Guru Guruganesh, Jon Schneider, and Joshua~R Wang.
\newblock Contracts under moral hazard and adverse selection.
\newblock In \emph{Proceedings of the 22nd ACM Conference on Economics and
  Computation}, pages 563--582, 2021.

\bibitem[Guruganesh et~al.(2023)Guruganesh, Schneider, Wang, and
  Zhao]{GuruganeshPower23}
Guru Guruganesh, Jon Schneider, Joshua Wang, and Junyao Zhao.
\newblock The power of menus in contract design.
\newblock In \emph{Proceedings of the 24th ACM Conference on Economics and
  Computation}, EC '23, page 818–848, 2023.

\bibitem[Ho et~al.(2014)Ho, Slivkins, and Vaughan]{ho2015adaptive}
Chien-Ju Ho, Aleksandrs Slivkins, and Jennifer~Wortman Vaughan.
\newblock Adaptive contract design for crowdsourcing markets: Bandit algorithms
  for repeated principal-agent problems.
\newblock In \emph{Proceedings of the Fifteenth ACM Conference on Economics and
  Computation}, EC '14, page 359–376, New York, NY, USA, 2014. Association
  for Computing Machinery.
\newblock ISBN 9781450325653.

\bibitem[Ivanov et~al.(2024)Ivanov, D{\"u}tting, Talgam-Cohen, Wang, and
  Parkes]{ivanov2024principal}
Dima Ivanov, Paul D{\"u}tting, Inbal Talgam-Cohen, Tonghan Wang, and David~C
  Parkes.
\newblock Principal-agent reinforcement learning.
\newblock \emph{arXiv preprint arXiv:2407.18074}, 2024.

\bibitem[Kleinberg and Kleinberg(2018)]{kleinberg2018delegated}
Jon Kleinberg and Robert Kleinberg.
\newblock Delegated search approximates efficient search.
\newblock In \emph{Proceedings of the 2018 ACM Conference on Economics and
  Computation}, pages 287--302, 2018.

\bibitem[Wu et~al.(2024)Wu, Chen, Wang, Wang, and Xu]{wu2024contractual}
Jibang Wu, Siyu Chen, Mengdi Wang, Huazheng Wang, and Haifeng Xu.
\newblock Contractual reinforcement learning: Pulling arms with invisible
  hands.
\newblock \emph{arXiv preprint arXiv:2407.01458}, 2024.

\bibitem[Yu and Ho(2022)]{yu2022environment}
Guanghui Yu and Chien-Ju Ho.
\newblock Environment design for biased decision makers.
\newblock In \emph{IJCAI}, pages 592--598, 2022.

\bibitem[Zhang and Parkes(2008)]{zhang2008value}
Haoqi Zhang and David~C Parkes.
\newblock Value-based policy teaching with active indirect elicitation.
\newblock In \emph{AAAI}, volume~8, pages 208--214, 2008.

\bibitem[Zhang et~al.(2009)Zhang, Parkes, and Chen]{zhang2009policy}
Haoqi Zhang, David~C Parkes, and Yiling Chen.
\newblock Policy teaching through reward function learning.
\newblock In \emph{Proceedings of the 10th ACM conference on Electronic
  commerce}, pages 295--304, 2009.

\bibitem[Zhu et~al.(2023)Zhu, Bates, Yang, Wang, Jiao, and
  Jordan]{zhu2023sample}
Banghua Zhu, Stephen Bates, Zhuoran Yang, Yixin Wang, Jiantao Jiao, and
  Michael~I Jordan.
\newblock The sample complexity of online contract design.
\newblock In \emph{Proceedings of the 24th ACM Conference on Economics and
  Computation}, pages 1188--1188, 2023.

\end{thebibliography}

\newpage
\appendix
\section{Markovian policies}\label{sec:app_markov}
In this section we introduce Markovian policies and show wh they are not the ``right'' type of policies to tackle our problem.
First, we introduce the notation regarding Markovian policies.
Then, we propose an instance where no Markovian policy is optimal, and thus a history-dependent policy must be considered.
Finally, we prove that approximating the optimal Markovian policy is NP-hard.

\subsection{Markovian policies}
A Markovian (time-inhomogeneous) policy prescribes the current contract $p$ and action $a$ depending only the current state $s$ and time step $h$.
W.l.o.g., we restrict our attention to deterministic Markovian policies, as the optimal one is deterministic.
Formally, a Markovian policy $\rho \coloneqq \{p^s_h,a^s_h\}_{s \in \sset,h \in \mathcal{H}}$ prescribes contract $p^s_h$ and action $a^s_h$ at step $h \in \mathcal{H}$ and state $s \in \sset$.


As for history-dependent policies, we defined the value-functions for the principal and the agent.
The principal's value function $\vprincipalmark : \sset \times \mathcal{H}$ returns, given a state $s \in \sset$ and step $h \in \mathcal{H}$, the principal's expected cumulative utility from state $s$ and step $h$, formally:

\begin{equation*}
	\vprincipalmark_h(s) \coloneqq \sum_{s' \in \sset} P_h(s'|s,a^s_h) \left( r^\text{P}_h(s,p^s_h,s') +\vprincipalmark_{h+1}(s') \right).
\end{equation*}
Similarly, the agent's value function $\vprincipalmark : \sset \times \mathcal{H}$ returns, given a state $s \in \sset$ and step $h \in \mathcal{H}$, the agent's expected cumulative utility from state $s$ and step $h$, formally:
\begin{equation*}
	\vagentmark_h(s) \coloneqq \sum_{s' \in \sset} P_h(s'|s,a^s_h) \left( r^\text{A}_h(s,p^s_h,a^s_h,s') +\vagentmark_{h+1}(s') \right).
\end{equation*}

We will only consider IC Markovian policies, \emph{i.e.}, policies $\rho \coloneqq \{p^s_h,a^s_h\}_{s \in \sset,h \in \mathcal{H}}$ such that for every state $s \in \sset$ and step $h \in \mathcal{H}$, the action $a^s_h$ maximizes the agent's expected cumulative utility from state $s$ and step $h$.\footnote{Observe that, as long as the agent break ties in favor of the principal when indifferent between multiple actions, a persuasive Markovian policy does not need to recommend an action and can prescribes the contract only.}
As a result, a IC Markovian policy must satisfy the following:
\begin{equation*}
	\begin{split}
		\sum_{s' \in \sset} &P_h(s'|s,a^s_h) \left( r^\text{A}_h(s,p^s_h,a^s_h,s') +\vagentmark_{h+1}(s') \right) \ge \sum_{s' \in \sset} P_h(s'|s,\widehat{a}) \left( r^\text{P}_h(s,p^s_h,\widehat{a},s') +\vagentmark_{h+1}(s') \right),
	\end{split}
\end{equation*}
for every state $s \in \sset$, step $h \in \mathcal{H}$, and action $\widehat{a} \in \A$.
We denote with $\varrho$ the set of IC Markovian policies.

Finally, the expected principal's cumulative utility of a Markovian policy $\rho \coloneqq \{p^s_h,a^s_h\}_{s \in \sset,h \in \mathcal{H}}$ is defined as:
\begin{equation*}
	\vprincipalmark \coloneqq \sum_{s \in \sset}\mu(s) \vprincipalmark_1(s),
\end{equation*}
while the optimal Markovian policy $\rho^\star$ maximizes the expected principal's cumulative utility, formally:
\begin{equation*}
	\rho^\star \coloneqq \argmax_{\rho \in X} \vprincipalmark
\end{equation*}

\subsection{Markovian policies are suboptimal}

\begin{figure}[H]
	\centering
	\includegraphics{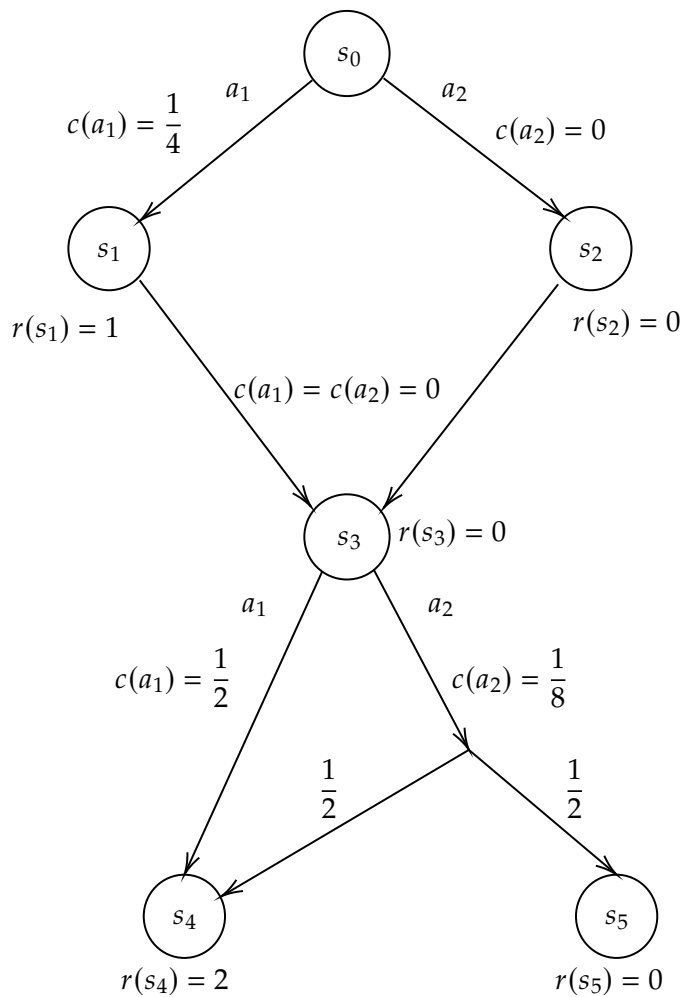}
	\caption{An instance where no Markovian policy is optimal.}
	\label{fig:history_necessary}
\end{figure}

\historyNecessary*
\begin{proof}
	Consider the instance depicted in Figure~\ref{fig:history_necessary}.
	The initial state is $s_0$ and the time horizon is $H=3$.
	In this example every state can only be reached at a single time step, thus we will drop the subscript $h$ when denoting functions that depend on them.
	Furthermore, rewards are not normalized in the set $[0,1]$ to ease the exposition.
	
	From state $s_0$, the agent can play action $a_1$ with cost $\nicefrac{1}{4}$ to reach state $s_2$, or action $a_2$ with cost zero to reach state $s_2$.
	Furthermore, the principal collects a reward of one is state $s_1$ is reached and of zero if, instead, the state transitions to $s_2$.
	From state $s_1$ and $s_2$ the agent can only reach state $s_3$ with no cost.
	For the sake of simplicity, we assume that the set of actions depends on the state, and that a single action $a_1$ is available in this states.
	From state $s_3$ the agent can play either $a_1$ with cost $\nicefrac{1}{2}$ or $a_2$ with cost $\nicefrac{1}{8}$. 
	We further assume the agent can play an action $a_3$ without cost and providing zero reward to the principal. 
	This action is not depicted in Figure~\ref{fig:history_necessary} and it is only required to guarantee a cost $0$ action and individual Rationality.
	States $s_4$ and $s_5$ are ``terminal states'', in the sense that they are reached at the end of the time horizon and no action is taken in them.
	
	
	We observe that the agent chooses an action only in states $s_0$ and $s_3$, while in states $s_1$ ad $s_2$ there is a single available action.
	Furthermore, given a policy that in state $s_1$ proposes a contract with payment $p^{s_1}(s_3)=k>0$, we can build an equivalent policy by increasing the previous payment $p^{s_0}(s_1)$ by $k$ and setting $p^{s_1}(s_3)$ to zero. A similar argument holds for state $s_2$.
	As a result, we can consider policies that do not pay in states $s_1$ and $s_2$, \emph{i.e.} they propose a contract with $p^{s_1}(s_3)=p^{s_2}(s_3)=0$.
	
	First, we upperbound the principal's utility achievable with Markovian policies.
	Against a Markovian policy in this instance, the agent acts myopically maximizing their immediate utility.
	Indeed, the first choice of the agent is in state $s_0$ and impacts the utilities at step $h=1$.
	Then, regardless of this choice, both the principal and the agent collects rewards of zero at step $h=2$, and the state reaches $s_3$ at step $h=3$.
	Finally, the Markovian policy ignores how state $s_3$ has been reached.
	As a result, the agent can act in order to maximize their immediate utility in state $s_1$, as this action does not impact the next choice at state $s_3$, and in state $s_3$, which is reached at the end of the time horizon.
	Hence, the optimal Markovian policy $\rho \coloneqq \{p^s,a^s\}_{s \in \sset}$ maximizes the principal's immediate utility in each state.
	In particular, such a policy prescribes action $a^{s_0}=a_1$ in state $s_0$, with payment $p^{s_0}(s_1)=\nicefrac{1}{4}$ if the future state is $s_1$ and $p^{s_0}(s_2)=0$ if the future state $s_2$.
	In state $s_3$, the action $a_1$ is incentivized by a contract $p^{s_3}$ such that:
	\begin{equation*}
		p^{s_3}(s_4) - \frac{1}{2} \ge \frac{1}{2}p^{s_3}(s_5) +\frac{1}{2}p^{s_3}(s_5) -\frac{1}{8}.
	\end{equation*}
	The inequality above implies:
	\begin{equation*}
		p^{s_3}(s_4) \ge \frac{3}{4} +p^{s_3}(s_5).
	\end{equation*}
	Thus, the contract to incentivize action $a_1$ with the minimum payment is the one that pays $p^{s_3}(s_4)=\nicefrac{3}{4}$ for the terminal state $s_4$ and zero for any other.
	With such a contract, the principal collects an immediate utility of $\vprincipalmark(s_3)=\nicefrac{5}{4}$.
	This is optimal, as any other agent's action would provide an immediate utility smaller than $1$.
	Overall, the expected cumulative of the principal is $\vprincipalmark(s_0) = 2$.
	
	Now, we design an history-dependent policy with strictly larger utility, concluding the proof.
	We observe that while $s_3$ is always reached, the contract and hence the agent's action in $s_3$ depend on the agent's action in $s_0$.
	Consequently, the agent does not necessarily play in order to maximize the immediate utility in state $s_0$, since the action in this state influences the following choice.
	This can be exploited by a history-dependent policy in order to reduce the payment in state $s_0$ with respect to the optimal Markovian policy.
	Let us define the following two histories:
	\begin{align*}
		\tau_1 &\coloneqq (s_0,a_1,s_1,a_1,s_3) \\
		\tau_2 &\coloneqq (s_0,a_2,s_2,a_1,s_3).
	\end{align*}
	Intuitively, history $\tau_1$ corresponds to the case in which the agent plays action $a_1$ in state $s_0$, while history $\tau_2$ encodes the case in which the agent plays $a_2$.
	Thus, these two histories represent the two paths through which the agent can reach state $s_3$.
	
	We design a deterministic policy $\pi$ that, in state $s_0$, prescribes action $a_1$ and does not provide any payment.
	If the state $s_3$ is reached through history $\tau_1$, the policy prescribes the same contract proposed by the Markovian policy in $s_3$ and recommend action $a_1$.
	Instead, it does not provide any payment for history $\tau_2$, while it recommends action $a_3$.
	
	We can show that the policy $\pi$ is IC.
	Indeed, as shown before, the agent is incentivized to play action $a_1$ in state $s_1$ when the history is $\tau_1$.
	This provides an immediate utility:
	\begin{equation*}
		\vagent(\tau_1) = \frac{3}{4} - \frac{1}{2} = \frac{1}{4}.
	\end{equation*}
	Instead, if the history is $\tau_2$, then the agent plays action $a_3$ with no cost and no payment, as any other action incurs in a cost with no payment.
	Thus:
	\begin{equation*}
		\vagent(\tau_2) = 0.
	\end{equation*}
	Furthermore, if the agent plays $a_1$ in $s_0$, then it collects a cumulative utility of:
	\begin{equation*}
		\vagent(\tau_0) = \vagent(\tau_1) - \frac{1}{4} = 0,
	\end{equation*}
	where $\tau_0 \coloneqq(s_0)$ is the initial history.
	Instead, by playing action $a_2$ in $s_0$, the agent collects a cumulative utility of zero.
	As a result, by recommending action $a_1$ in $s_0$ the policy $\pi$ is IC.
	
	Finally, the policy $\pi$ provides the principal with a cumulative utility of:
	\begin{equation*}
		\vprincipal(\tau_0) = 1 + \vprincipal(\tau_1) = 1 +2 -\frac{3}{4} = \frac{9}{4} > 2,
	\end{equation*}
	concluding the proof.
\end{proof}

\begin{figure}[h]
	\centering
	\includegraphics{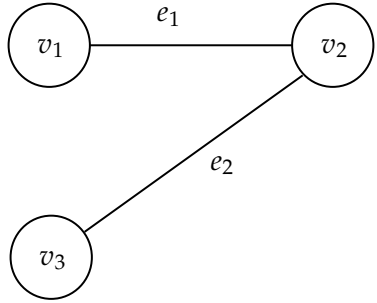}
	\caption{Illustrative instance of a cubic graph with $V=\{v_1,v_2,v_3\}$ and $E=\{e_1,e_2\}$.}
	\label{fig:cubic_graph}
\end{figure}

\begin{figure}
	\centering
	\resizebox{\linewidth}{!}{\input{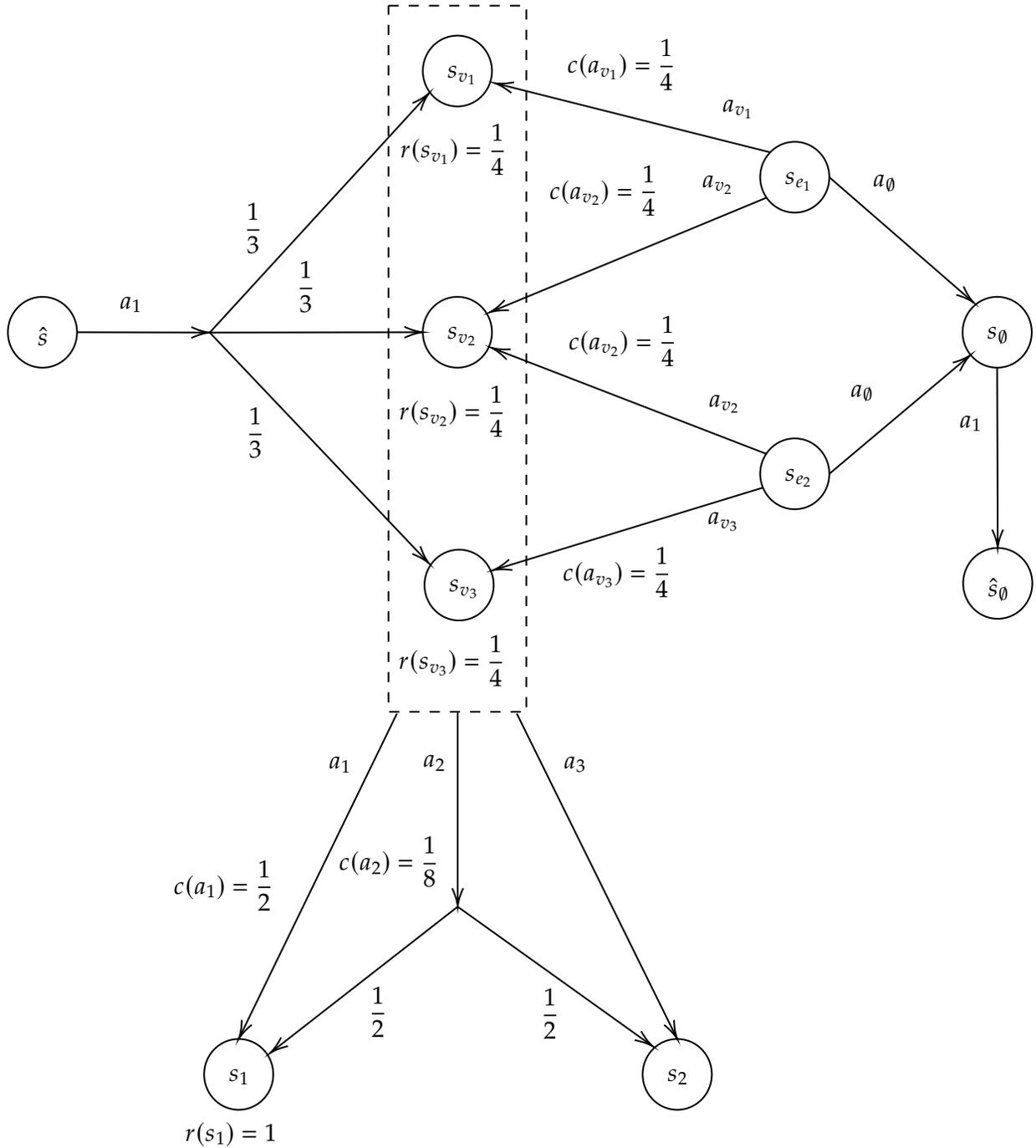}}
	\caption{Instance of the Contract MDP problem constructed from the Cubic Graph depict in Figure~\ref{fig:cubic_graph}. Costs and rewards are reported only when different from zero. The states $s_{v_1}$, $s_{v_2}$ and $s_{v_3}$ share the same set of actions $\{a_1,a_2,a_3\}$.}
	\label{fig:mdp_from_graph}
\end{figure}

\Hard*

\begin{proof}
	
	We reduce from an NP-Hard problem related to find vertex covers in cubic graphs. In particular, \citet{alimonti2000some} show that there exists a constant $\epsilon>0$ such that the following promise problem is NP-Hard:
	 Given a graph $(V,E)$ in which each nodes has degree at most $3$, distinguish whether there is a vertex cover of size $k$ or all the vertex covers have size at least $k(1+\epsilon)$.
	  
	\paragraph{Construction.}
	Given an graph $(V,E)$, we build an instance as follows.
In the following, for the ease of exposition we assume that the set of actions depends on the state.
	There is a state $s_e$ for each edge $e \in E$. The agent has three actions in each state $s_e$, $e=(v_1,v_2)$: 
	\begin{itemize} 
		\item $a_\varnothing$ that leads to state $s_{\varnothing}$  that has reward $r(s_e,s_\varnothing)=0$ and cost $c(s_e,a_1)=0$,
		\item  $a_{v_1}$ that has cost $c(s_e,a_1)=1/4$ and leads deterministically to node $s_{v_1}$,
		\item  $a_{v_2}$ that has cost $c(s_e,a_1)=1/4$ and leads deterministically to node $s_{v_2}$.
	\end{itemize}
	In state $s_\varnothing$, there is a single action $a_1$ with cost $c(s_\varnothing,a_1)=0$ that leads to a state $\hat s_\varnothing$.
	
	There exists a state $\hat s$. In this  state, the agent has only an action $a_1$ available with  cost $c(\hat s, a_1)=0$ that leads to an uniform distribution over $s_v$, $v\in V$. 
	Finally, for each $v \in V$, there is a state $s_v$ with two available actions:
	\begin{itemize}
		\item  $a_1$ that leads deterministically to $s_1$ and has cost $c(s_v, a_1)=\frac{1}{2}$, \item $a_{2}$ that lead to $s_1$ with probability $1/2$ and to $s_2$ with probability $1/2$ and has cost $c(s_v, a_2)=1/8$,
		\item  $a_3$ that leads deterministically to $s_2$ and has cost $c(s_v, a_3)=0$.
	\end{itemize}

	The time horizon is $H=3$, and hence the interaction ends when the agent enters in state $s_1$ or $s_2$. Hence, it is not required to specified the action in these states. Moreover, since each state can be reach only in a specific time step, we remove this dependency.
	
	Then, we define the rewards. 
	Since they depend only on the reached state, we drop the dependency on the initial state, \emph{i.e.}, $r(s')=r(s,s')$ for all $s,s' \in \sset$.
	The reward reaching a state $s_v$, $v \in V$, is $1$, i.e., $r(s_v)=\frac{1}{4}$ for all $v$.
	Moreover, the reward reaching a state $s_1$ is $r(s_1)=1$. All the other rewards are set to $0$.
	 
	We conclude the construction specifing the initial state. The initial state is $\hat s$ with probability $\frac{1}{10}$, and for each $e \in E$ the initial state is $s_e$ with probability $\frac{9}{10|E|}$.

\paragraph{If.}

	Suppose that there exists a vertex cover $V^*$ of size $k$.
	Consider a Markovian policy that sets payment on all nodes $s_v$, $v\in V^*$ to $p^{s_v}(s_{1})=\frac{3}{4}$.
	On all nodes $s_{v}$, $v \notin V^*$, the payment is set to $p^{s_v}(s_{1})=\frac{1}{4}$.
	All the other payments are set to $0$.
	
	It is easy to see that when the initial node is $s_e$, $e=(v_1,v_2) \in E$, the agent will take action $a_v$ for a $v\in \{v_1,v_2\}  \cap V^*$, and then $a_1$ in the reached state $s_{v}$ with $v \in \{v_1,v_2\}$.
	Indeed, the utility of the agent playing this sequence of actions is $3/4- 1/4-1/2=0$.
	Playing sequence $a_1$ and $a_2$, the utility of the agent is $3/8- 1/4-1/8=0$.
	Finally, playing any other sequence of actions the payment is $0$, and hence the agent's utility is at most $0$.
	Then, simple calculations show that the cumulative principal's utility is $5/4-3/4=1/2$ when the initial state is $s_e$, $e\in E$.
	
	If the initial state is a $\hat s$, the agent plays the only action, i.e., $a_1$, and the transaction leads to a node $s_v$, $v \in V^*$, the principal's utility is $\frac{1}{4}$. Indeed, the agent is incentivized to play action $s_1$, with reward $1$ and payment $3/4$.
	
	If the initial state is a $\hat s$ and the transaction leads to a node $s_v$, $v \in V\setminus V^*$ the principal's utility is $\frac{1}{2}-\frac{1}{8}=\frac{3}{8}$. Indeed, the agent is incentivized to play action $a_2$, with expected reward $1$ and expected payment $1/8$.

	Hence, the expected utility is:
	\[\frac{9}{10}\frac{1}{2}+ \frac{1}{10|V|}\left[ \frac{1}{4} k +\frac{3}{8}(|V|-k)\right]= \frac{9}{20} + \frac{3}{80} - \frac{1}{80}\frac{k}{|V|}= \frac{39}{80}- \frac{1}{80}\frac{k}{|V|}.\]
	
	\paragraph{Only If.}
	 
	Suppose that all the vertex covers have size at least $(1+\epsilon) k$.
	
	Consider a policy $\rho\coloneqq \{p^s,a^s\}_{s \in \mathcal{S}}$ of the agent. We upperbound the utility of any policy such that the set of actions $\{a^s\}_{s \in \mathcal{S}}$ is IC.
	To do so, we show that it is sufficient to consider a subset of the IC contraints for the agent.
	
	In particular, we consider the following ‘‘sub-components'' of the graph.
	For each node $s_v$ with $a^{s_v}=a_1$, we let $E(v)$ be the set of $e$ such that $a^{s_e}=a_v$.
	Given a $v\in V$, we analyze the component $v \cup E(v)$ with $E(v)\neq \emptyset$. 
	Notice that the agent plays action $a_1$ in node $s_v$ by definition of $E(v)$  $E(v)\neq \emptyset$. Then, it must be the case that $p^{s_v}(s_{1})\ge\frac{3}{4}$.
	This implies that if the initial node is an $e \in E(v)$, the principal's utility is at most $\frac{1}{2}$ since the reward is at most $\frac{5}{4}$ and the payment at least $\frac{3}{4}$, while if the initial node is $\hat s$ and the transaction leads to state $v$, the utility is at most $1/4$ by a similar argument.

	Then, we upperbound the principal's utility if the initial node is $s_e$, $e \notin \cup_{v \in V} E(v)$.
	If the played action is $a_\varnothing$, the principal's utility is clearly $0$.
	If the played action is $a_v$, $v \in V$, and the following action is $a_2$ or $a_3$ (recall that by definition of $E(v)$ the agent will not play $a_1$), the utility is upperbounded by the social welfare of the sequence of actions that is $3/8$.
	
	Finally, consider the case in which the initial node is $\hat s$, the agent plays action $a_1$, and the transaction leads to state $v$, $E(v)= \emptyset$. Also in case, indipendently from the played action the principal's utility is upperbounded by the social welfare $\frac{3}{8}$.
	
	Now, we count the number of nodes belonging to each component. In particular, we let $n=\cup_{v\in V} E(v)$, and $m=|\{v:E(v)\neq \emptyset\}|$.
	Then, by the previous analysis of all the cases we can conclude that the expected principal's utility is at most:	
	\begin{align*}
		\frac{9}{10} \frac{1}{2} \frac{n}{|E|}+ \frac{9}{10} \frac{3}{8} \left(1-\frac{n}{|E|}\right) + \frac{1}{10}\frac{1}{4}\frac{m}{|V|} + \frac{1}{10}\frac{3}{8} \left(1-\frac{m}{|V|}\right)=  \frac{30}{80}+ \frac{9}{80}\frac{n}{|E|} - \frac{1}{80} \frac{m}{|V|}. 
	\end{align*}
	
	Now, we we proceed providing a lowerbound on $m$. Intuitively, we know that if $n=|E|$, then $m\ge (1+\epsilon)k$ since it is an vertex cover.
	We can generalize this argument noticing that adding $|E|-n$ nodes to $\{v:E(v)\neq \emptyset\}$ it is possible to build a vertex cover. To do that, it is sufficient to take all the edge $e(v_1,v_2)$ not in $\cup_v E(v)$ and add $v_1$ to the vertex cover.
	Hence, by the assumption on the size of vertex cover, we get
	\[  m+ |E|-n \ge (1+\epsilon)k.\]
	This implies:
	\[ m /|V| \ge \frac{(1+\epsilon)k+n-|E|}{|V|}  =  \frac{(1+\epsilon)k}{|V|} - \frac{|E|-n}{|V|}\ge  \frac{(1+\epsilon)k}{|V|} -\frac{2}{3} \frac{|E|-n}{|E|}, \]
	where the last inequality follows since each nodes has degree at most $3$.
	Hence,

		\begin{align*}
		 \frac{30}{80}+ \frac{9}{80}\frac{n}{|E|} - \frac{1}{80} \frac{m}{|V|}  & \le \frac{30}{80}  + \frac{9}{80}\frac{n}{|E|} - \frac{1}{80} \left[\frac{(1+\epsilon)k}{|V|} -\frac{2}{3} \frac{|E|-n}{|E|}\right]\\
		&  =   \frac{92}{240}  +\frac{25}{240} \frac{n}{|E|}- \frac{1}{80}\frac{(1+\epsilon)k}{|V|}  \\
		& \le   \frac{92}{240} +  \frac{25}{240} - \frac{1}{80}\frac{(1+\epsilon)k}{|V|} \\
		& \le   \frac{39}{80} - \frac{1}{80}\frac{(1+\epsilon)k}{|V|}.
	\end{align*}

	The proof is concluded observing that $\frac{k}{|V|}$ is constant.\footnote{Each non-trivial node is connected to at least another node. Hence $|E|\ge |V|/2$. Hence $k \ge |E|/3 \ge |V|/6$.}. Indeed, it implies that
	\[  \frac{\frac{39}{80} - \frac{1}{80}\frac{(1+\epsilon)k}{|V|}}{\frac{39}{80}- \frac{1}{80}\frac{k}{|V|}} \]
	is a constant strictly smaller than 1.

\end{proof}

\section{Details of Section~\ref{sec:promise_form}}
\label{appendix:promise_form}
\subsection{Value-functions of promise-form policies}
As we discussed in Section~\ref{sec:promise_form}, given a promise-form policy we can define equivalent (According to Lemma~\ref{lem:promise_functions_eq}) value and action-value functions that exploit the particular structure of this class of policies.
Formally, given a promise-form policy $\sigma = \{(I_h,J_h,\varphi_h,g_h)\}_{h \in \mathcal{H}}$, we define the principal/agent value function ${V_{h}}^{\textnormal{P/A},\sigma}$ as the function that, for each step $h \in \mathcal{H}$, state $s \in \mathcal{S}$ and promise $\iota \in I_{h}(s)$, recursively satisfies:
\begin{equation*}
	{V_{h}}^{\textnormal{P /A},\sigma} (s,\iota) \coloneqq \sum_{(p,a) \in \mathcal{X}_{\sigma}} \varphi_h(p,a | s,\iota)  {Q_{h}}^{\textnormal{P/A},\sigma} (s,\iota,p,a),
\end{equation*}
where the principal and the agent action-value functions are recursively defined as follows:
\begin{equation*}
	{Q_{h}}^{\textnormal{P},\sigma}(s,\iota,p,a) \coloneqq \sum_{s' \in \sset}P_h(s'|s,a) \left( r^\text{P}_h(s,p,s') + Q_{h+1}^{\textnormal{P},\sigma}(s',g_h(s,\iota,p,a,s')) \right),
\end{equation*}
and,
\begin{equation*}
	\qagentpr_h(s,\iota,p,a) \coloneqq \sum_{s' \in \sset} P_h(s'|s,a) \left(r^\text{A}_h(s,p,a,s') + \vagentpr_{h+1}(s',g_h(s,\iota,p,a,s')) \right).
\end{equation*}
for each step $h \in \mathcal{H}$, state $s \in \mathcal{S}$, contract $p \in \mathcal{J}$, action $a \in \mathcal{A}$ and promise $\iota \in I_{h}(s)$.
In the equations above, the support $\mathcal{X}_\sigma$ of the policy $\sigma$ is defined as $\mathcal{X}_\sigma \coloneqq \mathcal{X}_{\pi^\sigma}$.

We also define the value and action-value functions of the agent in the case in which the agent does not follow the principal's recommendation as follows:
\begin{equation*}
	\vagentdevpr_h(s,\iota) \coloneqq \sum_{(p,a) \in \mathcal{X}_\sigma} \varphi(p,a|s,\iota) \qagentdevpr_h(s,\iota,p,a).
\end{equation*}
\begin{equation*}
	\qagentdevpr_h(s,\iota,p,a) \coloneqq \max_{\widehat{a} \in A} \sum_{s' \in \sset} P(s'|s,\widehat{a}) \left(r^\text{A}_h(s,p,\widehat{a},s') +\vagentdevpr_{h+1}(s',g_h(s,\iota,p,a,s')) \right).
\end{equation*}

Finally, we observe that, thanks to Lemma~\ref{lem:promise_functions_eq}, the value of a promise-form policy $\sigma=\{(I_h,J_h,\varphi_h,g_h)\}_{h \in \mathcal{H}}$ can be written as:
\begin{equation*}
	V^{\text{P}, \pi^\sigma} = \sum_{s \in \sset} \mu(s) \vprincipal_1(s) = V^{\text{P}, \sigma} \coloneqq \sum_{s \in \sset} \mu(s) \vprincipalpr_1(s,i(s)).
\end{equation*}

\subsection{Omitted proofs from Section~\ref{sec:promise_form}}

\promiseFunctionsEq*
\begin{proof}
	As a first step, we observe that $\mathcal{X}_\pi = \mathcal{X}_\sigma$ by construction (see Algorithm~\ref{alg:implement}).
	We prove the statement by induction over the steps $\mathcal{H}$ starting from $h = H$.
	For $h=H$, by the definitions of $\qprincipal_H(\tau,p,a)$ and $\qprincipalpr_H(s,\iota^\sigma_\tau,p,a)$, it holds that:
	\begin{equation*}
		\qprincipal_H(\tau,p,a) = \sum_{s' \in \sset} P_H(s'|s,a)r^P_H(s,p,s') = \qprincipalpr_H(s,\iota^\sigma_\tau,p,a)
	\end{equation*}
	for every history $\tau \in \mathcal{T}_H$ ending in state $s \in \mathcal{S}$, and pair $(p,a) \in \mathcal{X}$.
	Furthermore, since by construction $\varphi_H(p,a|s,\iota^\sigma_\tau) = \pi(p,a|\tau)$, it holds that:
	\begin{equation*}
		\vprincipal_H (\tau) =\hspace{-0.3cm} \sum_{(p,a) \in \mathcal{X}_\sigma}\hspace{-0.3cm} \pi(p,a | \tau)  \qprincipal_H(\tau,p,a) = \hspace{-0.3cm} \sum_{(p,a) \in \mathcal{X}_\sigma} \hspace{-0.3cm} \varphi_H(p,a|s,\iota^\sigma_\tau)  \qprincipalpr_H(s,\iota^\sigma_\tau,p,a) = \vprincipalpr_H(s,\iota^\sigma_\tau).
	\end{equation*}
	With a similar argument, it is possible to show that $\vagentpr_H (s,\iota^\sigma_\tau) = \vagent_H(\tau)$ and $\qagentpr_H(s,\iota^\sigma_\tau,p,a) = \qagent_H(\tau,p,a)$ for every state $s \in \mathcal{S}$, history $\tau \in \mathcal{T}_H$ ending in state $s$, and pair $(p,a) \in \mathcal{X}$.
	Furthermore, we observe that:
	\begin{equation*}
		\qagentdev_H(\tau,p,a) = \max_{\widehat{a} \in A} \sum_{s' \in \sset} P_H(s'|s,\widehat{a}) r^\textnormal{A}(s,p,a,s') = \qagentdevpr_H(s,\iota^\sigma_\tau,p,a).
	\end{equation*}
	Consequently:
	\begin{equation*}
		\vagentdev_H (\tau) = \hspace{-0.3cm} \sum_{(p,a) \in \mathcal{X}_\sigma} \hspace{-0.3cm} \pi(p,a | \tau)  \qagentdev_H(\tau,p,a) = \hspace{-0.3cm} \sum_{(p,a) \in \mathcal{X}_\sigma} \hspace{-0.3cm} \varphi_H(p,a|s,\iota^\sigma_\tau)  \qagentdevpr_H(s,\iota^\sigma_\tau,p,a) = \vagentdevpr_H(s,\iota^\sigma_\tau).
	\end{equation*}
	
	Now we consider a generic step $0<h<H$ and assume that the statement holds for $h+1$.
	For every history $\tau \in \mathcal{T}_h$, contract $p \in \mathcal{P}$ and action $a \in \A$, the following holds:
	\begin{equation*}
		\qprincipal_h(\tau,p,a) = \sum_{s' \in \sset}P_h(s'|s,a) \left( r^\text{P}_h(s,p,s') + \vprincipal_{h+1}(\tau \oplus (p,a,s')) \right).
	\end{equation*}
	We observe that $\vprincipal_{h+1}(\tau \oplus (p,a,s')) = \vprincipalpr_{h+1}(s',\iota^\sigma_{\tau \oplus (p,a,s')})$ by inductive hypothesis.
	Furthermore, $\iota^\sigma_{\tau \oplus (p,a,s')} = g_h(s,\iota^\sigma_\tau,p,a,s')$ by construction (see Algorithm~\ref{alg:implement}). 
	Consequently, the following holds:\footnote{We let $g_h(s,\iota,p,a,s') \coloneqq \min(I_{h+1}(s'))$ for any contract $p \notin J_h(s,\iota,a)$. This way, the policy $\pi$ is well defined for every history, even those that happen with probability zero (as they include contracts that are never prescribed).}
	\begin{align*}
		\qprincipal_h(\tau,p,a) &= \sum_{s' \in \sset}P_h(s'|s,a) \left( r^\text{P}_h(s,p,s') + \vprincipal_{h+1}(\tau \oplus (p,a,s')) \right) \\
		&= \sum_{s' \in \sset}P_h(s'|s,a) \left( r^\text{P}_h(s,p,s') + \vprincipalpr_{h+1}(s',g_h(s,\iota^\sigma_\tau,p,a,s')) \right) \\
		&= \qprincipalpr_h(s,\iota^\sigma_\tau,p,a)
	\end{align*}
	As a result:
	\begin{equation*}
		\vprincipal_h (\tau) = \hspace{-0.3cm}  \sum_{(p,a) \in \mathcal{X}_\sigma} \hspace{-0.3cm} \pi(p,a | \tau)  \qprincipal_h(\tau,p,a) = \hspace{-0.3cm} \sum_{(p,a) \in \mathcal{X}_\sigma} \hspace{-0.3cm} \varphi_H(p,a|s,\iota^\sigma_\tau)  \qprincipalpr_H(s,\iota^\sigma_\tau,p,a) = \vprincipalpr_H(s,\iota^\sigma_\tau).
	\end{equation*}
	A similar argument shows that $\vagentpr_H (s,\iota^\sigma_\tau) = \vagent_H(\tau)$ and $\qagentpr_H(s,\iota^\sigma_\tau,p,a) = \qagent_H(\tau,p,a)$.
	
	Finally, we have that:
	\begin{align*}
		\qagentdev_h(\tau,p,a) &= \max_{\widehat{a} \in \mathcal{A}} \sum_{s' \in \sset} P_h(s'|s,\widehat{a})\left( r^{\text{A}}_h(s,p,\widehat{a},s') +\vagentdev_{h+1}(\tau \oplus (p,a,s')) \right) \\
		&= \max_{\widehat{a} \in \mathcal{A}} \sum_{s' \in \sset} P_h(s'|s,\widehat{a})\left( r^{\text{A}}_h(s,p,\widehat{a},s') +\vagentdevpr_{h+1}(s',g_h(s,\iota^\sigma_\tau,p,a,s')) \right) \\
		&=\qagentdevpr_h(s,\iota^\sigma_\tau,p,a)
	\end{align*}
	and
	\begin{equation*}
		\vagentdev_h (\tau) = \hspace{-0.3cm} \sum_{(p,a) \in \mathcal{X}_\sigma} \hspace{-0.3cm} \pi(p,a | \tau)  \qagentdev_h(\tau,p,a) = \hspace{-0.3cm} \sum_{(p,a) \in \mathcal{X}_\sigma} \hspace{-0.3cm} \varphi_H(p,a|s,\iota^\sigma_\tau)  \qagentdevpr_H(s,\iota^\sigma_\tau,p,a) = \vagentdevpr_H(s,\iota^\sigma_\tau),
	\end{equation*}
	concluding the proof.
\end{proof}

\honestyValue*
\begin{proof}
	We prove the statement by induction over the time steps $h \in \mathcal{H}$, starting from $h=H$.
	When $h=H$, for every state $s \in \mathcal{S}$, promise $\iota \in I_H(s)$ the following holds:
	\begin{align*}
		\left|\vagentpr_H(s,\iota) -\iota\right| &= \left| \sum_{(p,a) \in \mathcal{X}} \varphi_H(p,a|s,\iota) \qagentpr_H(s,\iota,p,a) -\iota \right|\\
		&= \left|\sum_{(p,a) \in \mathcal{X}} \varphi_H(p,a|s,\iota) \sum_{s' \in \sset} P_H(s'|s,a)r^\text{A}_H(s,p,a,s') - \iota \right| \\
		&= \left|\sum_{(p,a) \in \mathcal{X}} \varphi_H(p,a|s,\iota) \sum_{s' \in \sset} P_H(s'|s,a)\left(r^\text{A}_H(s,p,a,s') +g_H(s,\iota,p,a,s')\right) - \iota \right| \\
		&\le \eta,
	\end{align*}
	where the last equality holds because $I_{H+1}(s)=\{0\}$ by definition, while the inequality holds because $\sigma$ is $\eta$-honest at step $H$.
	
	Now we consider a time step $h<H$ and assume that the statement holds for $h' = h+1$.
	Then, for every state $s \in \mathcal{S}$, promise $\iota \in I_h(s)$ and pair $(p,a) \in \mathcal{X}$ such that $\varphi_H(p,a|s,\iota)>0$ the following holds:
	\begin{align*}
		\qagentpr_h(s,\iota,p,a) &= \sum_{s' \in \sset}P_h(s'|s,a) \left(r^\text{A}_h(s,p,a,s') + \vagentpr_{h+1}(s', g_h(s,p,a,\iota,s')) \right) \\
		&\ge \sum_{s' \in \sset}P_h(s'|s,a) \left(r^\text{A}_h(s,p,a,s') + g_h(s,p,a,\iota,s')-\eta(H-h) \right) \\
		&= \sum_{s' \in \sset}P_h(s'|s,a) \left(r^\text{A}_h(s,p,a,s') + g_h(s,p,a,\iota,s') \right) -\eta(H-h),
	\end{align*}
	where the first inequality holds by inductive hypothesis, while the last one holds because $\sum_{s' \in \sset} P_h(s'|s,a) = 1$.
	Thanks to this result and since $\sigma$ is $\eta$-honest, we have that:
	\begin{align*}
		\vagentpr_h(s,\iota) &= \sum_{(p,a) \in \mathcal{X}_\sigma} \varphi_h(p,a|s,\iota) \qagentpr_h(s,\iota,p,a) \\
		&\ge \iota - \eta -\eta(H-h) \\
		&= \iota - \eta(H-h+1).
	\end{align*}
	
	With a similar argument, we can show that $\vagentpr_h(s,\iota) \le \iota + \eta(H-h+1)$, concluding the proof.
\end{proof}

\localICConstr*
\begin{proof}
	Let $\pi$ be an history-dependent policy associated with $\sigma$.
	In order to prove that $\sigma$ (equivalently $\pi$) is $\epsilon$-IC, we show that for every step $h \in \mathcal{H}$, history $\tau \in \mathcal{T}_h$ contract $p \in \mathcal{P}$, and pair of actions $a,\widehat{a} \in \mathcal{A}$ such that $\pi(p,a|\tau)>0$, the following holds:
	\begin{equation}\label{eq:ic_to_prove}
		\qagent_h(\tau,p,a) \ge \sum_{s' \in \sset}P_h(s'|s,\widehat{a}) \left(r^\text{A}_h(s,p,\widehat{a},s') +\vagentdev_{h+1}(\tau \oplus (a,s'))\right) - \epsilon_h,
	\end{equation} 
	where $\epsilon_h \coloneqq 2\eta(H-h)^2 \le 2\eta H^2$.
	
	In order to prove that Equation~\ref{eq:ic_to_prove} holds for every time step $h \in \mathcal{H}$, we reason by induction, starting from $h=H$.
	Take any state $s \in \sset$, history $\tau \in \mathcal{T}$ ending in state $s$, contract $p \in \mathcal{P}$ and pair of actions $a,\widehat{a} \in \A$ such that $\pi(p,a|\tau)>0$.
	Then, thanks to Lemma~\ref{lem:promise_functions_eq}, we have that $\qagent_h(\tau,p,a) = \qagentpr_h(s,\iota^\sigma_\tau,p,a)$.
	Moreover, observing that $I_{H+1}(s) \coloneqq \{0\}$, the action-value of the agent can be lower bounded as follows:
	\begin{align*}
		\qagent_H(\tau,p,a) &= \qagentpr_H(s,\iota^\sigma_\tau,p,a) \\
		&= \sum_{s' \in \sset} P_H(s'|s,a) r^{\text{A}}_H(s,p,a,s') \\
		&= \sum_{s' \in \sset} P_H(s'|s,a) \left(r^{\text{A}}_H(s,p,a,s') +g_H(s,\iota,p,a,s')\right) \\
		&\ge \sum_{s' \in \sset} P_H(s'|s,\widehat{a}) \left(r^{\text{A}}_H(s,p,\widehat{a},s') +g_H(s,\iota,p,a,s')\right) \\
		&= \sum_{s' \in \sset} P_H(s'|s,\widehat{a}) r^{\text{A}}_H(s,p,\widehat{a},s'),
	\end{align*}
	where the inequality holds by hypothesis thanks to Equation~\ref{eq:local_ic_constr}.
	Since $\epsilon_H=0$, we conclude that Equation~\ref{eq:ic_to_prove} holds for $h=H$.
	
	Now we let $h<H$ and we assume that Equation~\ref{eq:ic_to_prove} holds for $h'=h+1$. Thus, given the inductive hypothesis, the following holds:
	\begin{equation*}
		\qagent_{h+1}(\tau,p,a) + \epsilon_{h+1} \ge \qagentdev_{h+1}(\tau,p,a),
	\end{equation*}
	for every history $\tau \in \mathcal{T}_{h+1}$ and every contract $p \in \mathcal{P}$ and action $a, \in \A$ such that $\pi(p,a|\tau)>0$.
	%
	Furthermore, observing that $\sum_{(p,a) \in \mathcal{X}} \pi(p,a|\tau)=1$, we can prove that:
	\begin{equation*}
		\vagent_{h+1}(\tau) + \epsilon_{h+1} \ge \vagentdev_{h+1}(\tau),
	\end{equation*}
	%
	which implies that:
	\begin{equation}\label{eq:v_dev_upper_bound}
		\vagent_{h+1}(\tau)  \ge \vagentdev_{h+1}(\tau) -\epsilon_{h+1} \ge \vagentdev_{h+1}(\tau) - 2\eta(H-h+1)(H-h),
	\end{equation}
	as $\epsilon_{h+1} = 2\eta(H-h-1)^2 \le 2\eta(H-h)(H-h-1)$.
	
	Consider any other action $\widehat{a} \neq a$. Since $\sum_{s' \in \sset} P_h(s'|s,a) = 1$, the action-value function of the agent can be lower bounded as follows:
	\begin{align}
		\qagentpr_h(s,\iota^\sigma_\tau,p,a) &= \sum_{s' \in \sset}P_h(s'|s,a) \left(r^\text{A}_h(s,p,a,s') + \vagentpr_{h+1}(s', g_h(s,\iota^\sigma_\tau,p,a,s')) \right) \nonumber \\
		&\ge \sum_{s' \in \sset}P_h(s'|s,a) \left(r^\text{A}_h(s,p,a,s') + g_h(s,\iota^\sigma_\tau,p,a,s') -\eta(H-h) \right) \label{eq:lower_q_1}\\
		&= \sum_{s' \in \sset}P_h(s'|s,a) \left(r^\text{A}_h(s,p,a,s') + g_h(s,\iota^\sigma_\tau,p,a,s') \right) -\eta(H-h) \nonumber \\
		&\ge \sum_{s' \in \sset}P_h(s'|s,\widehat{a}) \left(r^\text{A}_h(s,p,\widehat{a},s') + g_h(s,\iota^\sigma_\tau,p,a,s') \right) -\eta(H-h) \label{eq:lower_q_2}\\
		&\ge \sum_{s' \in \sset}P_h(s'|s,\widehat{a}) \left(r^\text{A}_h(s,p,\widehat{a},s') + \vagent_{h+1}(\tau \oplus (p,a,s')) \right) -2\eta(H-h) \label{eq:lower_q_3}\\ 
		&\ge \sum_{s' \in \sset}P_h(s'|s,\widehat{a}) \left(r^\text{A}_h(s,p,\widehat{a},s') + \vagentdev_{h+1}(\tau \oplus (p,a,s')) \right) -2\eta(H-h)^2 \label{eq:lower_q_4}
	\end{align}
	where Equation~\ref{eq:lower_q_1} holds thanks to Lemma~\ref{lem:honesty_value}, Equation~\ref{eq:lower_q_2} employs Equation~\ref{eq:local_ic_constr}, Equation~\ref{eq:lower_q_3} holds due to Lemma~\ref{lem:honesty_value} and Lemma~\ref{lem:promise_functions_eq}, while Equation~\ref{eq:lower_q_4} holds due to Equation~\ref{eq:v_dev_upper_bound}.
	As a result, by considering that $\qagentpr_h(s,\iota^\sigma_\tau,p,a) = \qagent_h(\tau,p,a)$, Equation~\ref{eq:ic_to_prove} is satisfied. This concludes the proof.
\end{proof}

\promisingOptimal*
\begin{proof}
	We show that, given an optimal IC policy $\pi^\star$, one can always build a promise-form $\sigma = \{(I_h,J_h,\varphi_h,g_h)\}_{h \in \mathcal{H}}$ that achieves the same expected principal's cumulative utility.
	Let us define, for every step $h \in \mathcal{H}$, the set $\mathcal{T}^\pi_h$ as:
	\begin{equation*}
		\mathcal{T}^\pi_h \hspace{-0.1cm} \coloneqq \{\tau \in \mathcal{T}_h | \tau = (s_1,p_1,a_1, \dots,s_{h-1},p_{h-1},a_{h-1},s), \pi(p_i,a_i|\tau(s_i)) > 0, i \in \mathcal{H} \cap [1,H-1] \},
	\end{equation*}
	where $\tau(s_i) = (s_1,p_1,a_1,\dots,s_{i-1},p_{i-1},a_{i-1},s)>0$.
	Given a history-dependent policy $\pi$ with finite support $\mathcal{X}_\pi$, we define $\sigma = \sigma_\pi$ as follows:
	\begin{itemize}
		\item $I_h(s) \coloneqq \{\vagent_h(\tau) | \tau \in \mathcal{T}^\pi_h\}$ for all $s \in \mathcal{S}$ and $h \in \mathcal{H}$. 
		Observe that $|I_1(s)|=1$ for very $s \in \sset$, since $\mathcal{T}^\pi_1 = \{(s_i) \mid s_i \in \sset\}$.
		\item $\varphi_h(p,a|s,\iota) = \pi(p,a|\tau_{h,s,\iota})$ for all $(p,a) \in \mathcal{X}$, $h \in \mathcal{H}$, $s \in \mathcal{S}$ and $\iota \in I_h(s)$, where:
		\begin{equation*}
			\tau_{h,s,\iota} \in \argmax_{\substack{ \tau \in \mathcal{T}^\pi_h : \\ \tau=(s_1,a_1,\dots,s_{h-1},a_{h-1},s_h), \\ s_h=s \; \wedge \; \vagent_h(\tau)=\iota }} \{\vprincipal_h(\tau)\}
		\end{equation*}
		\item $J_h(s,\iota,a) = \{p \in \mathcal{P} \mid \pi(p,a|\tau_{h,s,\iota})>0\}$ for every $h \in \mathcal{H}$, $\iota \in I_h(s)$ and $a \in \mathcal{A}$.
		\item $g_h(s,\iota,p,a,s') = \vagent_{h+1}(\tau_{h,s,\iota} \oplus (p,a,s'))$ for all $h \in \mathcal{H}$, $s,s' \in \mathcal{S}$, $\iota \in I_h(s)$, $a \in \mathcal{A}$ and $p \in J_h(s,\iota,a)$.
	\end{itemize}
	We observe that given this definition, it holds that $\mathcal{X}_\sigma \subseteq \mathcal{X}_\pi$.
	Furthermore, for every $h \in \mathcal{H}$ the set $\mathcal{T}^\pi_h$ is finite.
	Consequently, both the sets $\mathcal{I}$ and $\mathcal{J}$ are finite.
	
	As a first step, we show that, $\sigma$ is $\eta$-honest with $\eta=0$.
	For every $h \in \mathcal{H}$, $s \in \mathcal{S}$, $\iota \in I_h(s)$ and $(p,a) \in \mathcal{X}$ such that $\varphi_h(p,a|s,\iota) = \pi(p,a|\tau_{h,s,\iota})>0$ the following holds:
	\begin{align*}
		&\sum_{(p,a) \in \mathcal{X}_\sigma} \varphi(p,a|s,\iota) \sum_{s' \in \sset} P_h(s'|s,a) \left(r^\text{A}_h(s,p,a,s') +g_h(s,\iota,p,a,s') \right) \\
		&= \sum_{\substack{(p,a) \in \mathcal{X} :\\ \pi(p,a|\tau_{h,s,\iota})>0}} \pi(p,a|\tau_{h,s,\iota}) \sum_{s' \in \sset} P_h(s'|s,a) \left(r^\text{A}_h(s,p,a,s') +\vagent_{h+1}(\tau_{h,s,\iota} \oplus (p,a,s')) \right) \\
		&= \vagent_h(\tau_{h,s,\iota}) = \iota,
	\end{align*}
	where the first equality holds thanks to the definitions of $g_h$ and $\varphi_h$, the second one due to the definition of $\vagent_h(\tau_{h,s,\iota})$, and the last equality thanks to the definition of $\tau_{h,s,\iota}$.
	Thus, $\sigma$ is $\eta$-honest, with $\eta = 0$.
	
	Now we prove that, given that $\pi$ is IC and $\sigma$ is honest, $\sigma$ is IC as well.
	In order to do that, we show that Equation~\ref{eq:local_ic_constr} is satisfied and apply Lemma~\ref{lem:local_ic_constr}.
	First, we observe that thanks to the definitions of $\qprincipal_h(\tau_{h,s,\iota},p,a)$ and $g_h(s,\iota,a,s')$:
	\begin{equation}
		\begin{split}
		\label{eq:promise_form_optimal_rewrite_q}
		\qagent_h(\tau_{h,s,\iota},p,a) &= \sum_{s' \in \sset} P_h(s'|s,a) \left(r^\text{A}_h(s,p,a,s') + \vagent_{h+1}(\tau_{h,s,\iota} \oplus (p,a,s')) \right) \\
		&= \sum_{s' \in \sset} P_h(s'|s,a) \left(r^\text{A}_h(s,p,a,s') + g_h(s,\iota,p,a,s') \right),
		\end{split}
	\end{equation}
	for every step $h \in \mathcal{H}$, state $s \in \mathcal{S}$, promise $\iota \in I_h(s)$, and action $a \in \mathcal{A}$ with corresponding contract $p \in \mathcal{P}$ such that $\varphi_h(p,a|s,\iota)>0$.
	Furthermore, for every action $\widehat{a} \in \A$:
	\begin{equation}
		\begin{split}
		\label{eq:promise_form_optimal_lower_q}
		\qagent_h(\tau_{h,s,\iota},p,a) &\ge \sum_{s' \in \sset} P(s'|s,\widehat{a}) \left(r^\text{A}_h(s,p,\widehat{a},s') +\vagentdev_{h+1}(\tau_{h,s,\iota} \oplus (p,a,s')) \right) \\
		&= \sum_{s' \in \sset} P(s'|s,\widehat{a}) \left(r^\text{A}_h(s,p,\widehat{a},s') +\vagent_{h+1}(\tau_{h,s,\iota} \oplus (p,a,s')) \right) \\
		&=\sum_{s' \in \sset} P(s'|s,\widehat{a}) \left(r^\text{A}_h(s,p,\widehat{a},s') +g_h(s,\iota,p,a,s') \right)
		\end{split}
	\end{equation}
	since $\pi$ is IC and $\vagent_{h+1}(\tau_{h,s,\iota} \oplus (p,a,s')) =g_h(s,\iota,p,a,s')$.
	As a result, by combining Equation~\ref{eq:promise_form_optimal_rewrite_q} and Equation~\ref{eq:promise_form_optimal_lower_q}, we have that Equation~\ref{eq:local_ic_constr} is satisfied. 
	Consequently, by applying Lemma~\ref{lem:local_ic_constr}, the policy $\sigma$ is incentive compatible.
	
	In order to conclude the proof, we show that the promise-form policy $\sigma$ achieves a principal's expected cumulative utility at least as good as that obtained by $\pi$.
	%
	Let $\pi^\sigma$ be the history-dependent policy that implements $\sigma$ according to Algorithm~\ref{alg:implement} (notice that it may be different from $\pi$).
	Formally, we prove by induction over $h \in \mathcal{H}$ that $V^{\text{P},\pi^\sigma}_h(\tau) \ge \vprincipal_h(\tau)$ for every step $h \in \mathcal{H}$ and history $\tau \in \mathcal{T}^\pi_h$.
	
	As a base case, we consider $h=H$. 
	For every state $s \in \sset$ and history $\tau \in \mathcal{T}^\pi_H$ ending in state $s$ it holds that:
	\begin{align*}
		\vprincipal_H(\tau) &\le \vprincipal_H(\tau_{H,s,\iota^\sigma_\tau}) \\
		&=\sum_{(p,a) \in \mathcal{X}_\pi} \pi(p,a|\tau_{H,s,\iota^\sigma_\tau}) \sum_{s' \in \sset} P_H(s'|s,a) r^P_H(s,p,s') \\
		&= \sum_{(p,a) \in \mathcal{X}_\sigma} \varphi_H(p,a|s,\iota^\sigma_\tau) \sum_{s' \in \sset} P_H(s'|s,a) r^P_H(s,p,s') \\
		&= \vprincipalpr_H(s,\iota^\sigma_\tau) \\
		&= V^{\text{P},\pi^\sigma}_h(\tau),
	\end{align*}
	where the inequality holds due to the definition of $\tau_{H,s,\iota^\sigma_\tau}$ and the last equality one thanks to Lemma~\ref{lem:promise_functions_eq}.
	
	Now we consider $h < H$ and assume that the statement holds for $h' = h+1$.
	Then for every history $\tau \in \mathcal{T}^\pi_h$ ending in some state $s \in \mathcal{S}$ the following holds:
	\begin{equation}
		\begin{split}
		\vprincipal_h(\tau) &\le \vprincipal_h(\tau_{h,s,\iota^\sigma_\tau}) \\
		&=\sum_{(p,a) \in \mathcal{X}_\pi} \pi(p,a|\tau_{h,s,\iota^\sigma_\tau}) \qprincipal_h(\tau_{h,s,\iota^\sigma_\tau},p,a) \\
		&=\sum_{(p,a) \in \mathcal{X}_\pi} \varphi_h(p,a|s,\iota^\sigma_\tau) \qprincipal_h(\tau_{h,s,\iota^\sigma_\tau},p,a),
		\label{eq:promise_form_optimal_proof_vprincipal}			
		\end{split}
	\end{equation}
	where the inequality holds due to the definition of $\tau_{H,s,\iota^\sigma_\tau}$.
	Furthermore, for every contract $p \in \mathcal{P}$ and action $a \in \A$ such that $\pi(p,a|\tau)>0$, it holds that:
	\begin{align}
		\qprincipal_h(\tau_{h,s,\iota^\sigma_\tau},p,a) &= \sum_{s' \in \sset}P_h(s'|s,a) \left( r^\text{P}_h(s,p,s') + \vprincipal_{h+1}(\tau_{h,s,\iota^\sigma_\tau} \oplus (p,a,s')) \right) \nonumber \\
		&= \sum_{s' \in \sset}P_h(s'|s,a) \left( r^\text{P}_h(s,p,s') + g_h(s,\iota^\sigma_\tau,p,a,s') \right) \label{eq:promise_form_optimal_1}\\ 
		&\le \sum_{s' \in \sset}P_h(s'|s,a) \left( r^\text{P}_h(s,p,s') + V^{\text{P},\pi^\sigma}_{h+1}(\tau_{h+1,s',g_h(s,\iota^\sigma_\tau,p,a,s')}) \right) \label{eq:promise_form_optimal_3}\\
		&= \sum_{s' \in \sset}P_h(s'|s,a) \left( r^\text{P}_h(s,p,s') + \vprincipalpr_{h+1}(s',g_h(s,\iota^\sigma_\tau,p,a,s')) \right) \label{eq:promise_form_optimal_4}\\
		&= \qprincipalpr_h(s,\iota^\sigma_\tau) \nonumber
	\end{align}
	where Equation~\ref{eq:promise_form_optimal_1} holds due to the definition of $g_h(s,\iota^\sigma_\tau,p,a,s')$, 
	Equation~\ref{eq:promise_form_optimal_3} holds by induction, while Equation~\ref{eq:promise_form_optimal_4} follows from Lemma~\ref{lem:promise_functions_eq}.
	Thus, by combining this result with Equation~\ref{eq:promise_form_optimal_proof_vprincipal} we get:
	\begin{equation*}
		\vprincipal_h(\tau) \le \sum_{(p,a) \in \mathcal{X}_\pi} \varphi_h(p,a|s,\iota^\sigma_\tau) \qprincipalpr_h(s,\iota^\sigma_\tau) =\vprincipalpr_h(s,\iota^\sigma_\tau) = V^{\text{P},\pi^\sigma}_h(\tau),
	\end{equation*}
	where the last equality holds thanks to Lemma~\ref{lem:promise_functions_eq} and the definition of $\iota^\sigma_\tau$, concluding the proof.
	
	Since $V^{\text{P},\pi^\sigma}_h(\tau) \ge \vprincipal_h(\tau)$ for every step $h \in \mathcal{H}$ and history $\tau \in \mathcal{T}^\pi_h$, it follows that:
	\begin{align*}
		\vprincipal &= \sum_{s \in \sset} \mu(s) \vprincipal_1(s) \\
		&= \sum_{s \in \sset} \mu(s) V^{\text{P},\pi^\sigma}_1(s) = V^{\text{P},\pi^\sigma},
	\end{align*}
	as $\mathcal{T}^\pi_1 = \mathcal{T}_1 = \{(s) \mid s \in \sset\}$.
	
	Finally, we observe that the promise-form policy $\sigma^\star = \sigma^\star_{\pi^\star}$, where $\pi^\star$ is an optimal history-dependent policy, is optimal, as $V^{\text{P},\pi^{\sigma^\star}} \ge V^{\text{P},\pi^\star}$.
\end{proof}

\section{From $\epsilon$-IC to IC}\label{sec:appendix_ic}
\subsection{Relationship between agent's functions and value functions}
In order to prove Lemma~\ref{lem:from_ep_ic_to_pseudo_ic}, we first need to characterized the relationship between the agent's functions introduced in Section~\ref{sec:epsilon_ic_to_ic} and the value concepts described in Section~\ref{sec:preliminaries}.
In particular, we show that $V^{\text{P},\pi,\alpha} = V^{\text{P},\pi}$ and $V^{\text{A},\pi,\alpha} = V^{\text{A},\pi}$ when $\alpha$ is the recommended agent's function, where the cumulative agent's utility is defined as:
\begin{equation*}
	\vagent \coloneqq \coloneqq \sum_{s1 \in \sset} \mu(s1) \vagent_h(s1).
\end{equation*}
On the other hand,  $V^{\text{A},\pi,\alpha} = \widehat{V}^{\text{A},\pi}$ when $\alpha$ is an agent's function incentivized by $\pi$, where:
\begin{equation*}
	\vagentdev \coloneqq \sum_{s1 \in \sset} \mu(s1) \vagentdev_h(s1).
\end{equation*}
Finally, we show that if $\pi$ is $\epsilon$-IC, $\alpha$ is the recommended agent's function and $\widehat{\alpha}$ is any other agent's function, then $V^{\text{A},\pi,\alpha} \ge V^{\text{A},\pi,} - \epsilon$
These results are formalized in the following lemmas.

As the proofs involve multiple summations over histories $\tau \in \mathcal{T}'$, for the sake of  notation, whenever we denote a history as $\tau \in \mathcal{T}'$ in this section, we let $\tau \coloneqq (s_1,p_1,a_1,\dots,s_h,p_h,a_h,s_{h+1})$ for the appropriate $h \in \mathcal{H}$.


\begin{restatable}{lemma}{valpharec}
	\label{lem:value_rec_alpha}
	Let $\alpha$ be the recommended agent's function for some policy $\pi$.
	Then $V^{\textnormal{P},\pi} = V^{\textnormal{P},\pi,\alpha}$ and $V^{\textnormal{A},\pi} = V^{\textnormal{A},\pi,\alpha}$.
\end{restatable}
\begin{proof}
	We prove that $V^{\textnormal{P},\pi} = V^{\textnormal{P},\pi,\alpha}$. 
	A similar argument can be made for the agent's value function.
	
	In order to prove the statement, we show by induction that for any step $\bar{h} \in \mathcal{H}$ and history $\bar{\tau} \in \mathcal{T}_{\bar{h}}$ such that $\mathbb{P}(\bar{\tau}|\pi,\alpha)>0$ it holds that:
	\begin{equation*}
		\vprincipal_{\bar{h}}(\bar{\tau}) = \sum_{\substack{\tau \in \mathcal{T}'_{\pi,\alpha}, : \tau(s_{\bar{h}})=\bar{\tau} }} \mathbb{P}(\tau | \pi, \alpha, \bar{\tau}) \left(r_h(s_{h},s_{h+1}) - p_h(s_{h+1})\right),
	\end{equation*}
	where the summation is over the histories $\tau = (s_1,p_1,a_1,\dots,s_h,p_h,a_h,s_{h+1})$ that begin with $\bar{\tau}$ and contain $h>\bar{h}$ actions, while:
	\begin{equation*}
		\mathbb{P}(\tau | \pi, \alpha, \bar{\tau}) = \prod_{h'=\bar{h}}^{h} \pi(p_{h'},a_{h'}|\tau(s_{h'}))P_{h'}(s_{h'+1}|s_{h'},a_{h'}),
	\end{equation*}
	as $\alpha_{h'}(\tau) = a_{h'}$.
	
	The base step is $\bar{h} = H$. Given any history $\bar{\tau} = (s_1,p_1,a_1,\dots,s_{H-1},p_{H-1},a_{H-1},s_H)$, we have that:
	\begin{align*}
		\vprincipal_H(\bar{\tau}) &= \sum_{(p,a) \in \mathcal{X}_\pi} \pi(p,a|\bar{\tau}) \sum_{s' \in \sset} P_H(s'|s_H,a)\left( r_H(s_H,s') - p(s')\right) \\
		&= \sum_{(p,a) \in \mathcal{X}_\pi} \sum_{s' \in \sset} \pi(p,a|\bar{\tau})  P_H(s'|s_H,a)\left( r_H(s_H,s') - p(s')\right) \\
		&=\sum_{\substack{\tau=\bar{\tau} \oplus (p,a,s'),\\ (p,a,s') \in \mathcal{X}_\pi \times \sset}} \pi(p,a|\tau(s_H))P_H(s'|s_H,a)\left( r_H(s_H,s')-p(s')\right) \\
		&=\sum_{\substack{\tau\in \mathcal{T}' :\\ \tau(s_{\bar{h}})=\bar{\tau}, \\ \pi(p_h,a_h|\bar{\tau})>0}} \pi(p_h,a_h|\tau(s_h))P_h(s_{h+1}|s_h,a_h)\left( r_h(s_h,s_{h+1})-p_h(s_{h+1})\right) \\
		&=\sum_{\substack{\tau\in \mathcal{T}' :\\ \tau(s_{\bar{h}})=\bar{\tau}, \\ \mathbb{P}(\tau|\pi,\alpha,\bar{\tau})>0}} \mathbb{P}(\tau|\pi,\alpha,\bar{\tau})P_h(s_{h+1}|s_h,a_h)\left( r_h(s_h,s_{h+1})-p_h(s_{h+1})\right) \\
		&=\sum_{\substack{\tau \in \mathcal{T}'_{\pi,\alpha}, : \\ \tau(s_{\bar{h}})=\bar{\tau} }} \mathbb{P}(\tau | \pi, \alpha, \bar{\tau}) \left(r_h(s_{h},s_{h+1}) - p_h(s_{h+1})\right),
	\end{align*}
	where $h=H$.
	
	Now we consider $\bar{h} < H$ and assume that the inductive hypothesis holds for every $h' \ge \bar{h}$.
	Given any history $\bar{\tau} = (s_1,p_1,a_1,\dots,s_{\bar{h}-1},p_{\bar{h}-1},a_{\bar{h}-1},s_{\bar{h}})$ such that $\mathbb{P}(\bar{\tau}|\pi,\alpha)>0$, we have that:
	\begin{equation}
	\label{eq:value_as_histories_1}
	\begin{split}		
		&\sum_{(p,a) \in \mathcal{X}_\pi} \pi(p,a|\bar{\tau}) \sum_{s' \in \sset} P_{\bar{h}}(s'|s_{\bar{h}},a)\left( r_{\bar{h}}(s_{\bar{h}},s') - p(s')\right) \\
		&= \sum_{(p,a) \in \mathcal{X}_\pi} \sum_{s' \in \sset} \pi(p,a|\bar{\tau})  P_{\bar{h}}(s'|s_{\bar{h}},a)\left( r_{\bar{h}}(s_{\bar{h}},s') - p(s')\right) \\
		&=\sum_{\substack{\tau=\bar{\tau} \oplus (a,p,s'),\\ (a,p,s') \in \mathcal{X}_\pi \times \sset}} \pi(a|\tau(s_{\bar{h}}))P_{\bar{h}}(s'|s_{\bar{h}},a)\left( r_{\bar{h}}(s_{\bar{h}},s')-p(s')\right) \\
		&=\sum_{\substack{\tau\in \mathcal{T}_{\bar{h}+1} :\\ \tau(s_{\bar{h}})=\bar{\tau}, \\ \mathbb{P}(\tau|\pi,\alpha)>0}} \pi(a|\tau(s_h))P_h(s_{h+1}|s_h,a)\left( r_h(s_h,s_{h+1})-p_h(s_{h+1})\right) \\
		&=\sum_{\substack{\tau\in \mathcal{T}_{\bar{h}+1} :\\ \tau(s_{\bar{h}})=\bar{\tau}, \\ \mathbb{P}(\tau|\pi,\alpha)>0}} \mathbb{P}(\tau|\pi, \alpha, \bar{\tau})\left( r_h(s_h,s_{h+1})-p_h(s_{h+1})\right).
	\end{split}
	\end{equation}
	We also observe that:
	\begin{equation}
	\label{eq:value_as_histories_2}
	\begin{split}
		&\sum_{(p,a) \in \mathcal{X}_\pi} \pi(p,a|\bar{\tau}) \sum_{s' \in \sset} P_{\bar{h}}(s'|s_{\bar{h}},a)\vprincipal_{\bar{h}+1}(\bar{\tau} \oplus (p,a,s')) \\
		&=\hspace{-0.3cm} \sum_{(p,a) \in \mathcal{X}_\pi} \hspace{-0.4cm} \pi(p,a|\bar{\tau}) \sum_{s' \in \sset} P_{\bar{h}}(s'|s_{\bar{h}},a) \hspace{-0.7cm} \sum_{\substack{\tau \in \mathcal{T}'_{\pi,\alpha} : \\ \tau(s_{\bar{h}})=\bar{\tau} \oplus (p,a,s') }} \hspace{-0.7cm} \mathbb{P}(\tau | \pi, \alpha, \bar{\tau} \oplus (p,a,s')) \left(r_h(s_{h},s_{h+1}) - p_h(s_{h+1})\right) \\
		&=\sum_{\substack{\tau' = \bar{\tau} \oplus (p,a,s'), \\(p,a,s') \in \mathcal{X}_\pi \times \sset}} \mathbb{P}(\tau' | \pi, \alpha, \bar{\tau})  \hspace{-0.5cm} \sum_{\substack{\tau \in \mathcal{T}'_{\pi,\alpha} : \\ \tau(s_{\bar{h}})=\bar{\tau} \oplus (p,a,s')}}  \hspace{-0.5cm} \mathbb{P}(\tau | \pi, \alpha, \bar{\tau} \oplus (p,a,s')) \left(r_h(s_{h},s_{h+1}) - p_h(s_{h+1})\right) \\
		&=\sum_{\substack{\tau' = \bar{\tau} \oplus (p,a,s'), \\(p,a,s') \in \mathcal{X}_\pi \times \sset}} \mathbb{P}(\tau' | \pi, \alpha, \bar{\tau}) \sum_{\substack{\tau \in \mathcal{T}'_{\pi,\alpha} : \\ \tau(s_{\bar{h}})=\tau'}} \mathbb{P}(\tau | \pi, \alpha, \tau') \left(r_h(s_{h},s_{h+1}) - p_h(s_{h+1})\right) \\
		&=\sum_{\substack{\tau \in \mathcal{T}'_{\pi,\alpha} : h \ge \bar{h}+1 \\ \tau(s_{\bar{h}})=\bar{\tau}}}  \hspace{-0.3cm} \mathbb{P}(\tau | \pi, \alpha, \bar{\tau}) \left(r_h(s_{h},s_{h+1}) - p_h(s_{h+1})\right),
	\end{split}
	\end{equation}
	where the summation in the last equality is over the histories that begin with $\bar{\tau}$ and have at least two transactions more (we recall that the last transition of $\tau$ is from state $s_h$ to state $s_{h+1}$).	
	Thus, by combining Equation~\ref{eq:value_as_histories_1} and Equation~\ref{eq:value_as_histories_2} we get:
	\begin{align*}
		\vprincipal_{\bar{h}}(\bar{\tau}) &= \sum_{(p,a) \in \mathcal{X}_\pi} \pi(p,a|\tau) \sum_{s' \in \sset} P_{\bar{h}}(s'|s_{\bar{h}},a) \left( r_{\bar{h}}(s_{\bar{h}},s') - p(s') + \vprincipal_{\bar{h}+1}(\bar{\tau} \oplus (p,a,s'))\right) \\
		&=\sum_{\substack{\tau\in \mathcal{T}_{\bar{h}+1} :\\ \tau(s_{\bar{h}})=\bar{\tau}, \mathbb{P}(\tau|\pi,\alpha)>0}} \mathbb{P}(\tau|\pi, \alpha, \bar{\tau})\left( r_h(s_h,s_{h+1})-p_h(s_{h+1})\right) \\
		&\hspace{3cm} + \hspace{-0.3cm} \sum_{\substack{\tau \in \mathcal{T}'_{\pi,\alpha} : h \ge \bar{h}+1 \\ \tau(s_{\bar{h}})=\bar{\tau}}}  \hspace{-0.3cm} \mathbb{P}(\tau | \pi, \alpha, \bar{\tau}) \left(r_h(s_{h},s_{h+1}) - p_h(s_{h+1})\right) \\
		&=\sum_{\tau \in \mathcal{T}'_{\pi,\alpha} : \tau(s_{\bar{h}})=\bar{\tau}} \mathbb{P}(\tau | \pi, \alpha, \bar{\tau}) \left(r_h(s_{h},s_{h+1}) - p_h(s_{h+1})\right).
	\end{align*}
		
	Given this result, we can prove the lemma as follows:
	\begin{align*}
		\vprincipal &= \sum_{s \in \sset} \mu(s) \vprincipal_1(s) \\
		&=\sum_{s \in \sset} \mu(s) \sum_{\tau \in \mathcal{T}'_{\pi,\alpha} : s_1 = s} \mathbb{P}(\tau | \pi, \alpha, s) \left(r_h(s_{h},s_{h+1}) - p_h(s_{h+1})\right) \\
		&=\sum_{\tau \in \mathcal{T}'_{\pi,\alpha}} \mathbb{P}(\tau | \pi, \alpha) \left(r_h(s_{h},s_{h+1}) - p_h(s_{h+1})\right) = V^{\text{P},\pi,\alpha},
	\end{align*}
	concluding the proof.
\end{proof}


\begin{restatable}{lemma}{things}
	\label{lem:value_rec_alpha_at_most_dev}
	Let $\alpha$ be any agent's function and consider a policy $\pi$.
	Then it holds that:
	\begin{equation*}
		V^{\text{A},\pi,\alpha} \le \vagentdev.
	\end{equation*}
	Furthermore, when $\alpha$ is the incentivized agent's function the inequality above holds as equality.
\end{restatable}
\begin{proof}
	In order to prove the statement, we show by induction that for any step $\bar{h} \in \mathcal{H}$ and history $\bar{\tau} \in \mathcal{T}_{\bar{h}}$ such that $\mathbb{P}(\bar{\tau}|\pi,\alpha)>0$ it holds that:
	\begin{equation*}
		\vagentdev_{\bar{h}}(\bar{\tau}) \ge \sum_{\tau \in \mathcal{T}'_{\pi,\alpha} : \tau(s_{\bar{h}})=\bar{\tau}} \mathbb{P}(\tau | \pi, \alpha, \bar{\tau}) \left(p_h(s_{h+1}) - c_h(s_h,\alpha_h(\tau))\right),
	\end{equation*}
	where the summation is over the histories $\tau = (s_1,p_1,a_1,\dots,s_{\bar{h}},p_{\bar{h}},a_{\bar{h}},\dots,s_h,p_h,a_h,s_{h+1})$ that begin with $\bar{\tau}$ and have at least a transition more, while:
	\begin{equation*}
		\mathbb{P}(\tau | \pi, \alpha, \bar{\tau}) = \prod_{h'=\bar{h}}^{h} \pi(p_{h'},a_{h'}|\tau(s_{h'}))P_{h'}(s_{h'+1}|s_{h'},\alpha_{h'}(\tau)).
	\end{equation*}
	
	The base step is $\bar{h} = H$. 
	Given any history $\bar{\tau} = (s_1,p_1,a_1,\dots,s_{H-1},p_{H-1},a_{H-1},s_H)$ that happens with probability larger than zero, we have that:
	\begin{align*}
		\vagentdev_H(\bar{\tau}) &= \sum_{(p,a) \in \mathcal{X}_\pi} \pi(p,a|\bar{\tau}) \max_{\widehat{a} \in \mathcal{A}} \sum_{s' \in \sset} P_H(s'|s_H,\widehat{a})\left( p(s') - c_H(s_H,\widehat{a})\right) \\
		&\ge \hspace{-0.3cm} \sum_{(p,a) \in \mathcal{X}_\pi} \hspace{-0.3cm} \pi(p,a|\bar{\tau}) \sum_{s' \in \sset} P_H(s'|s_H,\alpha_H(\bar{\tau} \oplus (p,a,s')))\left( p(s') - c_H(s_H,\alpha_H(\bar{\tau} \oplus (p,a,s')))\right) \\
		&=\sum_{\substack{\tau=\bar{\tau} \oplus (p,a,s'),\\ (p,a,s') \in \mathcal{X}_\pi \times \sset}} \pi(p,a|\tau(s_H))P_H(s'|s_H,\alpha_H(\tau))\left( p_H(s') - c_H(s_H,\alpha_H(\tau))\right) \\
		&=\sum_{\substack{\tau\in \mathcal{T}'_{\pi,\alpha} :\\ \tau(s_H)=\bar{\tau}}} \pi(p_h,a_h|\tau(s_h))P_h(s_{h+1}|s_h,\alpha_h(\tau))\left( p_h(s_{h+1}) - c_h(s_h,\alpha_h(\tau)) \right) \\
		&= \sum_{\substack{\tau\in \mathcal{T}'_{\pi,\alpha} :\\ \tau(s_H)=\bar{\tau}}} \mathbb{P}(\tau|\pi,\alpha,\bar{\tau}) \left( p_h(s_{h+1}) - c_h(s_h,\alpha_h(\tau)) \right),
	\end{align*}
	where $h=H$.
	
	Now we consider $\bar{h} < H$ and assume that the inductive hypothesis holds for $h' \ge \bar{h}$.
	Given any history $\bar{\tau} = (s_1,p_1,a_1,\dots,s_{\bar{h}-1},p_{\bar{h}-1},a_{\bar{h}-1},s_{\bar{h}})$ with an argument similar to the one above we can prove that:
	\begin{equation}
		\label{eq:dev_value_as_histories_1}
		\begin{split}		
		&\sum_{(p,a) \in \mathcal{X}_\pi} \pi(p,a|\bar{\tau}) \sum_{s' \in \sset}  P_{\bar{h}}(s'|s_{\bar{h}},\alpha_{\bar{h}}(\bar{\tau} \oplus (p,a,s'))) \left( p(s') -c_h{\bar{h}}(s_{\bar{h}},\alpha_{\bar{h}}(\bar{\tau} \oplus (p,a,s'))) \right) \\
		&= \sum_{\substack{\tau \in \mathcal{T}_{\bar{h}+1} : \\ \tau(s_{\bar{h}}) = \bar{\tau}, \\ \mathbb{P}(\tau|\pi,\alpha)>0}} \mathbb{P}(\tau | \pi, \alpha, \bar{\tau}) \left( p_h(s_{h+1}) - c_h(s_h,\alpha_h(\tau)) \right)
		\end{split}
	\end{equation}
	Moreover, with an argument similar to the one adopted in the proof of Lemma~\ref{lem:value_rec_alpha}, we can show that:
	\begin{equation}
		\label{eq:dev_value_as_histories_2}
		\begin{split}
			&\sum_{(p,a) \in \mathcal{X}_\pi} \pi(p,a|\bar{\tau}) \sum_{s' \in \sset}  P_{\bar{h}}(s'|s_{\bar{h}},\alpha_{\bar{h}}(\bar{\tau} \oplus (p,a,s')))\vagentdev_{\bar{h}+1}(\bar{\tau} \oplus (p,a,s')) \\
			&=\sum_{\substack{\tau \in \mathcal{T}'_{\pi,\alpha} : h \ge \bar{h}+1 \\ \tau(s_{\bar{h}})=\bar{\tau}}} \hspace{-0.3cm} \mathbb{P}(\tau | \pi, \alpha, \bar{\tau}) \left( p_h(s_{h+1}) -c_h(s_{h},\alpha_h(\tau)) \right),
		\end{split}
	\end{equation}
	where the summation in the last equality is over the histories that begin with $\bar{\tau}$ and have at least two transactions more.	
	Thus, by combining Equation~\ref{eq:dev_value_as_histories_1} and Equation~\ref{eq:dev_value_as_histories_2} we get:
	\begin{align*}
		&\vagentdev_{\bar{h}}(\bar{\tau}) \\
		&= \sum_{(p,a) \in \mathcal{X}_\pi} \pi(p,a|\bar{\tau})\max_{\widehat{a} \in A} \sum_{s' \in \sset}  P_{\bar{h}}(s'|s_{\bar{h}},\widehat{a}) \left( p(s') -c_{\bar{h}}(s_{\bar{h}},\widehat{a}) + \vagentdev_{\bar{h}+1}(\bar{\tau} \oplus (p,a,s'))\right) \\
		&\ge \sum_{(p,a) \in \mathcal{X}_\pi} \pi(p,a|\bar{\tau}) \sum_{s' \in \sset}  P_{\bar{h}}(s'|s_{\bar{h}},\alpha_{\bar{h}}(\tau')) \left( p(s') -c_{\bar{h}}(s_{\bar{h}},\alpha_{\bar{h}}(\tau')) + \vagentdev_{\bar{h}+1}(\tau')\right) \\
		&= \hspace{-0.35cm} \sum_{\substack{\tau \in \mathcal{T}_{\bar{h}+1} : \\ \tau(s_{\bar{h}}) = \bar{\tau}, \\ \mathbb{P}(\tau|\pi,\alpha)>0}} \hspace{-0.35cm} \mathbb{P}(\tau | \pi, \alpha, \bar{\tau}) \left( p_h(s_{h+1} - c_h(s_h,\alpha_h(\tau))) \right) + \hspace{-0.55cm} \sum_{\substack{\tau \in \mathcal{T}' : h \ge \bar{h}+1 \\ \tau(s_{\bar{h}})=\bar{\tau}}}  \hspace{-0.55cm} \mathbb{P}(\tau | \pi, \alpha, \bar{\tau}) \left( p_h(s_{h+1}) -c_h(s_{h},\alpha_h(\tau)) \right) \\
		&= \sum_{\substack{\tau \in \mathcal{T}'_{\pi,\alpha} : \\ \tau(s_{\bar{h}}) = \bar{\tau}}} \mathbb{P}(\tau | \pi, \alpha, \bar{\tau}) \left( p_h(s_{h+1} - c_h(s_h,\alpha_h(\tau))) \right) = V^{\text{A},\pi,\alpha},
	\end{align*}
	where in the inequality $\tau' = \bar{\tau} \oplus (p,a,s')$ and the second equality holds thanks to Equation~\ref{eq:dev_value_as_histories_1} and Equation~\ref{eq:dev_value_as_histories_2}.
	
	Finally, we recall that the incentivized agent's function is such that:
	\begin{equation*}
		\widehat{\alpha}_{h'}(\tau) \in \argmax_{\widehat{a} \in \mathcal{A}} \sum_{s' \in \sset}P_h(s'|s_{h'},\widehat{a}) \left( p_{h'}(s') - c_{h'}(s_{h'},\widehat{a}) +\widehat{V}^{\text{A},\pi}_{h+1}(\tau(s_{h'}) \oplus (p_{h'},a_{h'},s')) \right)
	\end{equation*} 
	for every history $\tau$ and step $h' \le h$.
	One can easily verify that for $\alpha=\widehat{\alpha}$, the inequalities above hold as equalities, concluding the proof.
\end{proof}

\begin{restatable}{lemma}{valphaepic}
	\label{lem:value_alpha_epsilon_ic}
	Let $\alpha$ be the recommended agent's function for some $\epsilon$-IC policy $\pi$.
	Then it holds that:
	\begin{equation*}
		\sum_{\tau \in \mathcal{T}'} \mathbb{P}(\tau|\pi,\alpha) \left(p^{\pi,\tau}(s_{h+1})-c_h(s_h,\alpha_h(\tau))\right) \geq \sum_{\tau \in \mathcal{T}'} \mathbb{P}(\tau|\pi,\widehat{\alpha}) \left(p^{\pi,\tau}(s_{h+1})-c_h(s_h,\widehat{\alpha}_h(\tau))\right) -\epsilon
	\end{equation*}
	for every agent's function $\widehat{a}$.
\end{restatable}
\begin{proof}	
	As of Lemma~\ref{lem:value_rec_alpha}, $V^{\text{A},\pi,\alpha}$ is equal to the agent's value, formally:
	\begin{equation*}
		 \sum_{s \in \sset} \mu(s) V^{\text{A},\pi}_1(s) = \sum_{\tau \in \mathcal{T}'_{\pi,\alpha}} \mathbb{P}(\tau|\pi,\alpha) \left(p_h(s_{h+1})-c_h(s_h,\alpha_h(\tau))\right).
	\end{equation*}
	Furthermore, by Lemma~\ref{lem:value_rec_alpha_at_most_dev}, it holds that:
	\begin{equation*}
		\sum_{\tau \in \mathcal{T}'_{\pi,\widehat{\alpha}}} \mathbb{P}(\tau|\pi,\widehat{\alpha}) \left(p_h(s_{h+1})-c_h(s_h,\widehat{\alpha}_h(\tau))\right) \le  \sum_{s \in \sset} \mu(s) \vagentdev_1(s).
	\end{equation*}
	Thus, since $\pi$ is $\epsilon$-IC it follows that:
	\begin{align*}
		&\sum_{\tau \in \mathcal{T}'_{\pi,\alpha}} \mathbb{P}(\tau|\pi,\alpha) \left(p_h(s_{h+1})-c_h(s_h,\alpha_h(\tau))\right) \\
		&= \sum_{s \in \sset} \mu(s) \vagent_1(s) \\
		&\ge \sum_{s \in \sset} \mu(s) \vagentdev_1(s) -\epsilon \\
		&\ge \sum_{\tau \in \mathcal{T}'_{\pi,\widehat{\alpha}}} \mathbb{P}(\tau|\pi,\widehat{\alpha}) \left(p_h(s_{h+1})-c_h(s_h,\widehat{\alpha}_h(\tau))\right) -\epsilon,
	\end{align*}
	concluding the proof.
\end{proof}

As a further step, we also want to write the value $V^{\text{P},\pi,\alpha}$ for some incentivized agent's function $\alpha$ in a recursive fashion.
With an abuse of notation, given a policy $\pi$ and agent's function $\alpha$ we let $V^{\text{P},\pi,\alpha}_h : \mathcal{T}_h \rightarrow \mathbb{R}$ the function defined as:
\begin{equation*}
	V^{\text{P},\pi,\alpha}_h(\tau') = \hspace{-0.35cm} \sum_{(p,a) \in \mathcal{X}_\pi} \hspace{-0.35cm} \pi(p,a|\tau') \hspace{-0.1cm} \sum_{s' \in \sset} \hspace{-0.1cm} P_h(s'|s,\alpha_h(\tau' \oplus (p,a,s'))) \hspace{-0.1cm} \left( r^{\text{P}}_h(s,p,s') + V^{\text{P},\pi,\alpha}_{h+1}(\tau' \oplus (p,a,s'))\right).
\end{equation*}
Intuitively, the function $V^{\text{P},\pi,\alpha}_h$ is similar to the function $V^{\text{P},\pi}_h$ defined in Section~\ref{sec:preliminaries}, but employs the agent's function $\alpha$ to determine the actions of the agent, rather than assuming that they follow the recommendations.
This function can be used to compute the value for the principal $V^{\text{P},\alpha}$ as follows:
\begin{equation*}
	V^{\text{P}, \pi, \alpha} = \sum_{\tau \in \mathcal{T}'_{\pi,\alpha}} \mathbb{P}(\tau|\pi,\alpha) \left( r_h(s_h,s_{h+1}) - p_h(s_{h+1}) \right)  \sum_{s \in \sset} \mu(s) V^{\text{P},\pi,\alpha}_1(s).
\end{equation*}
The equation above is formalized in the following Lemma:
\begin{restatable}{lemma}{recursiveprincipaldev}
	\label{lem:recursive_principal_dev}
	Given any policy $\pi$ and agent's function $\alpha$, it holds that $V^{\text{P}, \pi, \alpha} = \sum_{s \in \sset} \mu(s) V^{\text{P},\pi,\alpha}_1(s)$.
\end{restatable}
\begin{proof}
	In order to prove the statement, we show by induction that for any step $\bar{h} \in \mathcal{H}$ and history $\bar{\tau} \in \mathcal{T}_{\bar{h}}$ it holds that:
	\begin{equation*}
		V^{\text{P},\pi,\alpha}_{\bar{h}}(\bar{\tau}) = \sum_{\substack{\tau \in \mathcal{T}'_{\pi,\alpha} : \\ \tau(s_{\bar{h}}) = \bar{\tau}}} \mathbb{P}(\tau|\pi,\alpha,\bar{\tau}) \left( r_h(s_h,s_{h+1}) - p_h(s_{h+1}) \right),
	\end{equation*}
	where the summation is over the histories $\tau = (s_1,p_1,a_1,\dots,s_h,p_h,a_h,s_{h+1})$ that begin with $\bar{\tau}$ and contain $h>\bar{h}$ actions.
	
	The base step is $\bar{h} = H$. Given any history $\bar{\tau} = (s_1,p_1,a_1,\dots,s_{H-1},p_{H-1},a_{H-1},s_H)$, we have that:
	\begin{align*}
		V^{\text{P},\pi,\alpha}_H(\bar{\tau}) &= \sum_{(p,a) \in \mathcal{X}_\pi} \pi(p,a|\bar{\tau}) \sum_{s' \in \sset} P_H(s'|s_H,\alpha_H(\bar{\tau} \oplus (p,a,s')))\left( r_H(s_H,s') - p(s')\right) \\
		&= \sum_{(p,a) \in \mathcal{X}_\pi} \sum_{s' \in \sset} \pi(p,a|\bar{\tau})  P_H(s'|s_H,\alpha_H(\bar{\tau} \oplus (p,a,s')))\left( r_H(s_H,s') - p(s')\right) \\
		&=\sum_{\substack{\tau=\bar{\tau} \oplus (p,a,s'),\\ (p,a,s') \in \mathcal{X}_\pi \times \sset}} \pi(p,a|\tau(s_H))P_H(s'|s_H,\alpha_H(\tau))\left( r_H(s_H,s')-p(s')\right) \\
		&=\sum_{\substack{\tau\in \mathcal{T}' :\\ \tau(s_{\bar{h}})=\bar{\tau}, \\ \pi(p_h,a_h|\bar{\tau})>0}} \pi(p_h,a_h|\tau(s_h))P_h(s_{h+1}|s_h,\alpha_h(\tau))\left( r_h(s_h,s_{h+1})-p_h(s_{h+1})\right) \\
		&=\sum_{\substack{\tau\in \mathcal{T}' :\\ \tau(s_{\bar{h}})=\bar{\tau}, \\ \mathbb{P}(\tau|\pi,\alpha,\bar{\tau})>0}} \mathbb{P}(\tau|\pi,\alpha,\bar{\tau})P_h(s_{h+1}|s_h,\alpha_h(\tau))\left( r_h(s_h,s_{h+1})-p_h(s_{h+1})\right) \\
		&=\sum_{\substack{\tau \in \mathcal{T}'_{\pi,\alpha}, : \\ \tau(s_{\bar{h}})=\bar{\tau} }} \mathbb{P}(\tau | \pi, \alpha, \bar{\tau}) \left(r_h(s_{h},s_{h+1}) - p_h(s_{h+1})\right),
	\end{align*}
	where $h=H$.
	
	Now we consider $\bar{h} < H$ and assume that the inductive hypothesis holds for every $h' \ge \bar{h}$.
	Let $\bar{\tau} = (s_1,p_1,a_1,\dots,s_{\bar{h}-1},p_{\bar{h}-1},a_{\bar{h}-1},s_{\bar{h}})$ be an history such that $\mathbb{P}(\bar{\tau}|\pi,\alpha)>0$.
	An argument similar to the one employed for the base inductive step shows that::
	\begin{equation}
		\label{eq:p_dev_value_as_histories_1}
		\begin{split}		
			&\sum_{(p,a) \in \mathcal{X}_\pi} \pi(p,a|\bar{\tau}) \sum_{s' \in \sset} P_{\bar{h}}(s'|s_{\bar{h}},\alpha_{\bar{h}}(\bar{\tau} \oplus (p,a,s')))\left( r_{\bar{h}}(s_{\bar{h}},s') - p(s')\right) \\
			&=\sum_{\substack{\tau\in \mathcal{T}_{\bar{h}+1} :\\ \tau(s_{\bar{h}})=\bar{\tau}, \\ \mathbb{P}(\tau|\pi,\alpha)>0}} \mathbb{P}(\tau|\pi, \alpha, \bar{\tau})\left( r_h(s_h,s_{h+1})-p_h(s_{h+1})\right).
		\end{split}
	\end{equation}
	Furthermore, an argument similar to the one employed by Lemma~\ref{lem:value_rec_alpha} proves that:
	\begin{equation}
		\label{eq:p_dev_value_as_histories_2}
		\begin{split}
			&\sum_{(p,a) \in \mathcal{X}_\pi} \pi(p,a|\bar{\tau}) \sum_{s' \in \sset} P_{\bar{h}}(s'|s_{\bar{h}},\alpha_{\bar{h}}(\bar{\tau} \oplus (p,a,s')))\vprincipal_{\bar{h}+1}(\bar{\tau} \oplus (p,a,s')) \\
			&=\sum_{\substack{\tau \in \mathcal{T}'_{\pi,\alpha} : h \ge \bar{h}+1 \\ \tau(s_{\bar{h}})=\bar{\tau}}}  \hspace{-0.3cm} \mathbb{P}(\tau | \pi, \alpha, \bar{\tau}) \left(r_h(s_{h},s_{h+1}) - p_h(s_{h+1})\right),
		\end{split}
	\end{equation}
	where the summation in the last equality is over the histories that begin with $\bar{\tau}$ and have at least two transactions more (we recall that the last transition of $\tau$ is from state $s_h$ to state $s_{h+1}$).	
	Thus, by combining Equation~\ref{eq:value_as_histories_1} and Equation~\ref{eq:value_as_histories_2}, one can easily verify that:
	\begin{align*}
		\vprincipal_{\bar{h}}(\bar{\tau}) &= \sum_{\tau \in \mathcal{T}'_{\pi,\alpha} : \tau(s_{\bar{h}})=\bar{\tau}} \mathbb{P}(\tau | \pi, \alpha, \bar{\tau}) \left(r_h(s_{h},s_{h+1}) - p_h(s_{h+1})\right).
	\end{align*}
	
	Given this result, we can prove the lemma as follows:
	\begin{align*}
		V^{\text{P},\pi,\alpha} &=\sum_{\tau \in \mathcal{T}'_{\pi,\alpha}} \mathbb{P}(\tau | \pi, \alpha) \left(r_h(s_{h},s_{h+1}) - p_h(s_{h+1})\right) \\
		&=\sum_{s \in \sset} \mu(s) \sum_{\tau \in \mathcal{T}'_{\pi,\alpha} : s_1 = s} \mathbb{P}(\tau | \pi, \alpha, s) \left(r_h(s_{h},s_{h+1}) - p_h(s_{h+1})\right) \\
		&=\sum_{s \in \sset} \mu(s) V^{\text{P},\pi,\alpha}_1(s),
	\end{align*}
	concluding the proof.
\end{proof}

\subsection{Subprocedures of Algorithm~\ref{alg:from_ep_ic_to_ic}}
\label{appendix:subprocedures_ep_ic_to_ic}
Algorithm~\ref{alg:from_ep_ic_to_ic} is divided in three subprocedures that we describe in this section.
The first procedure is $\texttt{Change-Contracts}$ (Algorithm~\ref{alg:from_ep_ic_to_ic_phase_1}).
It takes in input an $\epsilon$-IC promise-form policy $\sigma$ and computes a new promise-form policy $\sigma^1$ that achieves a ``good'' principal's utility when the agent best respond to it.
The algorithm directly applies Lemma~\ref{lem:from_ep_ic_to_pseudo_ic}, changing the contracts of the policy.

\begin{algorithm}[H]
	\caption{\texttt{Change-Contracts}}\label{alg:from_ep_ic_to_ic_phase_1}
	\begin{algorithmic}[1]
		\Require $\epsilon$-IC direct policy $\sigma = \{(I_h,J_h,\varphi_h,g_h)\}_{h \in \mathcal{H}}$.
		\For{$h=H,\dots,1$}\label{line:conversion_block_1}
		\State $I^1_h \gets I_h$
		\ForAll{$s \in \sset, \iota \in I_h(s)$}
		\ForAll{$a \in \mathcal{A}, p \in J_h(s,\iota,a)$}
		\State $p'(s') \gets (1-\sqrt{\epsilon})p(s') + \sqrt{\epsilon}r_h(s,s') \quad \forall s' \in \sset$
		\State $\strut \varphi^1_h(p',a|s,\iota) \gets \varphi_h(p,a|s,\iota)$
		\State $\strut J^1_h(s,\iota,a) \gets \{p'\}$
		\State $g^1_h(s,\iota,p',a,s') \gets g^2_h(s,\iota,p,a,s') \quad \forall s' \in \sset$ 
		\EndFor
		\EndFor
		\EndFor

		\State \textbf{Return} $\sigma^1 = \{I^1_h, J^1_h, \varphi^1_h, g^1_h\}_{h \in \mathcal{H}}$.
	\end{algorithmic}
\end{algorithm}

As a further step, Algorithm~\ref{alg:from_ep_ic_to_ic} employs the procedure $\texttt{Realign-actions}$ (Algorithm~\ref{alg:from_ep_ic_to_ic_phase_2}) to compute a new IC policy.
Algorithm~\ref{alg:from_ep_ic_to_ic_phase_2} takes in input a promise-form policy $\sigma^1$ and builds a new policy $\sigma^2$ by that proposes the same contracts of $\sigma^1$, but recommend IC actions.
This way, the value of $\sigma^2$ becomes $V^{\text{P},\sigma^2} = V^{\text{P},\sigma^1,\alpha}$ for some agent's function incentivized by $\sigma^1$.
In particular, the two policies $\sigma^1$ and $\sigma^2$ share the same sets of promises.
Algorithm~\ref{alg:from_ep_ic_to_ic_phase_2} computes the actual best response $a'$ for every contract proposed by the policy $\sigma^1$ (Line~\ref{line:phase_2_aprime_argmax}), and builds $\sigma^2$ by prescribing the same contract and recommending $a'$.
We remark that while the best way to break ties at Line~\ref{line:phase_2_aprime_argmax} Algorithm~\ref{alg:from_ep_ic_to_ic_phase_2} is to choose the action that maximizes the principal's cumulative utility, the theoretical gauntness of our Algorithm hold regardless of how these ties are broken.

Algorithm~\ref{alg:from_ep_ic_to_ic_phase_2} employs dynamic programming to compute the values of the function $\widehat{V}^{\text{A},\sigma^1}$.
In particular, at Line~\ref{line:phase_2_dynamic_prog} it computes and stores the value of $\widehat{V}^{\text{A},\sigma^1}_{h}(s,\iota)$ by using the values of $\widehat{V}^{\text{A},\sigma^1}_{h+1}$ computed at the previous time step.
Furthermore, the values of $\widehat{V}^{\text{A},\sigma^1}_{h+1}$ are employed at Line~\ref{line:phase_2_aprime_argmax}.

\begin{algorithm}[H]
	\caption{\texttt{Realign-actions}}\label{alg:from_ep_ic_to_ic_phase_2}
	\begin{algorithmic}[1]
		\Require promise-form policy $\sigma^1 = \{(I^1_h,J^1_h,\varphi^1_h,g^1_h)\}_{h \in \mathcal{H}}$.
		\For{$h=H,\dots,1$}
		\State $I^2_h \gets I^1_h$
		\ForAll{$s \in \sset, \iota \in I^1_h(s)$}
		\State $J^2_h(s,\iota,a) \gets \emptyset \quad \forall a \in \mathcal{A}$
		\State Compute and store $\widehat{V}^{\text{A},\sigma^1}_h(s,\iota)$ \label{line:phase_2_dynamic_prog}
		\ForAll{$a \in \mathcal{A}, p \in J^1_h(s,\iota,a)$}
		\State $a' \gets \argmax_{\widehat{a} \in \mathcal{A}} \sum_{s' \in \sset}P_h(s'|s,\widehat{a})\left( r^\text{A}_h(s,p,\widehat{a},s') + \widehat{V}^{\text{A},\sigma^1}_{h+1}(s',g^1_h(s,\iota,p,a,s')) \right)$ \label{line:phase_2_aprime_argmax}
		\State $\strut \varphi^2_h(p,a'| s, \iota) \gets \varphi^1_h(p,a | s,\iota)$ \label{line:phase_2_phi_assignment}
		\State $\strut J^2_h(s,\iota,a') \gets J^2_h(s,\iota,a') \cup \{p\}$
		\ForAll{$s' \in \sset$}
		\State $g^2_h(s,\iota,p,a',s') \gets g^1_h(s,\iota,p,a,s')$ \label{line:phase_2_g_assignment}
		\EndFor
		\EndFor
		\EndFor
		\EndFor
		
		\State \textbf{Return} $\sigma^2 = \{I^2_h, J^2_h, \varphi^2_h, g^2_h\}_{h \in \mathcal{H}}$
	\end{algorithmic}
\end{algorithm}

Thanks to Lemma~\ref{lem:recursive_principal_dev}, we can prove the two main guarantees of Algorithm~\ref{alg:from_ep_ic_to_ic_phase_2}.
The first is that the value of the policy $\sigma^2$ is at least $V^{\text{P},\pi^{\sigma^2}} \ge V^{\text{P},\pi^{\sigma^1},\alpha}$ for some agent's function $\alpha$ incentivized by $\sigma^2$.
The actual incentivized function depend on how ties are broken at Line~\ref{line:phase_2_aprime_argmax} Algorithm~\ref{alg:from_ep_ic_to_ic_phase_2}.
The second important guarantee is that the policy $\sigma^2$ is IC.
These properties are formalized in the following lemma:

\begin{restatable}{lemma}{fromepictoicphase2}
	\label{lem:from_ep_ic_to_ic_phase_2}
	Let $\sigma^1$ be a promise-form policy.
	Then Algorithm~\ref{alg:from_ep_ic_to_ic_phase_2} returns an IC promise-form policy $\sigma^2$ with value $V^{\text{P},\sigma^2} = V^{\text{P},\sigma^1,\alpha}$ for some agent's function $\alpha$ incentivized by $\sigma^1$.
	Furthermore, Algorithm~\ref{alg:from_ep_ic_to_ic_phase_2} runs in time polynomial in the size of $\sigma^1$ and the instance, while $|\mathcal{I}^2| = |\mathcal{I}^1|$ and $|\mathcal{J}^2| = |\mathcal{J}^1|$.
\end{restatable}
\begin{proof}
	In the following we let $\pi^1 = \pi^{\sigma^1}$ and $\pi^2 = \pi^{\sigma^2}$.
	We first prove that $V^{\text{P},\sigma^2} = V^{\text{P},\sigma^1,\alpha}$, then we will show that $\sigma^2$ is IC.
	
	We observe that the action $a'=a'(s,\iota,p,a)$ computed at Line~\ref{line:phase_2_aprime_argmax} maximizes the agent's cumulative utility from state $s \in \sset$ and promise $\iota \in I_h(s)$ when the principal plays the policy $\sigma^1$.
	Thus, there exists an agent's function $\alpha$ incentivized by $\pi^1$ such that $\alpha_h(\tau) = a'(s_h,\iota^{\sigma^1}_{\tau(s_h)},p_h,a_h)$ for every history $\tau = (s_1,p_1,a_1,\dots,s_{h},p_{h},a_{h},s_{h+1})$.
	
	Thanks to the particular structure of $\alpha$, the actions played according to $\alpha$ depend only on the promises and states.
	Let us define, with an abuse of notation, the function $V^{\text{P},\sigma^1,\alpha}_h : \sset \times \mathcal{I}^1 \rightarrow \mathbb{R}$ as:
	\begin{equation*}
		V^{\text{P},\sigma^1,\alpha}_h (s,\iota) \coloneqq \hspace{-0.5cm} \sum_{(p,a) \in \mathcal{X}_{\sigma^1}} \hspace{-0.5cm} \varphi^1_h(p,a|s,\iota) \sum_{s' \in \sset} P_h(s'|s,a') \left( r_h(s,s') - p(s') + V^{\text{P},\sigma^1}_{h+1}(s',g^1(s,\iota,p,a,s'))\right),
	\end{equation*}
	where $a'=a'(s,\iota,p,a)$.
	Intuitively, $V^{\text{P},\sigma^1,\alpha}_h$ is ``equivalent'' to the function $V^{\text{P},\pi^1,\alpha}_h$, but it is written leveraging the promise-form representation.
	With an argument similar to the one employed by Lemma~\ref{lem:promise_functions_eq}, one can show that for every step $\bar{h} \in \mathcal{H}$ and history $\bar{\tau} = (s_1,p_1,a_1,\dots,s_{\bar{h}-1},p_{\bar{h}-1},a_{\bar{h}-1},s_{\bar{h}}) \in \mathcal{T}_{\bar{h}}$, it holds that:
	\begin{equation}
		\label{eq:v_sigma_1_alpha_eq_p1_alpha}
		V^{\text{P},\pi^1,\alpha}_{\bar{h}}(\bar{\tau}) = V^{\text{P},\sigma^1,\alpha}_{\bar{h}} (s,\iota),
	\end{equation} 
	where $s=s_{\bar{h}}$ and $\iota = \iota^{\sigma^1}_{\bar{\tau}}$.
	
	Now we prove by induction that $V^{\text{P},\sigma^2}_h(s,\iota) = V^{\text{P},\sigma^1,\alpha}_h(s,\iota)$ for every step $h \in \mathcal{H}$, state $s \in \sset$ and promise $\iota \in I^1_h(s)$ (observe that $I^1_h(s) = I^2_h(s)$ by construction).
	The base step is $h=H$:
	\begin{align*}
		V^{\text{P},\sigma^2}_H(s,\iota) &= \hspace{-0.3cm} \sum_{\substack{(p,a') : \\a' \in \mathcal{A}, \\ p \in J^2_H(s,\iota,a')}} \hspace{-0.3cm} \varphi^2_H(p,a'|s,\iota) \sum_{s' \in \sset} P_H(s'|s,a') \left( r_H(s,s') - p(s') \right) \\
		&=\hspace{-0.3cm} \sum_{\substack{(p,a) : \\a \in \mathcal{A}, \\ p \in J^1_H(s,\iota,a)}}\hspace{-0.3cm}  \varphi^2_H(p,a'(s,\iota,p,a)|s,\iota) \sum_{s' \in \sset} P_H(s'|s,a'(s,\iota,p,a)) \left( r_H(s,s') - p(s') \right) \\
		&=\hspace{-0.3cm} \sum_{\substack{(p,a) : \\a \in \mathcal{A}, \\ p \in J^1_H(s,\iota,a)}}\hspace{-0.3cm} \varphi^1_H(p,a|s,\iota) \sum_{s' \in \sset} P_H(s'|s,a'(s,\iota,p,a)) \left( r_H(s,s') - p(s') \right) = V^{\text{P},\sigma^1,\alpha}_H(s,\iota),
	\end{align*}
	where the second equality holds by construction, as for every action $a' \in \mathcal{A}$ and contract $p \in J^2_H(s,\iota,a')$ there is another action $a \in \mathcal{A}$ such that $a'=a'(s,\iota,p,a)$ $J^2_H(s,\iota,a') = J^1_H(s,\iota,a)$, while the third equality holds thanks to Line~\ref{line:phase_2_phi_assignment} Algorithm~\ref{alg:from_ep_ic_to_ic_phase_2}.
	
	Now we consider a step $h<H$ and assume that $V^{\text{P},\sigma^2}_{h+1}(s',\iota') = V^{\text{P},\sigma^1,\alpha}_{h+1}(s',\iota')$ for every state $s' \in \sset$ and promise $\iota' \in I^1_{h+1}(s')=I^2_{h+1}(s')$.
	Then, for every state $s \in \sset$ and promise $\iota \in I^1_h(s)$ the following holds:
	\begin{align*}
		&V^{\text{P},\sigma^2}_h(s,\iota) \\
		&= \hspace{-0.5cm} \sum_{\substack{(p,a') : \\a' \in \mathcal{A}, \\ p \in J^2_h(s,\iota,a')}} \hspace{-0.5cm} \varphi^2_h(p,a'|s,\iota) \sum_{s' \in \sset} P_h(s'|s,a') \left( r_h(s,s') - p(s') + V^{\text{P},\sigma^2}_{h+1}(s',g^2(s,\iota,p,a',s')) \right) \\
		&= \hspace{-0.5cm} \sum_{\substack{(p,a) : \\a \in \mathcal{A}, \\ p \in J^1_h(s,\iota,a)}} \hspace{-0.5cm} \varphi^1_h(p,a|s,\iota) \sum_{s' \in \sset} P_h(s'|s,a') \left( r_h(s,s') - p(s') + V^{\text{P},\sigma^2}_{h+1}(s',g^2(s,\iota,p,a',s')) \right) \\
		&= \hspace{-0.5cm} \sum_{\substack{(p,a) : \\a \in \mathcal{A}, \\ p \in J^1_h(s,\iota,a)}} \hspace{-0.5cm} \varphi^1_h(p,a|s,\iota) \sum_{s' \in \sset} P_h(s'|s,a') \left( r_h(s,s') - p(s') + V^{\text{P},\sigma^2}_{h+1}(s',g^1(s,\iota,p,a,s')) \right) \\
		&= \hspace{-0.5cm} \sum_{\substack{(p,a) : \\a \in \mathcal{A}, \\ p \in J^1_h(s,\iota,a)}} \hspace{-0.5cm} \varphi^1_h(p,a|s,\iota) \sum_{s' \in \sset} P_h(s'|s,a') \left( r_h(s,s') - p(s') + V^{\text{P},\sigma^1,\alpha}_{h+1}(s',g^1(s,\iota,p,a,s')) \right) \\
		&= V^{\text{P},\sigma^1,\alpha}_h(s,\iota).
	\end{align*}
	where the second equality holds, with $a' = a'(s,\iota,p,a)$, by means of the same argument employed in the base step $h=H$, while the third one holds since $g^1(s,\iota,p,a,s') = g^2(s,\iota,p,a',s')$ by construction (see Line~\ref{line:phase_2_g_assignment} Algorithm~\ref{alg:from_ep_ic_to_ic_phase_2}), and the second to last one thanks to the inductive hypothesis.
	
	As a result, we can show that $V^{\text{P},\pi^2} \ge V^{\text{P},\pi^1,\alpha}$, as:
	\begin{align*}
		V^{\text{P},\pi^2} &= \sum_{s \in \sset} \mu(s) V^{\text{P},\pi^2}_1(s) \\
		&=\sum_{s \in \sset} \mu(s) V^{\text{P},\sigma^2}_1(s,i^2(s)) \\
		&=\sum_{s \in \sset} \mu(s) V^{\text{P},\sigma^1,\alpha}_1(s,i^1(s)) \\
		&=\sum_{s \in \sset} \mu(s) V^{\text{P},\pi^1,\alpha}_1(s) = V^{\text{P},\pi^1,\alpha},
	\end{align*}
	where the first equality holds by definition, the second one thanks to Lemma~\ref{lem:promise_functions_eq}, the third one thanks to the result proved by induction above and because $\sigma^1$ and $\sigma^2$ share the initial promises, the second to last one holds according to Equation~\ref{eq:v_sigma_1_alpha_eq_p1_alpha}, and the last one by Lemma~\ref{lem:recursive_principal_dev}.
	
	In order to conclude the proof, we show that the policy $\sigma^2$ is IC.
	
	First, we show that $V^{\text{A},\sigma^2}_h(s,\iota) = \widehat{V}^{\text{A},\sigma^1}_h(s,\iota)$ for every step $h \in \mathcal{H}$, state $s \in \ss$ and promise $\iota \in I^2_h(s)=I^1_h(s)$. 
	Consider any action $a \in \mathcal{A}$ and contract $p \in J^1_h(s,\iota,a)$, and let $a' = a'(s,\iota,p,a)$
	Then it holds that:
	\begin{align}
		Q^{\text{A},\sigma^2}_h(s,\iota,p,a') &= \sum_{s' \in \sset} P_h(s'|s,a') \left( r^\text{A}_h(s,p,a',s') + V^{\text{A},\sigma^2}_{h+1}(s',g^2_h(s,\iota,p,a',s')) \right) \nonumber \\
		&= \sum_{s' \in \sset} P_h(s'|s,a') \left( r^\text{A}_h(s,p,a',s') + V^{\text{A},\sigma^2}_{h+1}(s',g^1_h(s,\iota,p,a,s')) \right) \nonumber \\
		&=\widehat{Q}^{\text{A},\sigma^1}_h(s,\iota,p,a) \label{eq:q_sigma_2_to_widehat_q_sigma_1}
	\end{align}
	where the second equality holds because $g^2_h(s,\iota,p,a',s') = g^1_h(s,\iota,p,a,s')$ (Line~\ref{line:phase_2_g_assignment}), while the last equality holds by construction given the definition of $a'$ (Line~\ref{line:phase_2_aprime_argmax}).
	Consequently:
	\begin{align*}
		V^{\text{A},\sigma^2}_h(s,\iota) &= \sum_{\substack{(p,a') : \\a' \in \mathcal{A}, \\ p \in J^2_H(s,\iota,a')}} \varphi^2_h(s,\iota|p,a') Q^{\text{A},\sigma^2}_h(s,\iota,p,a') \\
		&=\sum_{\substack{(p,a') : \\a \in \mathcal{A}, \\ p \in J^1_H(s,\iota,a)}} \varphi^2_h(s,\iota|p,a'(s,\iota,p,a)) Q^{\text{A},\sigma^2}_h(s,\iota,p,a'(s,\iota,p,a)) \\
		&=\sum_{\substack{(p,a') : \\a \in \mathcal{A}, \\ p \in J^1_H(s,\iota,a)}} \varphi^2_h(s,\iota|p,a'(s,\iota,p,a)) \widehat{Q}^{\text{A},\sigma^1}_h(s,\iota,p,a)  \\
		&=\sum_{\substack{(p,a') : \\a \in \mathcal{A}, \\ p \in J^1_H(s,\iota,a)}} \varphi^1_h(s,\iota|p,a') \widehat{Q}^{\text{A},\sigma^1}_h(s,\iota,p,a) = \widehat{V}^{\text{A},\sigma^1}_h(s,\iota),
	\end{align*}
	where the second equality holds by construction, as for every action $a' \in \mathcal{A}$ and contract $p \in J^2_H(s,\iota,a')$ such that $\varphi^2_h(s,\iota|p,a') > 0$ there is a (possibly different) action $a \in \mathcal{A}$ such that $a'=a'(s,\iota,p,a)$ and $J^2_H(s,\iota,a') = J^1_H(s,\iota,a)$, while the last one holds because $\varphi^2_h(s,\iota|p,a'(s,\iota,p,a)) = \varphi^1_h(s,\iota|p,a)$.
	
	Now we prove by induction that $Q^{\text{A},\sigma^2}_h(s,\iota,p,a') = Q^{\text{A},\sigma^2}_h(s,\iota,p,a')$ for every step $h \in \mathcal{H}$, state $s \in \sset$, promise $\iota \in I^2_h(s)$, and pair $(p,a) \in \mathcal{X}$ such that $\varphi^2_h(p,a'|s,\iota)>0$.
	As $\varphi^2_h(p,a'|s,\iota)>0$, there exists some action $a \in \mathcal{A}$ such that $a =a'(s,\iota,p,a)$, and thus Equation~\ref{eq:q_sigma_2_to_widehat_q_sigma_1} holds.
	The base step is $h=H$.
	\begin{align*}
		\widehat{Q}^{\text{A},\sigma^2}_h(s,\iota,p,a') &= \max_{\widehat{a} \in A} \sum_{s' \in \sset} P_h(s'|s,\widehat{a}) \left( r^\text{A}_h(s,p,\widehat{a},s')\right) \\
		&= \widehat{Q}^{\text{A},\sigma^1}_h(s,\iota,p,a) =Q^{\text{A},\sigma^2}_h(s,\iota,p,a'),
	\end{align*}
	where the first two equalities hold by definition, and the last one according to Equation~\ref{eq:q_sigma_2_to_widehat_q_sigma_1}.
	Now we consider $h<H$ we assume that the base step holds for the step $h+1$.
	Then:
	\begin{align*}
		\widehat{Q}^{\text{A},\sigma^2}_h(s,\iota,p,a') &= \max_{\widehat{a} \in A} \sum_{s' \in \sset} P_h(s'|s,\widehat{a}) \left( r^\text{A}_h(s,p,\widehat{a},s') + \widehat{V}^{\text{A},\sigma^2}_{h+1}(s',g^2_h(s,\iota,p,a',s'))\right) \\
		&=\max_{\widehat{a} \in A} \sum_{s' \in \sset} P_h(s'|s,\widehat{a}) \left( r^\text{A}_h(s,p,\widehat{a},s') + V^{\text{A},\sigma^2}_{h+1}(s',g^2_h(s,\iota,p,aì,s'))\right) \\
		&=Q^{\text{A},\sigma^2}_h(s,\iota,p,a'),
	\end{align*}
	where the second equality holds thanks to the inductive hypothesis.
	
	Thanks to the result above, we can prove that $\sigma^2$ is IC.
	Let $h \in \mathcal{H}$, $s \in \sset$, $\tau \in \mathcal{T}_h$ ending in state $s$, and consider a pair $(p,a) \in \mathcal{X}$ such that $\pi^2(p,a|\tau)>0$.
	Furthermore, let $\iota = \iota^{\sigma^2}_\tau$. 
	Then the following holds:
	\begin{equation*}
		Q^{\text{A},\pi^2}_h(\tau,p,a) = Q^{\text{A},\sigma^2}_h(s,\iota,p,a) = \widehat{Q}^{\text{A},\sigma^2}_h(s,\iota,p,a) = \widehat{Q}^{\text{A},\pi^2}_h(s,\iota,\tau),
	\end{equation*}
	where the first and last equalities hold thanks to Lemma~\ref{lem:promise_functions_eq}.
	As a result, the policy $\pi^2$ that implements $\sigma^2$ is IC.
	
	As a further step, we observe that $\sigma^2$ uses the same set of promises of $\sigma^1$. 
	Furthermore, every contract $p \in \mathcal{J}^2$ is added to some set $J_h(s,\iota,a')$ for some $h \in \mathcal{H}$, $s \in \sset$, $\iota \in I^3_h(s)$ and $a \in \mathcal{A}$.
	It is also easy to verify that Algorithm~\ref{alg:from_ep_ic_to_ic_phase_2} runs in time polynomial in the sizes of the instance and $\sigma^1$. 
\end{proof}

The last subprocedure of Algorithm~\ref{alg:from_ep_ic_to_ic} is provided in 
Algorithm~\ref{alg:from_ep_ic_to_ic_phase_3}.
While $\sigma^2$ is IC, it is not honest.
Indeed, it keeps the promises of the original policy $\sigma$, which work correctly as ``internal states'' of the policy, but have no semantic meaning.
In order to get some economic insights on the inner workings of our final policy, Algorithm~\ref{alg:from_ep_ic_to_ic_phase_3} changes the promises of policy $\sigma^2$ in order to build a policy $\sigma^3$ that is honest.
Intuitively it maps each promise $\iota_2$ of $\sigma^2$ to a new honest promise $\iota_3$.
In particular, for every step $h \in \mathcal{H}$, state $s \in \sset$ and promise $\iota_2 \in I^2_h(s)$, the new honest promise is $\iota_3 = V^{\text{A},\sigma^2}_h(s,\iota_2)$ (we recall that since $\sigma^2$ is IC, it holds that $V^{\text{A},\sigma^2}_h = \widehat{V}^{\text{A},\sigma^2}_h$).
The other functions of the policy $\sigma^3$ are defined in order to preserve this mapping, \emph{i.e.}, $\varphi^3_h(s,\iota_3) \coloneqq \varphi^2_h(s,\iota_2)$ and $g^3_h(s,\iota_3,p,a,s') \coloneqq V^{\text{A},\sigma^2}_h(s',g^2(s,\iota_2,p,a,s'))$.

However, it may happen that two different promises $\iota_2,\widetilde{\iota}_2 \in I^2_h(s)$ are mapped to the same new promise $\iota_3 = V^{\text{A},\sigma^2}_h(s,\iota_2) = V^{\text{A},\sigma^2}_h(s,\widetilde{\iota}_2)$.
To solve this ``conflict'', we define the functions $\varphi^3_h(s,\iota_3)$ and $g^3_h(s,\iota_3,p,a,s')$ considering the old promise that achieves the largest principal value, \emph{i.e.} $\varphi^3_h(s,\iota_3) \coloneqq \varphi^2_h(s,\iota_2)$ and $g^3_h(s,\iota_3,p,a,s') \coloneqq V^{\text{A},\sigma^2}_h(s',g(s,\iota_2,p,a,s'))$ if they provide the largest principal value, or the same using $\widetilde{\iota}_2$ in place of $\iota_2$ otherwise.
The value that the principal achieves by using an old promise $\iota_2 \in I^2_h(s)$ to solve a conflict in $\iota_3 \in I^3_h(s)$ is:
\begin{equation}
	\label{eq:realign_promises_conflict_value}
	G_h(s,\iota_2,\iota_3) \coloneqq \hspace{-0.35cm} \sum_{(p,a) \in \mathcal{X}_{\sigma^2}} \hspace{-0.35cm} \varphi^2_h(p,a|s,\iota_2) \sum_{s' \in \sset} P_h(s'|s,a)\left(r^\text{P}_h(s,p,s') + V^{\text{P},\sigma^3}_{h+1}(s',\iota') \right),
\end{equation} 
with $\iota' = V^{\text{A},\sigma^2}_{h+1}(s',g^2_h(s,\iota_2,p,a,s'))$.

Practically, Algorithm~\ref{alg:from_ep_ic_to_ic_phase_3} makes use of two dictionaries for each time step, $C_h : \sset \times \mathcal{I}^3 \rightarrow \mathcal{I}^2$ and $G_h : \sset \times \mathcal{I}^3 \times \mathcal{I}^2 \rightarrow \mathbb{R}$.
The first stores, for every state $s \in \sset$ and promise $\iota_3 \in \mathcal{I}^3_h(s)$, the set of old promises $C_h(s,\iota_3) \subseteq I^2_h(s)$ that have been mapped into $\iota_3$.
When $|C_h(s,\iota_3)|>1$, a conflict arises.
To solve it, Algorithm~\ref{alg:from_ep_ic_to_ic_phase_3} employs the dictionary $G_h$, which is filled according to Equation~\ref{eq:realign_promises_conflict_value}.

Algorithm~\ref{alg:from_ep_ic_to_ic_phase_3} employs dynamic programming to compute the values of the functions $V^{\text{A},\sigma^2}_h$ and $V^{\text{P},\sigma^3}_h$.
In particular, observe that when $G_h(s,\iota_3,\iota_2)$ is set at Line~\ref{line:realign_promises_conflict_value}, the policy $\sigma^3$ has been already computed for the step $h+1$, thus making it possible to compute $V^{\text{P},\sigma^3}_{h+1}(s',V^{\text{A},\sigma^2}_{h+1}(g^2_h(s,\iota_2,p,a,s')))$.

\begin{algorithm}[H]
	\caption{\texttt{Realign-promises}}\label{alg:from_ep_ic_to_ic_phase_3}
	\begin{algorithmic}[1]
		\Require IC promise-form policy $\sigma = \{(I^2_h,J^2_h,\varphi^2_h,g^2_h)\}_{h \in \mathcal{H}}$.
		
		\State Initialize dictionaries $C_h(s,\iota_3)$ and $G_h(s,\iota_3,\iota_2)$ to empty.
		\State \Comment{$C_h(s,\iota_3) \subseteq I^2_h(s)$ is the set of old promises mapped to the new promise $\iota'$}
		\State \Comment{$G_h(s,\iota_3,\iota_2)$ is the value that the principal gets if it copies the policy $\sigma^2$ with promise $\iota_2$ when the new promise is $\iota_3$.}
		
		\For{$h=H,\dots,1$}
		\ForAll{$s \in \sset$}
		\ForAll{$\iota_2 \in I^2_h(s)$}
		\State Compute and store $V^{\text{A},\sigma^2}_h(s,\iota_2)$
		\State $\iota_3 \gets V^{\text{A},\sigma^2}_h(s,\iota_2)$
		\State $C_h(s,\iota_3) \gets C_h(s,\iota_3) \cup \{\iota_2\}$
		\State $G_h(s,\iota_3,\iota_2) \gets $the value of Equation~\ref{eq:realign_promises_conflict_value}. \label{line:realign_promises_conflict_value}
		\State $I^3_h(s) \gets I^3_h(s) \cup \{\iota_3\}$
		\EndFor
		\ForAll{$\iota_3 \in I^3_h(s)$}
		\State $\iota_2 \gets \argmax_{\widetilde{\iota}_2 \in C_h(s,\iota_3)}G_h(s,\iota_3,\widetilde{\iota}_2)$ \label{line:phase_3_argmax}
		\State $\strut \varphi^3_h(p,a|s,\iota_3) \gets \varphi^2_h(p,a|s,\iota_2) \quad \forall a \in \mathcal{A}, p \in J^2_h(s,\iota_2,a)$ \label{line:phase_3_phi_assignment}
		\State $\strut J^3_h(s,\iota_3,a) \gets J^2_h(s,\iota_2,a) \quad \forall a \in \mathcal{A}$
		\State $g^3_h(s,\iota_3,p,a,s') \gets V^{\text{A},\sigma^2}_{h+1}(s',g^2_h(s,\iota_2,p,a,s')) \quad \forall a \in \mathcal{A}, p \in J^2_h(s,\iota_2,a), s' \in \sset$ \label{line:phase_3_g_assignment}
		\State Compute and store $V^{\text{P},\sigma^3}_h(s,\iota_3)$
		\EndFor
		\EndFor
		\EndFor
		
		\State \textbf{Return} $\sigma3 = \{I^3_h, J^3_h, \varphi^3_h, g^3_h\}_{h \in \mathcal{H}}$
	\end{algorithmic}
\end{algorithm}

\begin{restatable}{lemma}{fromepictoicphase3}
	\label{lem:from_ep_ic_to_ic_phase_3}
	Let $\sigma^2$ be an IC promise-form policy.
	Then Algorithm~\ref{alg:from_ep_ic_to_ic_phase_3} returns an IC and honest promise-form policy $\sigma^3$ with value $V^{\text{P},\sigma^3} \ge V^{\text{P},\sigma^2}$.
	Furthermore, Algorithm~\ref{alg:from_ep_ic_to_ic_phase_3} runs in time polynomial in $|\mathcal{J}^2|$, $|\mathcal{I}^2|$ and the instance size, while $|\mathcal{J}^3| \le |\mathcal{J}^2|$ and $\mathcal{I}^3 \le \mathcal{I}^2$.
\end{restatable}
\begin{proof}
	As a first step, we observe that for every step $h \in \mathcal{H}$ and state $s \in \ss$, by construction:
	\begin{equation*}
		I^3_h(s) = \{\iota_3 \in \mathbb{R} \mid \iota_3 = V^{\text{A},\sigma^2}_h(s,\iota_2), \iota \in I^2_h(s) \}.
	\end{equation*}
	and:
	\begin{equation*}
		C_h(s,\iota_3) = \{\iota_2 \in I^2_h(s) \mid \iota_3 = V^{\text{A},\sigma^2}_h(s,\iota_2) \} \quad \forall \iota_3 \in I^3_h(s),
	\end{equation*}
	where $V^{\text{A},\sigma^2}_h = \widehat{V}^{\text{A},\sigma^2}_h$ as $\sigma^2$ is IC.
	Furthermore, $\mathcal{X}_{\sigma^3} \subseteq \mathcal{X}_{\sigma^2}$, as $\sigma^3$ prescribes only pairs $(p,a) \in \mathcal{X}$ that are prescribed also by $\sigma^2$.
	
	Now we will first show that $V^{\text{P},\sigma^3} \ge V^{\text{P},\sigma^2}$, and then that the policy $\sigma^3$ is both honest and IC. 
	
	We prove by recursion that for every step $h \in \mathcal{H}$, state $s \in \sset$, and promise $\widetilde{\iota}_2 \in I^2_h(s)$, it holds that $V^{\text{P},\sigma^2}_h(s,\widetilde{\iota}_2) \le V^{\text{P},\sigma^3}_h(s,\iota_3)$, where $\iota_3 = V^{\text{A},\sigma^2}_h(s,\widetilde{\iota}_2)$.
	Let:
	\begin{equation*}
		\iota_2 = \argmax_{\bar{\iota}_2 \in C_h(s,\iota_3)}G_h(s,\iota_3,\bar{\iota}_2).
	\end{equation*}
	be the value selected at Line~\ref{line:phase_3_argmax} Algorithm~\ref{alg:from_ep_ic_to_ic_phase_3}.
	We also observe that $\iota_2, \widetilde{\iota}_2 \in C_h(s,\iota_3)$.
	
	The base step is $h = H$.
	Then, according to Equation~\ref{eq:realign_promises_conflict_value} and the definition of $\iota_2$:
	\begin{equation*}
		 \hspace{-0.35cm} \sum_{(p,a) \in \mathcal{X}_{\sigma^2}} \hspace{-0.35cm} \varphi^2_H(p,a|s,\iota_2) \sum_{s' \in \sset} \hspace{-0.15cm} P_H(s'|s,a)r^\text{P}_H(s,p,s') \ge \hspace{-0.35cm} \sum_{(p,a) \in \mathcal{X}_{\sigma^2}} \hspace{-0.35cm} \varphi^2_H(p,a|s,\widetilde{\iota}_2) \sum_{s' \in \sset} \hspace{-0.15cm} P_H(s'|s,a)r^\text{P}_H(s,p,s').
	\end{equation*}
	Consequently:
	\begin{align*}
		V^{\text{P},\sigma^2}_H(s,\widetilde{\iota}_2) &= \sum_{(p,a) \in \mathcal{X}_{\sigma^2}} \varphi^2_H(p,a|s,\widetilde{\iota}_2) \sum_{s' \in \sset} P_H(s'|s,a)r^\text{P}_H(s,p,s') \\
		&\le \sum_{(p,a) \in \mathcal{X}_{\sigma^2}} \varphi^2_H(p,a|s,\iota_2) \sum_{s' \in \sset} P_H(s'|s,a)r^\text{P}_H(s,p,s') \\
		&= \sum_{(p,a) \in \mathcal{X}_{\sigma^3}} \varphi^3_H(p,a|s,\iota_3) \sum_{s' \in \sset} P_H(s'|s,a)r^\text{P}_H(s,p,s') = V^{\text{P},\sigma^3}_H(s,\iota_3),
	\end{align*}
	where the second to last equality holds because $\mathcal{X}_{\sigma^3} \subseteq \mathcal{X}_{\sigma^2}$ and, as of Line~\ref{line:phase_3_phi_assignment} Algorithm~\ref{alg:from_ep_ic_to_ic_phase_3}, $\varphi^2_H(p,a|s,\iota_2) = \varphi^3_H(p,a|s,\iota_3)$.
	
	Now consider $h<H$ and assume that $V^{\text{P},\sigma^2}_{h+1}(s,\widetilde{\iota}_2) \le V^{\text{P},\sigma^3}_{h+1}(s,\iota_3)$ for every state $s \in \sset$, and promises $\widetilde{\iota}_2 \in I^2_{h+1}(s)$ and $\iota_3 = V^{\text{A},\sigma^2}_h(s,\widetilde{\iota}_2)$.
	Let $h \in \mathcal{H}$, $s \in \sset$, and consider two promises $\widetilde{\iota}_2 \in I^2_h(s)$ and $\iota_3 = V^{\text{A},\sigma^2}_h(s,\widetilde{\iota})$.
	Let also $\iota_2$ be the value selected at Line~\ref{line:phase_3_argmax} Algorithm~\ref{alg:from_ep_ic_to_ic_phase_3}.
	Then, according to Equation~\ref{eq:realign_promises_conflict_value} and the definition of $\iota_2$:
	\begin{equation}
		\label{eq:phase_3_principal_inequality}
		\begin{split}
		\hspace{-0.35cm} \sum_{(p,a) \in \mathcal{X}_{\sigma^2}} \hspace{-0.35cm} &\varphi^2_h(p,a|s,\iota_2) \hspace{-0.1cm} \sum_{s' \in \sset} \hspace{-0.15cm} P_h(s'|s,a) \hspace{-0.1cm} \left(r^\text{P}_h(s,p,s') \hspace{-0.1cm}+\hspace{-0.1cm} V^{\text{P},\sigma^3}_{h+1}(s',V^{\text{A},\sigma^2}_{h+1}(s',g^2_h(s,\iota_2,p,a,s'))) \right) \\
		\ge &\hspace{-0.35cm} \sum_{(p,a) \in \mathcal{X}_{\sigma^2}} \hspace{-0.35cm} \varphi^2_h(p,a|s,\widetilde{\iota}_2) \hspace{-0.1cm} \sum_{s' \in \sset} \hspace{-0.15cm} P_h(s'|s,a) \hspace{-0.1cm} \left(r^\text{P}_h(s,p,s') \hspace{-0.1cm}+\hspace{-0.1cm} V^{\text{P},\sigma^3}_{h+1}(s',V^{\text{A},\sigma^2}_{h+1}(s',g^2_h(s,\widetilde{\iota}_2,p,a,s'))) \right).
		\end{split}
	\end{equation}
	Consequently:
	\begin{align*}
		&V^{\text{P},\sigma^2}_h(s,\widetilde{\iota}_2) \\
		&= \sum_{(p,a) \in \mathcal{X}_{\sigma^2}} \varphi^2_h(p,a|s,\widetilde{\iota}_2) \sum_{s' \in \sset} P_h(s'|s,a)\left(r^\text{P}_h(s,p,s') + V^{\text{P},\sigma^2}_{h+1}(s',g_h(s,\widetilde{\iota}_2,p,a,s'))\right) \\
		&\le \sum_{(p,a) \in \mathcal{X}_{\sigma^2}} \varphi^2_h(p,a|s,\widetilde{\iota}_2) \sum_{s' \in \sset} P_h(s'|s,a)\left(r^\text{P}_h(s,p,s') + V^{\text{P},\sigma^3}_{h+1}(s',V^{\text{A},\sigma^2}_{h+1}(s',g^2_h(s,\widetilde{\iota}_2,p,a,s')))\right) \\
		&\le \sum_{(p,a) \in \mathcal{X}_{\sigma^2}} \varphi^2_h(p,a|s,\iota_2) \sum_{s' \in \sset} P_h(s'|s,a)\left(r^\text{P}_h(s,p,s') + V^{\text{P},\sigma^3}_{h+1}(s',V^{\text{A},\sigma^2}_{h+1}(s',g^2_h(s,\iota_2,p,a,s'))) \right) \\
		&= \sum_{(p,a) \in \mathcal{X}_{\sigma^2}} \varphi^2_h(p,a|s,\iota_2) \sum_{s' \in \sset} P_h(s'|s,a)\left(r^\text{P}_h(s,p,s') + V^{\text{P},\sigma^3}_{h+1}(s',g^3_h(s,\iota_3,p,a,s')) \right) \\
		&= \sum_{(p,a) \in \mathcal{X}_{\sigma^3}} \varphi^3_h(p,a|s,\iota_3) \sum_{s' \in \sset} P_h(s'|s,a)\left(r^\text{P}_h(s,p,s') + V^{\text{P},\sigma^3}_{h+1}(s',g^3_h(s,\iota_3,p,a,s')) \right) \\
		&=V^{\text{P},\sigma^3}_h(s,\iota_3)
	\end{align*}
	where the first inequality holds thanks to the inductive hypothesis and the second one according to Equation~\ref{eq:phase_3_principal_inequality}.
	Furthermore, by construction, $\varphi^2_h(p,a|s,\iota_2) = \varphi^3_h(p,a|s,\iota_3)$ and $g^2_h(s,\iota_2,p,a,s') = g^3_h(s,\iota_3,p,a,s')$, which proves the last three equalities.
	
	Let $h \in \mathcal{H}$ and $s \in \sset$.
	Given the structure of the set $I^3_h(s)$, for every $\iota_3 \in I^3_h(s)$ there is a promise $\iota_2 \in I^2_h(s)$ such that $\iota_3 = V^{\text{A},\sigma^2}_h(s,\widetilde{\iota}_2)$.
	Since $V^{\text{P},\sigma^3}_{h+1}(s,\iota_3) \ge V^{\text{P},\sigma^2}_{h+1}(s,\widetilde{\iota}_2)$, it follows that $V^{\text{P},\pi^3} \ge V^{\text{P},\pi^2}$, where $\pi^3$ and $\pi^2$ are the history-dependent policies that implement $\sigma^3$ and $\sigma^2$ respectively.
	
	As a further step, we show that the policy $\sigma^3$ is honest.
	We prove that for every step $h \in \mathcal{H}$, state $s \in \sset$ and promise $\iota_3 \in I^3_h(s)$, the following holds:
	\begin{equation*}
		\sum_{(p,a) \in \mathcal{X}_{\sigma^3}} \varphi^3_h(p,a|s,\iota_3) \sum_{s' \in \sset} P_h(s'|s,a) \left(r^{\text{A}}_h(s,p,a,s') + g^3_h(s,\iota_3,p,a,s')\right) = \iota_3,
	\end{equation*}
	where, by construction, $\iota_3 = V^{\text{A},\sigma^2}_h(s,\iota_2)$ for some promise $\iota_2 \in I^2_h(s)$.
	We observe that:
	\begin{align*}
		&\sum_{(p,a) \in \mathcal{X}_{\sigma^3}} \varphi^3_h(p,a|s,\iota_3) \sum_{s' \in \sset} P_h(s'|s,a) \left(r^{\text{A}}_h(s,p,a,s') + g^3_h(s,\iota_3,p,a,s')\right) \\
		&=\sum_{(p,a) \in \mathcal{X}_{\sigma^2}} \varphi^2_h(p,a|s,\iota_2) \sum_{s' \in \sset} P_h(s'|s,a) \left(r^{\text{A}}_h(s,p,a,s') + V^{\text{A},\sigma^2}_{h+1}(s',g^2_h(s,\iota_2,p,a,s'))\right) \\
		&=V^{\text{A},\sigma^2}_{h}(s,\iota_2),
	\end{align*}
	where the first equality holds because $\varphi^3_h(p,a|s,\iota_3) = \varphi^2_h(p,a|s,\iota_2)$ (see Line~\ref{line:phase_3_phi_assignment}) and $g^3_h(s,\iota_3,p,a,s') = V^{\text{A},\sigma^2}_{h+1}(s',g^2_h(s,\iota_2,p,a,s'))$ given the definition of $I^3_{h+1}(s')$ and Line~\ref{line:phase_3_g_assignment} Algorithm~\ref{alg:from_ep_ic_to_ic_phase_3}.
	As a result, $\sigma^3$ is honest.
	
	As a further step, we show that $\sigma^3$ is IC by means of Lemma~\ref{lem:local_ic_constr}.
	Thus, we have to prove that for any step $h \in \mathcal{H}$, state $s \in \sset$, promise $\iota_3 \in I^3_h(s)$, actions $a,\widehat{a} \in \mathcal{A}$ and contract $p \in J^3_h(s,\iota_3,a)$, the following holds:
	\begin{equation*}
		\begin{split} 
			\sum_{s' \in \sset} &P_h(s'|s,a) \left(r^\text{A}_h(s,p,a,s') +g^3_h(s,\iota_3,p,a,s') \right) \ge \\
			&\sum_{s' \in \sset} P_h(s'|s,\widehat{a}) \left(r^\text{A}_h(s,p,\widehat{a},s') +g^3_h(s,\iota_3,p,a,s') \right).
		\end{split}	
	\end{equation*}
	Let $\iota_2 \in \mathcal{I}^2_h(s)$ be such that $\iota_2 = V^{\text{A},\sigma^2}_h(s,\iota_2)$.
	Then $g^3_h(s,\iota_3,p,a,s') = V^{\text{A},\sigma^2}_h(s,g^2_h(s,\iota_2,p,a,s'))$.
	Since the policy $\sigma^2$ is IC, it follows that the inequality above is satisfied.
	As a result, $\sigma^3$ is IC according to Lemma~\ref{lem:local_ic_constr}.
	
	In order to conclude the proof, we observe that Algorithm~\ref{alg:from_ep_ic_to_ic_phase_3} runs in time polynomial in the sizes of $\sigma^2$ and the instance.
	Furthermore, since for every step $h \in \mathcal{H}$ and state $s \in \sset$ we have that  $|I^3_h(s)| \le |I^2_h(s)|$ by construction, it follows that $|\mathcal{I}^3| \le |\mathcal{I}^2|$.
	A similar result holds for the contracts: fix a step $h \in \mathcal{H}$ and a state $s \in \sset$, promise $\iota^3 \in I^3_h(s)$.
	For every promise $\iota_3 \in I^3_h(s)$ there is a promise $\iota_2 \in I_h(s)$ such that $J^3_h(s,\iota_3,a) = J^2_h(s,\iota_2,a)$ for every $a \in \mathcal{A}$.
	Such a promise $\iota_2= \iota_2(\iota_3)$ is different for every $\iota_3 \in I^3_h(s)$.
	As $|I^3_h(s)| \le |I^2_h(s)|$, it follows that $|\mathcal{J}^3| \le |\mathcal{J}^2|$
\end{proof}

\subsection{Omitted proofs from Section~\ref{sec:epsilon_ic_to_ic}}
\fromepictopseudoic*
\begin{proof}
	To prove the statement, we introduce a set of pseudo-histories $t \in \mathcal{T}^\circ$.
	Formally, for each step $h \in \mathcal{H} \cup \{H+1\}$ we define the set $\mathcal{T}^\circ_h$ as:
	\begin{equation*}
		\mathcal{T}^\circ_h \coloneqq \{(s_1,a_1,\dots,s_{h-1},a_{h-1},s_h) : s_1 \in \sset, a_i \in \mathcal{A}\}.
	\end{equation*}
	We further let $\mathcal{T}^\circ \coloneqq \mathcal{T}^\circ_1 \cup \mathcal{T}^\circ_2 \cup \dots \cup \mathcal{T}^\circ_{H+1}$.
	Intuitively, the pseudo-histories in $\mathcal{T}^\circ$ correspond to the histories in $\mathcal{T}'$ without the contracts.
	Given a history $\tau$, the corresponding pseudo-history is $t^\tau=(s_1,a_1,\dots,s_h,a_h,s_{h+1})$.
	Furthermore, we let $t^\tau(s_{h'}) = (s_1,a_1,\dots,s_{h'-1},a_{h'-1},s_{h'})$ for any $h'<h$.
	For the sake of notation, whenever we refer to a pseudo-history as $t$, we let $t=(s_1,a_1,\dots,s_h,a_h,s_{h+1})$ for the appropriate $h \in \mathcal{H}$.
	
	By hypothesis, given a history $\tau \in \mathcal{T}$ and an action $a$, there is at most a contract $p$ such that $\pi(p,a|\tau)>0$.
	The same holds for $\pi'$.
	The pseudo-history $t^\circ(\tau)$ removes the contracts from the history $\tau$.
	These contracts can be reconstructed considering the property explained above.
	Formally, given a pseudo-history $t=(s_1,a_1,\dots,s_h,a_h,s_{h+1}) \in \mathcal{T}^\circ$, there is at most only one history $\tau \coloneqq f^\pi(t) \coloneqq (s_1,p_1,a_1,\dots,p_h,s_h,a_h,s_{h+1}) \in \mathcal{T}'$ such that:
	\begin{equation*}
		\pi(p_{h'},a_{h'}|\tau(s_{h'})) >0 \quad\forall h' \le h.
	\end{equation*}
	If such a history does not exist, we let $f^\pi(t) \coloneqq (s_1,p_1,a_1,\dots,s_h,p_h,a_h,s_{h+1}) \in \mathcal{T}'$ be the history filled with any contract.
	Indeed, regardless of how $f^\pi(t)$ is defined in this case, it holds that $\mathbb{P}(f^\pi(t)|\pi,\alpha) = 0$ for any agent's function $\alpha$.
	We also let $p^{\pi,t}$ be the last contract of $f^\pi(t)$.
	We observe that if $\mathbb{P}(\tau|\pi,\alpha)>0$ for some agent's function $\alpha$, then $f^\pi(t^\tau)=\tau$.
	
	Let $\alpha$ be the \emph{recommended} agent's function, and $\alpha'$ be an agent's function \emph{incentivized} by $\pi'$.
	Let also $\widehat{\alpha}$ be an agent's function such that:
	\begin{equation*}
		\widehat{\alpha}(f^{\pi}(t)) = \alpha'(f^{\pi'}(t)) \quad \forall t \in \mathcal{T}^\circ.
	\end{equation*}
	
	Now we rewrite the values that we will need as summations over the pseudo-histories and without using the policy $\pi'$.
	We start from $V^{\text{A},\pi,\alpha}$ and $V^{\text{A},\pi,\widehat{\alpha}}$:
	\begin{align*}
		V^{\text{A},\pi,\alpha} &= \sum_{\tau \in \mathcal{T}'_{\pi,\alpha}} \mathbb{P}(\tau|\pi,\alpha) \left(p_h(s_{h+1})-c_h(s_h,\alpha_h(\tau))\right) \\
		&=\sum_{\tau \in \mathcal{T}'_{\pi,\alpha}} \mathbb{P}(f^\pi(t^\circ(\tau))|\pi,\alpha) \left(p_h(s_{h+1})-c_h(s_h,\alpha_h(f^\pi(t^\tau)))\right) \\
		&=\sum_{t \in \mathcal{T}^\circ} \mathbb{P}(f^\pi(t)|\pi,\alpha) \left(p^{\pi,t}(s_{h+1})-c_h(s_h,\alpha_h(f^\pi(t)))\right),
	\end{align*}
	where the second equality holds because every pseudo-history $t \in \mathcal{T}^\circ$ correspond to a different history $f^\pi(\tau) \in \mathcal{T}'$.
	Similarly, we can show that:
	\begin{equation*}
		V^{\text{A},\pi,\widehat{\alpha}} = \sum_{t \in \mathcal{T}^\circ} \mathbb{P}(f^\pi(t)|\pi,\widehat{\alpha}) \left(p^{\pi,t}(s_{h+1})-c_h(s_h,\widehat{\alpha}_h(f^\pi(t)))\right).
	\end{equation*}
	
	By employing a similar argument, we can obtain a similar result for $V^{\text{A},\pi',\alpha'}$ and $V^{\text{A},\pi',\alpha}$:
	\begin{align*}
		V^{\text{A},\pi',\alpha'} &= \sum_{t \in \mathcal{T}^\circ} \mathbb{P}(f^{\pi'}(t)|\pi',\alpha') \left(p^{\pi',t}(s_{h+1})-c_h(s_h,\alpha'_h(f^{\pi'}(t)))\right) \text{ and} \\
		V^{\text{A},\pi',\alpha} &= \sum_{t \in \mathcal{T}^\circ} \mathbb{P}(f^{\pi'}(t)|\pi',\alpha) \left(p^{\pi',t}(s_{h+1})-c_h(s_h,\alpha_h(f^{\pi'}(t)))\right).
	\end{align*}
	Since $\alpha$ is the recommended agent's function, it holds that $\alpha_h(f^{\pi'}(t)) = \alpha_h(f^{\pi}(t))$, as the sequence of actions is the same in both $f^{\pi'}(t)$ and $f^{\pi}(t)$.
	We also observe that $\mathbb{P}(f^{\pi'}(t)|\pi',\alpha) = \mathbb{P}(f^{\pi}(t)|\pi,\alpha)$.
	Thus:
	\begin{equation*}
		V^{\text{A},\pi',\alpha'} = \sum_{t \in \mathcal{T}^\circ} \mathbb{P}(f^{\pi}(t)|\pi,\widehat{\alpha}) \left(p^{\pi',t}(s_{h+1}) -c_h(s_h,\widehat{\alpha}_h(f^{\pi'}(t)))\right).
	\end{equation*}
	Furthermore, we can show that $\mathbb{P}(f^{\pi'}(t)|\pi',\alpha') = \mathbb{P}(f^{\pi'}(t)|\pi',\alpha)$. Let:
	\begin{align*}
		\tau' &= (s_1,p'_1,a_1,\dots,s_h,p'_h,a_h,s_{h+1}) = f^{\pi'}(t) \text{ and } \\
		\tau &= (s_1,p_1,a_1,\dots,s_h,p_h,a_h,s_{h+1}) = f^{\pi}(t).
	\end{align*}
	Then:
	\begin{align*}
		\mathbb{P}(f^{\pi'}(t)|\pi',\alpha') &= \mu(s_1) \prod_{h'=1}^h P_{h'}(s_{h'+1}|s_{h'},\alpha'_{h'}(f^{\pi'}(t)))\pi'(p'_h,a_h|\tau'(s_{h'})) \\
		&= \mu(s_1) \prod_{h'=1}^h P_{h'}(s_{h'+1}|s_{h'},\widehat{\alpha}_{h'}(f^{\pi}(t)))\pi'(p_h,a_h|\tau(s_{h'})) = \mathbb{P}(f^{\pi'}(t)|\pi',\alpha).
	\end{align*}
	As a result:
	\begin{equation}
		V^{\text{A},\pi',\alpha} = \sum_{t \in \mathcal{T}^\circ} \mathbb{P}(f^{\pi}(t)|\pi,\alpha) \left(p^{\pi',t}(s_{h+1})-c_h(s_h,\alpha_h(f^{\pi}(t)))\right).
	\end{equation}
	
	Finally, similar arguments let us rewrite the principal values as:
	\begin{equation*}
		V^{\text{P},\pi,\alpha} = \sum_{t \in \mathcal{T}^\circ} \mathbb{P}(f^{\pi}(t)|\pi,\alpha) \left(r_h(s_h,s_{h+1}) -p^{\pi,t}(s_{h+1})\right)
	\end{equation*}
	and:
	\begin{equation*}
		V^{\text{P},\pi',\alpha'} = \sum_{t \in \mathcal{T}^\circ} \mathbb{P}(f^{\pi}(t)|\pi,\widehat{\alpha}) \left(r_h(s_h,s_{h+1}) -p^{\pi',t}(s_{h+1})\right).
	\end{equation*}
		
	We observe that, as of Lemma~\ref{lem:value_rec_alpha}, $V^{\text{P},\pi} = V^{\text{P},\pi,\alpha}$.
	Furthermore, as stated by Lemma~\ref{lem:value_alpha_epsilon_ic}, it holds that $V^{\text{A},\pi,\alpha} \ge V^{\text{A},\pi,\widehat{\alpha}} - \epsilon$.
	By rewriting those quantities as summations over the pseudo-histories, we get:
	\begin{equation}
		\label{eq:alpha_ep_ic}
		\begin{split}
			\sum_{t \in \mathcal{T}^\circ} &\mathbb{P}(f^\pi(t)|\pi,\alpha) \left(p^{\pi,t}(s_{h+1})-c_h(s_h,\alpha_h(f^\pi(t)))\right) \ge \\
			&\sum_{\tau \in \mathcal{T}^\circ} \mathbb{P}(f^\pi(t)|\pi,\widehat{\alpha}) \left(p^{\pi,t}(s_{h+1})-c_h(s_h,\widehat{\alpha}_h(f^\pi(t)))\right) - \epsilon \\
		\end{split}
	\end{equation}

	On the other hand, according to Lemma~\ref{lem:value_rec_alpha_at_most_dev}, it holds that $V^{\text{A},\pi',\alpha'} \ge V^{\text{A},\pi',\alpha}.$
	Indeed, the left side of the inequality above is equal to $\sum_{s \in \sset} \mu(s) \vagentdev_1(s)$, while the right side is less than or equal to it.
	The inequality above can be rewritten in term of pseudo-histories as:
	\begin{equation}
		\label{eq:alpha_p_ic}
		\begin{split}
			\sum_{t \in \mathcal{T}^\circ} &\mathbb{P}(f^{\pi}(t)|\pi,\widehat{\alpha}) \left(p^{\pi',t}(s_{h+1})-c_h(s_h,\widehat{\alpha}_h(f^\pi(t)))\right) \ge \\
			&\sum_{t \in \mathcal{T}^\circ} \mathbb{P}(f^{\pi}(t)|\pi,\alpha) \left(p^{\pi',t}(s_{h+1})-c_h(s_h,\alpha_h(f^{\pi}(t)))\right).
		\end{split}
	\end{equation}
	
	
	Now we employ Equation~\ref{eq:alpha_ep_ic} and Equation~\ref{eq:alpha_p_ic} to show that $V^{\textnormal{P},\pi,\alpha} -V^{\textnormal{P},\pi',\alpha'} \le (H+1)\sqrt{\epsilon}$.
	For the sake of notation, we let $\tau = f^\pi(t)$ for every pseudo-history $t \in \mathcal{T}^\circ$.
	As a first step, we observe that the difference $V^{\textnormal{P},\pi,\alpha} -V^{\textnormal{P},\pi',\alpha'}$ can be upper bounded as:
	\begin{align}
		&V^{\textnormal{P},\pi,\alpha} -V^{\textnormal{P},\pi',\alpha'} \nonumber \\
		&= \sum_{t \in \mathcal{T}^\circ}\mathbb{P}(\tau|\pi,\alpha) \left( r_h(s_h,s_{h+1}) -p^{\pi,t}(s_{h+1}) \right) \nonumber \\
		&\hspace{2cm} -\sum_{t \in \mathcal{T}^\circ}\mathbb{P}(\tau|\pi,\widehat{\alpha}) \left( r_h(s_h,s_{h+1}) -p^{\pi',t}(s_{h+1}) \right) \nonumber \\
		&= \sum_{t \in \mathcal{T}^\circ}\mathbb{P}(\tau|\pi,\alpha) \left( r_h(s_h,s_{h+1}) -p^{\pi,t}(s_{h+1}) \right) \nonumber \\
		&\hspace{2cm} -\sum_{t \in \mathcal{T}^\circ}\mathbb{P}(\tau|\pi,\widehat{\alpha}) \left( r_h(s_h,s_{h+1}) -(1-\sqrt{\epsilon})p^{\pi,t}(s_{h+1}) -\sqrt{\epsilon}r_h(s_h,s_{h+1}) \right) \nonumber \\
		&= \sum_{t \in \mathcal{T}^\circ}\mathbb{P}(\tau|\pi,\alpha) \left( r_h(s_h,s_{h+1}) -p^{\pi,t}(s_{h+1}) \right) \nonumber \\
		&\hspace{2cm} -(1-\sqrt{\epsilon}) \sum_{t \in \mathcal{T}^\circ}\mathbb{P}(\tau|\pi,\widehat{\alpha}) \left( r_h(s_h,s_{h+1}) -p^{\pi,t}(s_{h+1}) \right) \nonumber \\
		&\le \sum_{t \in \mathcal{T}^\circ}\left(\mathbb{P}(\tau|\pi,\alpha) - \mathbb{P}(\tau|\pi,\widehat{\alpha})\right) \left( r_h(s_h,s_{h+1}) -p^{\pi,t}(s_{h+1}) \right) +H\sqrt{\epsilon}, \label{eq:vp_vpp_partial_bound}
	\end{align}
	where the second equality holds according to the definition of $\pi'$, and the last inequality holds because:
	\begin{equation*}
		 \sqrt{\epsilon} \sum_{t \in \mathcal{T}^\circ}\mathbb{P}(\tau|\pi,\widehat{\alpha}) \left( r_h(s_h,s_{h+1}) -p^{\pi,t}(s_{h+1}) \right) = \sqrt{\epsilon}V^{\text{P},\pi,\widehat{\alpha}} \le H\sqrt{\epsilon}.
	\end{equation*}
	
	Now we show that:
	\begin{equation}
		\label{eq:fa_fap_r_p_bound}
		\sum_{t \in \mathcal{T}^\circ}\left(\mathbb{P}(\tau|\pi,\alpha) - \mathbb{P}(\tau|\pi,\widehat{\alpha})\right) \left( r_h(s_h,s_{h+1}) -p^{\pi,t}(s_{h+1}) \right) \le \sqrt{\epsilon}.
	\end{equation}
	As the policy $\pi'$ changes every contract $p$ proposed by $\pi$ at step $h$ to a new contract $p'$ defined as $p'(s') = (1-\sqrt{\epsilon})p(s') + \sqrt{\epsilon}r_h(s_h,s')$, Equation~\ref{eq:alpha_p_ic} can be rewritten as:
	\begin{align*}
		\sum_{t \in \mathcal{T}^\circ} &\mathbb{P}(\tau|\pi,\widehat{\alpha}) \left( (1-\sqrt{\epsilon})p^{\pi,t}(s_{h+1}) - \sqrt{\epsilon} r_h(s_h,s_{h+1}) -c_h(s_h,\widehat{\alpha}_h(\tau)) \right) \ge \\
		&\sum_{\tau \in \mathcal{T}^\circ} \mathbb{P}(\tau|\pi,\alpha) \left( (1-\sqrt{\epsilon})p^{\pi,t}(s_{h+1}) - \sqrt{\epsilon} r_h(s_h,s_{h+1}) -c_h(s_h,\alpha_h(\tau)) \right).
	\end{align*}
	This implies that:
	\begin{align*}
		&\sum_{t \in \mathcal{T}^\circ}\mathbb{P}(\tau|\pi,\widehat{\alpha}) \left( p^{\pi,t}(s_{h+1}) - c_h(s_h,\widehat{\alpha}_h(\tau)) \right) 
		-\sum_{t \in \mathcal{T}^\circ}\mathbb{P}(\tau|\pi,\alpha) \left( p^{\pi,t}(s_{h+1}) - c_h(s_h,\alpha_h(\tau)) \right) \\
		&\ge \sqrt{\epsilon} \sum_{t \in \mathcal{T}^\circ} \Big( \mathbb{P}(\tau|\pi,\widehat{\alpha})\left( p^{\pi,t}(s_{h+1})-r_h(s_h,s_{h+1}) \right)
		-\mathbb{P}(\tau|\pi,\alpha) \left( p^{\pi,t}(s_{h+1})-r_h(s_h,s_{h+1})\right) \Big) \\
		&= \sqrt{\epsilon} \sum_{t \in \mathcal{T}^\circ}\left( \mathbb{P}(\tau|\pi,\widehat{\alpha}) -\mathbb{P}(\tau|\pi,\alpha) \right) \left( p^{\pi,t}(s_{h+1}) - r_h(s_h,s_{h+1}) \right).
	\end{align*}
	By combining this result with Equation~\ref{eq:alpha_ep_ic}, we get:
	\begin{align*}
		\epsilon &\ge \sum_{\tau \in \mathcal{T}^\circ}\mathbb{P}(\tau|\pi,\widehat{\alpha}) \left( p^{\pi,t}(s_{h+1}) - c_h(s_h,\widehat{\alpha}_h(\tau)) \right) 
		-\sum_{t \in \mathcal{T}^\circ}\mathbb{P}(\tau|\pi,\alpha) \left( p^{\pi,t}(s_{h+1}) - c_h(s_h,\alpha_h(\tau)) \right) \\
		&\ge \sqrt{\epsilon} \sum_{t \in \mathcal{T}^\circ}\left( \mathbb{P}(\tau|\pi,\widehat{\alpha}) -\mathbb{P}(\tau|\pi,\alpha) \right) \left( p^{\pi,t}(s_{h+1}) - r_h(s_h,s_{h+1}) \right).
	\end{align*}
	Thus, Equation~\ref{eq:fa_fap_r_p_bound} holds.
	
	Finally, by combining Equation~\ref{eq:vp_vpp_partial_bound} and Equation~\ref{eq:fa_fap_r_p_bound}, we can upper bound the difference between $V^{\textnormal{P},\pi,\alpha}$ and $V^{\textnormal{P},\pi',\alpha'}$ as:
	\begin{equation*}
		V^{\textnormal{P},\pi,\alpha} -V^{\textnormal{P},\pi',\alpha'} \le \sqrt{\epsilon} + H\sqrt{\epsilon} = (H+1)\sqrt{\epsilon}.
	\end{equation*}
	As a result:
	\begin{equation*}
		V^{\textnormal{P},\pi',\alpha'} \ge V^{\textnormal{P},\pi,\alpha} - (H+1)\sqrt{\epsilon},
	\end{equation*}
	concluding the proof.
\end{proof}

\fromepictoic*
\begin{proof}
	By Lemma~\ref{lem:from_ep_ic_to_pseudo_ic}, the policy $\sigma^1$ computed by Algorithm~\ref{alg:from_ep_ic_to_ic} is such that $V^{\text{P},\sigma^1} \ge V^{\text{P},\sigma} -(H+1)\sqrt{\epsilon}$.
	Since $V^{\text{P},\sigma} \ge \text{OPT}$, it follows that:
	\begin{equation*}
		V^{\text{P},\sigma^1} \ge \text{OPT} -(H+1)\sqrt{\epsilon}.
	\end{equation*}
	As a further step we observe that, thanks to Lemma~\ref{lem:from_ep_ic_to_ic_phase_2}, the policy $\sigma^2$ is IC and achieves a value of:
	\begin{equation*}
		V^{\text{P},\sigma^2} = V^{\text{P},\sigma^1} \ge \text{OPT} -(H+1)\sqrt{\epsilon}.
	\end{equation*}
	Finally, according to Lemma~\ref{lem:from_ep_ic_to_ic_phase_3} the policy $\sigma^2$ is both IC and honest, and achieves a value of:
	\begin{equation*}
		V^{\text{P},\sigma^3} \ge V^{\text{P},\sigma^2} \ge \text{OPT} -(H+1)\sqrt{\epsilon}.
	\end{equation*}
	
	To conclude the proof, we observe that every subprocedure of Algorithm~\ref{alg:from_ep_ic_to_ic} can be executed in polynomial time, and, according to Lemma~\ref{lem:from_ep_ic_to_ic_phase_2} and Lemma~\ref{lem:from_ep_ic_to_ic_phase_3}, the size of $\sigma^3$ is given by $|\mathcal{I}^3| \le |\mathcal{I}|$ and $|\mathcal{J}^3| \le |\mathcal{J}|$.
\end{proof}

\section{Proofs omitted from Section~\ref{sec:approx_policy}}
In order to prove Theorem~\ref{th:main_theorem}, we need to introduce two additional lemmas.
The first characterizes the relationship between the value of the policy $\sigma$ computed by Algorithm~\ref{alg:approx_policy} and the values of the tables $M^\delta_h$.
The second lemma relates the tables $M^\delta_h$ to the value of an optimal policy $\pi^\star$.
By combining them, we will prove Theorem~\ref{th:main_theorem}.

\begin{restatable}{lemma}{valueAtLeastM}
	\label{lem:value_at_least_M}
	Let $\sigma =\{(I_h,J_h,\varphi_h,g_h)\}_{h \in \mathcal{H}}$ be the promise-form policy returned by Algorithm~\ref{alg:approx_policy}.
	For every $h \in \mathcal{H}$, $s \in \sset$, and $\iota \in I_h(s)$, it holds that $\vprincipalpr_h(s,\iota) \ge M^\delta_h(s,\iota)$.
\end{restatable}
\begin{proof}
	Let $\sigma = \{\{(I_h,\varphi_h,g_h)\}_{h \in \mathcal{H}}\}$ be the promise-form policy returned by Algorithm~\ref{alg:approx_policy}.
	We prove the statement by induction.
	For $h=H$, for any state $s \in \sset$ and promise $\iota \in I_H(s)$, we let $(\alpha,p,z,v) = O^{h,s}_{\iota,\delta}(M_{H+1})$. 
	Then, the following holds:
	\begin{align*}
		\vprincipalpr_H(s,\iota) &= \sum_{(p,a) \in \mathcal{X}} \varphi_H(p,a|s,\iota) \sum_{s' \in \sset} P_H(s'|s,a) r^{\text{P}}_H(s,p,s') \\
		&=\sum_{a \in \A} \alpha_a \sum_{s' \in \sset} P_H(s'|,s,a) r^{\text{P}}_H(s,p^a,s') \\
		&=F_{H,s,M^\delta_{H+1}}(\alpha,p,z) \\
		&\ge v,
	\end{align*}
	where the second equality holds by construction (see Algorithm~\ref{alg:approx_policy}), while the inequality follows from Definition~\ref{def:approx_oracle}.
	Since $v = M^\delta_{H}(s,\iota)$, the statement holds for the base case $h=H$.
	
	Now we consider a step $h<H$ and we assume that the statement holds for $h'=h+1$.
	Furthermore, we let $(\alpha,p,z,v) = O^{h,s}_{\iota,\delta}(M_{h+1})$.
	Then, the following holds:
	\begin{align*}
		\vprincipalpr_h(s,\iota) &= \sum_{(p,a) \in \mathcal{X}_\sigma} \varphi_h(p,a|s,\iota) \qprincipalpr_h(s,\iota,p,a) \\
		&= \sum_{a \in \A} \alpha_a \qprincipalpr_h(s,\iota,p^a,a) \\
		&= \sum_{a \in \A} \alpha_a \sum_{s' \in \sset} P_h(s'|s,a) \left( r^\text{P}_h(s,p^a,s') + \vprincipalpr_{h+1}(s',g_h(s,\iota,p^a,a,s')) \right) \\
		&\ge \sum_{a \in \A} \alpha_a \sum_{s' \in \sset} P_h(s'|s,a) \left( r^\text{P}_h(s,p^a,s') +M^\delta_{h+1}(s',g_h(s,\iota,p^a,a,s')) \right) \\
		&= \sum_{a \in \A} \alpha_a \sum_{s' \in \sset} P_h(s'|s,a) \left( r^\text{P}_h(s,p^a,s') +M^\delta_{h+1}(s',z(a,s')) \right) \\
		&= F_{h,s,M^\delta_{h+1}}(\alpha,p,z)\\
		&\ge v,
	\end{align*}
	where the first inequality holds employing the inductive hypothesis and the last inequality follows from Definition~\ref{def:approx_oracle}.
	The proof is concluded considering that $v = M^\delta_{h}(s,\iota)$ by construction.
\end{proof}

\begin{restatable}{lemma}{MAtLeastValue}
	\label{lem:M_at_least_value}
	Let $\sigma^\star =\{(I^\star_h,J^\star_h,\varphi^\star_h,g^\star_h)\}_{h \in \mathcal{H}}$ be a honest, direct, IC and optimal promise-form policy, and $M^\delta_h$, $I_h$ be computed accordingly to Algorithm~\ref{alg:approx_policy}.
	For every $h \in \mathcal{H}$, $s \in \sset$, and $\iota \in I^\star_h(s)$, it holds that $M^\delta_h(s,\lceil\iota\rceil_\delta) \ge V^{\text{P},\sigma^\star}_h(s,\iota)$ and $M^\delta_h(s,\lfloor\iota\rfloor_\delta) \ge V^{\text{P},\sigma^\star}_h(s,\iota)$.
\end{restatable}
\begin{proof}
	We recall that, w.l.o.g., we can consider a direct optimal, honest and IC promise-form policy $\sigma^\star$.
	Thus, for every action $a \in \A$ there is at most a single contract $p^{s,\iota,a}_h$ such that $\varphi^\star_h(p^{s,\iota,a}_h,a|s,\iota) > 0$.
	For the sake of notation, if $\varphi^\star_h(p,a|s,\iota)= 0 $ for every $p \in \mathcal{P}$, we let $p^{s,\iota,a}_h$ be any contract in $\mathcal{P}$.

	We prove the statement by induction, starting from $h=H$.
	We consider a state $s \in \sset$ and a promise $\iota \in I^\star_H(s)$.
	Let $(\alpha,p,z)$ be the solution defined as follows:
	\begin{align*}
		\alpha_a &= \varphi^\star_H(p^{s,\iota,a}_H,a|s,\iota) &\forall a \in \A, \\
		p_a &= p^{s,\iota,a}_H &\forall a \in \A, \\
		z(a,s') &= 0 &\forall a \in \A, s' \in \sset.
	\end{align*}
	We can verify that $(\alpha,p,z)$ is a feasible solution for the problem $\mathcal{P}_{h,s,\iota}(M^\delta_{H+1})$.
	Indeed, $z$ has values in $\mathcal{D}_\delta$, it satisfies Constraint~\ref{eq:opt_problem_constr_honest} since $\sigma^\star$ is honest and $g^\star_H(s,\iota,p^{s,\iota,a}_H,a,s') = 0 = z(a,s')$, and for every pair of actions $a,\widehat{a} \in \A$ such that $\alpha_a>0$ the following holds:
	\begin{align*}
		\sum_{s' \in \sset} P_H(s'|s,a) \left(r^A_H(s,p^a,a,s') + z(a,s')\right)  &=\sum_{s' \in \sset} P_h(s'|s,a) r^\textnormal{A}_h(s,p^a,a,s') \\
		&=\text{Q}^{A,\sigma^\star}_H(s,\iota,p^a,a) \\
		&\ge \sum_{s' \in \sset}  P_h(s'|s,\widehat{a}) r^\textnormal{A}_h(s,p^a,\widehat{a},s') \\
		&=\sum_{s' \in \sset} P_H(s'|s,\widehat{a}) \left(r^\textnormal{A}_H(s,p^a,\widehat{a},s') +z(a,s')\right),
	\end{align*}
	where the inequality holds because $\sigma^\star$ is IC. 
	As a result, Constraint~\ref{eq:opt_problem_constr_ic} is satisfied.
	
	Furthermore, the value of this solution is given by:
	\begin{align*}
		F_{H,s,M^\delta_{H+1}}(\alpha,p,z) &= \sum_{a \in \A} \alpha_a \sum_{s' \in \sset} P_H(s'|s,a) r^P_H(s,p^a,s') \\
		&=\sum_{(p',a) \in \mathcal{X}_\sigma^\star} \varphi^\star_H(p',a|s,\iota) \sum_{s' \in \sset} P_H(s',s,a)r^P_H(s,p',s') \\
		&=\text{V}^{P,\sigma^\star}_H(s,\iota).
	\end{align*}
	Now let $(\alpha^\circ,p^\circ,z^\circ,v^\circ) = O^{H,s}_{\iota,\delta}(M^\delta_{H+1})$.
	Thus, since $M^\delta_H(s,\lceil \iota \rceil_\delta)=v^\circ$ by construction is computed by an approximation oracle according to Definition~\ref{def:approx_oracle} and $\lceil \iota \rceil_\delta \in [\iota-\delta,\iota+\delta]$, the value of the optimal policy can be upper bounded as:
	\begin{align*}
		V^{\text{P},\sigma^\star}_H(s,\iota) &= F_{H,s,M^\delta_{H+1}}(\alpha,p,z) \\
		&\le \widehat{F}_{H,s,\iota}(M^\delta_{H+1}) \\
		&\le M^\delta_H(s,\lceil \iota \rceil_\delta),
	\end{align*}
	where the first inequality holds because $(\alpha,p,z)$ is a feasible solution of $\mathcal{P}_{h,s,\iota}(M^\delta_{H+1})$, while $\widehat{F}_{H,s,\iota}(M^\delta_{H+1})$ is the value of the optimal solution of such a problem.
	Similarly, $M^\delta_H(s,\lfloor \iota \rfloor_\delta) \ge \text{V}^{P,\sigma^\star}_H(s,\iota)$, proving the base step of the induction.
	
	Now we consider $h<H$ and assume that the statement holds for $h' = h+1$.
	Fix a state $s \in \sset$ and a promise $\iota \in I^\star_h(s)$. 
	Let $(\alpha,p,z)$ be the solution defined as:
	\begin{align*}
		\alpha_a &= \varphi^\star_h(p^{s,\iota,a}_h,a|s,\iota) &\forall a \in \A \\
		p_a &= p^{s,\iota,a}_h &\forall a \in \A \\
		z(a,s') &= \lfloor g^\star_h(s,\iota,a,s')\rfloor &\forall a \in \A, s' \in \sset.
	\end{align*}
	We observe that, since $\sigma^\star$ is honest and in every state the principal can take an action with cost zero, it holds that $\widetilde{\iota} \in [0,HB]$ for every promise $\widetilde{\iota} \in \mathcal{I}$.
	As $0,HB \in \mathcal{D}_\delta$, it follows that $z(a,s')$ (for every $s' \in \sset$), $\lfloor \iota \rfloor_\delta$ and $\lceil \iota \rceil_\delta$ belong to $\mathcal{D}_\delta$.
	
	For the sake of notation, let us define the following quantities:
	\begin{align*}
		k_{a,s'} &\coloneqq g^\star_h(s,\iota,a,s') - z(a,s'), \\
		\widetilde{k} &\coloneqq \sum_{a \in \A} \alpha_a \sum_{s' \in \sset} P_h(s'|s,a) k_{a,s'}, \\
		\widehat{k}_a &\coloneqq 
		\sum_{s' \in \sset} P_h(s'|s,a) k_{a,s'} \quad \forall a \in \A. \\
	\end{align*}
	
	We show that $(\alpha,p,z)$ is a feasible solution for the problem $\mathcal{P}_{h,s,\iota-\widetilde{k}}(M^\delta_{h+1})$.
	Constraint~\ref{eq:opt_problem_constr_honest} is satisfied as:
	\begin{align*}
		&\sum_{a \in A}\alpha_a \sum_{s' \in \sset} P_h(s'|s,a) \left(r^A_h(s,p^a,a,s') +z(a,s')\right) \\
		&= \sum_{a \in A}\alpha_a \sum_{s' \in \sset} P_h(s'|s,a) \left(r^A_h(s,p^a,a,s') +g^\star_h(s,\iota,p^a,a,s') -k_{a,s'}\right) \\
		&= \sum_{a \in A}\alpha_a \sum_{s' \in \sset} P_h(s'|s,a) \left(r^A_h(s,p^a,a,s') + g^\star_h(s,\iota,p^a,a,s')\right) -\widetilde{k} \\
		&= \iota -\widetilde{k},
	\end{align*}
	where the last equality holds because $\sigma^\star$ is honest.
	Furthermore, by considering that $\sum_{s' \in \sset} P_h(s'|s,a) =1$, for every pair of actions $a,\widehat{a} \in \A$ such that $\alpha_a>0$ the following holds:
	\begin{align*}
		&\sum_{s' \in \sset} P_h(s'|s,a) \left(r^A_h(s,p^a,a,s') + z(a,s')\right) \\
		&= \sum_{s' \in \sset} P_h(s'|s,a) \left(r^A_h(s,p^a,a,s') +g^\star_h(s,\iota,p^a,a,s') \right) -\widehat{k}_a \\
		&= \sum_{s' \in \sset} P_h(s'|s,a) \left(r^A_h(s,p^a,a,s') +\text{V}^{\text{P},\sigma^\star}_{h+1}(s',g^\star_h(s,\iota,p^a,a,s')) \right) -\widehat{k}_a \\
		&=\text{Q}^{A,\sigma^\star}_h(s,\iota,p^a,a) -\widehat{k}_a\\
		&\ge \sum_{s' \in \sset} P_h(\omega|s,\widehat{a}) \left(r^A_h(s,p^a,\widehat{a},s') \text{V}^{\text{P},\sigma^\star}_{h+1}(s',g^\star_h(s,\iota,p^a,a,s')) \right) -\widehat{k}_a \\
		&=\sum_{s' \in \sset} P_h(s'|s,\widehat{a}) \left(r^A_h(s,p^a,\widehat{a},s') +g^\star_h(s,\iota,p^a,a,s') \right) -\widehat{k}_a \\
		&=\sum_{s' \in \sset} P_h(s'|s,\widehat{a}) \left(r^A_h(s,p^a,\widehat{a},s') +z(a,s') \right),
	\end{align*}
	where we applied Lemma~\ref{lem:honesty_value} and the assumption that $\sigma^\star$ is both honest and IC.
	Consequently, Constraint~\ref{eq:opt_problem_constr_ic} is satisfied, and since $z$ has values in $\mathcal{D}_\delta$, it follows that $(\alpha,p,z)$ is a feasible solution for $\mathcal{P}_{h,s,\iota-\widetilde{k}}(M^\delta_{h+1})$.
	
	As a further step, we provide an upper bound for the value of the principal:
	\begin{align*}
		&V^{\text{P},\sigma^\star}_h(s,\iota) \\
		&=\sum_{a \in \A} \alpha_a \sum_{s' \in \sset} P_h(s'|s,a) \left(r^P_h(s,p^a,s') + V^{\text{P},\sigma^\star}_{h+1}(s',g^\star(s,\iota,p^a,a,s')) \right) \\
		&\le \sum_{a \in \A} \alpha_a \sum_{s' \in \sset} P_h(s'|s,a) \left(r^P_h(s,p^a,s') + M^\delta_{h+1}(s',\lfloor g^\star_h(s,\iota,p^a,a,s') \rfloor) \right) \\
		&=\sum_{a \in \A} \alpha_a \sum_{s' \in \sset} P_h(s'|s,a) \left(r^P_h(s,p^a,s') + M^\delta_{h+1}(s',z(a,s')) \right) \\
		&=F_{h,s,M^\delta_{h+1}}(\alpha,p,z) \\
		&\le \widehat{F}_{h,s,\iota-\widetilde{k}}(M^\delta_{h+1}),
	\end{align*}
	where we applied the definition of $V^{\text{P},\sigma^\star}_h(s,\iota)$ and $(\alpha,p,z)$, the inductive hypothesis, the definition of $z$, and the fact that $(\alpha,p,z)$ is a feasible solution for $\mathcal{P}_{h,s,\iota - \widetilde{k}}(M^\delta_{h+1})$.
	Finally, by considering that the value of $M^\delta_h(s,\lfloor \iota \rfloor_\delta)$ and $M^\delta_h(s,\lceil \iota \rceil_\delta)$ satisfy Definition~\ref{def:approx_oracle}, we can prove the statement and conclude the proof.
\end{proof}

\MainTh*
\begin{proof}
	First, we observe that $\sigma$ returned by Algorithm~\ref{alg:approx_policy} is a promise-form policy.
	Indeed, we have that $|I_1(s)|=1$ for every state $s \in \sset$ by construction (see Line~\ref{line:initial_promise} in  Algorithm~\ref{alg:approx_policy}). 
	By satisfying Definition~\ref{def:approx_oracle}, the approximating oracle guarantees the remaining properties of promise-form policies.
	
	As a further step, we show that $\sigma$ is $\eta$-honest, with $\eta = \lambda = 2\delta$.
	We observe that the policy $\sigma$ is constructed from the solutions provided by an approximation oracle that satisfies Definition~\ref{def:approx_oracle}.
	%
	Furthermore, for every state $s \in \sset$ and step $h \in \mathcal{H}$, a promise $\iota \in \mathcal{D}_\delta$ is added to the set $I_h(s)$ only if the value $v$ is larger than $-\infty$.
	Thus, according to Definition~\ref{def:approx_oracle}, the solutions $(\alpha,p,q)$ used to construct $\sigma$ belong to the corresponding set $\Psi^{h,s}_{\iota,2\delta}$ and satisfy Equation~\ref{eq:opt_rel_prob_constr_honest_ge} and Equation~\ref{eq:opt_rel_prob_constr_honest_le} for $\lambda = 2\delta$.
	Consequently, one can easily verify that $\sigma$ is $2\delta$-honest.
	
	Since $\sigma$ is $2\delta$-honest, we can employ Lemma~\ref{lem:local_ic_constr} to prove that it is also $4\delta H^2$-IC.
	In order to do that, we observe that $\sigma$ satisfies Equation~\ref{eq:local_ic_constr}, as the solutions returned by the approximation oracle satisfy Equation~\ref{eq:opt_rel_prob_constr_ic}.
	Thus, by observing that $\delta = \nicefrac{\epsilon}{4H^2}$, it follows that $\sigma$ is $\epsilon$-IC.
	
	To conclude the proof, we prove that the principal's expected cumulative utility achieved by of sigma is at least $\text{OPT}$.
	To show that, we consider an optimal and honest promise-form policy $\sigma^\star=\{I^\star_h,\varphi^\star_h,g^\star_h\}_{h \in \mathcal{H}}$, which always exists according to Theorem~\ref{th:promising_optimal}.
	By combining Lemma~\ref{lem:value_at_least_M} and Lemma~\ref{lem:M_at_least_value}, we have that $\vprincipalpr_h(s,\lfloor \iota \rfloor_\delta) \ge \text{V}^{\text{P},\sigma^\star}_h(s,\iota)$ for every state $s \in \sset$, step $h \in \mathcal{H} \setminus \{1\}$, and promise $\iota \in I^\star_h(s)$.
	Furthermore, at the time step $h=1$, thanks to Lemma~\ref{lem:value_at_least_M}, Lemma~\ref{lem:M_at_least_value} and Line~\ref{line:initial_promise} in Algorithm~\ref{alg:approx_policy}, it holds that $\vprincipalpr_1(s,i(s)) \ge V^{\text{P},\sigma^\star}_1(s,i^\star(s))$.
	Then, we have:
	\begin{align*}
		\text{OPT} &= V^{\text{P},\pi^{\sigma^\star}} \\
		&= \sum_{s \in \mathcal{S}} \mu(s) V^{\text{P},\pi^{\sigma^\star}}_1(s) \\
		&=\sum_{s \in \mathcal{S}} \mu(s) V^{\text{P},\sigma^\star}_1(s,i^\star(s)) \\
		&\le \sum_{s \in \mathcal{S}} \mu(s) V^{\text{P},\sigma}_1(s,i(s)) \\
		&= V^{\text{P},\pi^\sigma},
	\end{align*}
	which is the expected cumulative principal's utility of the policy $\pi^\sigma$ that implements $\sigma$.
	
	Finally, we observe that the size of $\sigma$ is polynomial in $\nicefrac{1}{\epsilon}$ and the instance size.
	Indeed, the number of possible promises is given by:
	\begin{equation*}
		|\mathcal{I}| \le |\mathcal{D}_\delta| H = \mathcal{O}\left(\frac{HB}{\delta} H \right) = \mathcal{O}\left(\frac{H^4B}{\epsilon} \right)
	\end{equation*}
	and, since $\sigma$ prescribes at most one contract for each action and promise in every state and time step, the number of contracts is bounded by:
	\begin{equation*}
		|\mathcal{J}| \le |\mathcal{I}||\mathcal{S}||\mathcal{A}| = \mathcal{O}\left(|\mathcal{S}||\mathcal{A}| \frac{H^4B}{\epsilon}\right).
	\end{equation*}
	Finally, given that the size of $\mathcal{D}_\delta$ is polynomial in $\nicefrac{1}{\epsilon}$ and the instance size and a call to the oracle can be executed in polynomial time (see Definition~\ref{def:approx_oracle}), it follows that the Algorithm~\ref{alg:approx_policy} runs in polynomial time.
\end{proof}

\section{Building an Approximation Oracle}
\label{appendix:approx_orcale}
Algorithm~\ref{alg:approx_policy} requires an approximation oracle that solves a relaxed version of the problem $\mathcal{P}_{h,s,\iota}(M)$ while satisfying Definition~\ref{def:approx_oracle}.
In this section we describe how to build such an oracle.
As a first step, we introduce the relaxed problem $\mathcal{R}^\delta_{h,s,\iota}(M)$, whose optimal solution allows us to retrieve a tuple $(\alpha,p,z)$ satisfying Definition~\ref{def:approx_oracle}. 
While this problem is still quadratic, it is equivalent to an LP, which can be solved in polynomial time.

\subsection{The relaxed problem}
Algorithm~\ref{alg:approx_oracle} provides an approximation oracle according to Definition~\ref{def:approx_oracle}.
The algorithm computes an optimal solution for the problem $\mathcal{R}^\delta_{h,s,\iota}(M)$ by solving an equivalent LP.
For the ease of notation, we introduce the set $\mathcal{D}^{M,s}_\delta \coloneqq \{\iota \in \mathcal{D}_\delta \mid M(s,\iota)>-\infty\}$.
\begin{algorithm}[H]
	\caption{\texttt{Approximation oracle}}\label{alg:approx_oracle}
	\begin{algorithmic}[1]
		\Require $h \in \mathcal{H}, s \in \sset, \delta \in \mathbb{R}, \iota \in \mathcal{D}_\delta, M: \sset \times \mathcal{D}_\delta \rightarrow \mathbb{R} \cup \{-\infty\}$.
		\If{$\mathcal{R}^\delta_{h,s,\iota}(M)$ is feasible}
		\State $(\alpha,p,\widetilde{q}) \gets $ Solution to $\mathcal{R}^\delta_{h,s,\iota}(M)$
		\State $v \gets G^\delta_{h,s,M}(\alpha,p,\widetilde{q})$
		\State $z(a,s') \gets \argmax_{\iota' \in \mathcal{D}^{M,s'}_\delta : |\iota'-q(a,s')|< \delta} M(s',\iota') \quad \forall s' \in \sset$ \label{line:oracle_discretizaion} 
		\Else
		\State $(\alpha,p,\widetilde{q}) \gets$ anything
		\State $v \gets -\infty$
		\EndIf
		\State \textbf{Return} $(\alpha,p,\widetilde{q},v)$
	\end{algorithmic}
\end{algorithm}

To define $\mathcal{R}^\delta_{h,s,\iota}(M)$, we write the future promises $z(a,s')$ as a discretization of linear combinations of the promises in $\mathcal{D}_\delta$.
This is accomplished by introducing the function $\widetilde{q}: \A \times \sset \rightarrow \Delta(\mathcal{D}_\delta)$, and defining the (continuous) future promise $q(a,s')$ as:
\begin{equation}
	\label{eq:q_from_q_tilde}
	q(a,s') \coloneqq \sum_{\iota' \in \mathcal{D}^{M,s'}_\delta} \widetilde{q}(\iota'|a,s') \iota',
\end{equation}
for every action $a \in \A$ and future state $s' \in \sset$.
Furthermore, we define the set of consistent functions $\widetilde{q}$ as:
\begin{equation*}
	\mathcal{Q}^M_\delta \coloneqq \{\widetilde{q}: \A \times \sset \rightarrow \Delta(\mathcal{D}_\delta) \mid \widetilde{q}(\iota'|a,s')=0 \; \forall a \in \A, s' \in \sset, \iota' \in \mathcal{D}_\delta \setminus \mathcal{D}^{M,s'}_\delta\}.
\end{equation*} 
In other words, $\widetilde{q} \in \mathcal{Q}^M_\delta$ if and only if $\widetilde{q}(a,s') \in \Delta(\mathcal{D}^{M,s'}_\delta)$ for every $a \in \mathcal{A}$ and $s' \in \sset$.

The actual future promise $z(a,s')$ is instead defined by taking the promise $\iota' \in \mathcal{D}^{M,s'}_\delta$ such that $|q(a,s')-z(a,s')|<\delta$ and with the largest value of $M(s',z(a,s'))$, formally:
\begin{equation*}\label{eq:z_from_q}
	z(a,s') = d^{M,s'}_\delta(q(a,s')) \coloneqq \argmax_{\iota' \in \mathcal{D}^{M,s'}_\delta : |\iota'-q(a,s')|< \delta} M(s',\iota').
\end{equation*}
Observe that if $q(a,s') \in \mathcal{D}^{M,s'}_\delta$, then $z(a,s')=q(a,s')$, while otherwise $z(a,s')$ is one of the two promises in $\mathcal{D}^{M,s'}_\delta$ that are nearest to $q(a,s')$.

Given this definitions, we let the objective function of the problem $\mathcal{R}^\delta_{h,s,\iota}(M)$ be:
\begin{equation*}
	G^\delta_{h,s,M}(\alpha,p,\widetilde{q}) \coloneqq \sum_{a \in \A} \alpha_a \sum_{s' \in \sset} P_h(s'|s,a) \left(r^P_h(s,p^a,s') +y(a,s')\right),
\end{equation*}
where $y(a,s') \coloneqq \sum_{\iota' \in \mathcal{D}^{M,s'}_\delta} \widetilde{q}(\iota'|a,s')M(s',\iota')$.

Now we can define the problem $\mathcal{R}^\delta_{h,s,\iota}(M)$ as:
\begin{equation*}
	\widehat{G}^\delta_{h,s,\iota}(M) \coloneqq \max_{\substack{\alpha \in \Delta(\A)\\ p \in \mathcal{P}^{|\A|}\\ \widetilde{q}:\A \times \sset \rightarrow \Delta(\mathcal{D}_\delta)}}
	G^\delta_{h,s,M}(\alpha,p,\widetilde{q}) \quad \text{s.t.} \quad (\alpha, p, q) \in \Psi^{h,s}_{\iota,\delta} \text{ and } \widetilde{q} \in \mathcal{Q}^M_\delta,
\end{equation*}
where $q$ is computed according to Equation~\ref{eq:q_from_q_tilde}.

We rewrite the complete problem $\mathcal{R}^\delta_{h,s,\iota}(M)$ for clarity:
\begin{maxi!}
	{\scriptstyle \substack{\alpha \in \Delta(\A)\\ p \in \mathcal{P}^{|\A|}\\ \widetilde{q}:\A \times \sset \rightarrow \Delta(\mathcal{D}_\delta)}}{ \sum_{a \in \A} \alpha_a \sum_{s' \in \sset} P_h(s'|s,a) \left(r^P_h(s,p^a,s') +y(a,s')\right)}{}{}\label{eq:r_obj}
	\addConstraint { \sum_{a \in A}\alpha_a \sum_{s' \in \sset} P_h(s'|s,a) \left(r^A_h(s,p^a,a,s') +q(a,s')\right) \ge \iota -\delta}{}{}\label{eq:r_honesty_ge}
	\addConstraint { \sum_{a \in A}\alpha_a \sum_{s' \in \sset} P_h(s'|s,a) \left(r^A_h(s,p^a,a,s') +q(a,s')\right) \le \iota +\delta}{}{}\label{eq:r_honesty_le}
	\addConstraint { \alpha_a \sum_{s' \in \sset} P_h(s'|s,a) \left(r^A_h(s,p^a,a,s') +q(a,s')\right) \ge}{}{\nonumber}
	\addConstraint { \hspace{0.5cm}\alpha_a \sum_{s' \in \sset} P_h(s'|s,\widehat{a}) \left(r^A_h(s,p^a,\widehat{a},s') +q(a,s') \right) \;\; \forall a,\widehat{a} \in \A,}{}{}\label{eq:r_ic}
	\addConstraint {\sum_{\iota' \in \mathcal{D}^{M,s'}_\delta} \widetilde{q}(\iota'|a,'s) = 1 \quad \forall a \in A, s' \in \sset}{}{}\label{eq:r_q_distribution}
\end{maxi!}
Since the problem is defined over functions with finite domains, it is possible to represent them as a limited number of real variables by adding opportune linear constraints to represent probability distributions and bounded payments.

Let us remark that the problem $\mathcal{P}_{h,s,\iota}(M)$ is defined optimizing over the functions $z: \A \times \sset \rightarrow \mathcal{D}_\delta$ with values in the set $\mathcal{D}_\delta$, as the table $M : \sset \times \mathcal{D}_\delta$ is not defined for a promise $\iota \notin \mathcal{D}_\delta$.
However, the ``relaxed'' set $\Psi^{h,s}_{\iota,\lambda}$ includes also any function $q: \A \times \sset \rightarrow \mathbb{R}$ with real values.
This means that a solution $(\alpha,p,\widetilde{q})$ of $\mathcal{R}^\delta_{h,s,\iota}(M)$ does not correspond directly to a relaxed solution $(\alpha,p,z)$ by taking $z=q$, since $q(a,s')$ can be any convex combination of the promises in $\mathcal{D}_\delta$, thus may not be in $\mathcal{D}_\delta$ itself.
For this reason, at Line~\ref{line:oracle_discretizaion} Algorithm~\ref{alg:approx_oracle} approximates the value of the future promise $q(a,s')$ to a value $z(a,s') = d_\delta^{M,s'}(q(a,s'))$ in $\mathcal{D}_\delta$.
Due to this further approximation, the solution $(\alpha,p,z)$ does not belong to $\Psi^{h,s}_{\iota,\delta}$, but to a set $\Psi^{h,s}_{\iota,\lambda}$ with a larger $\lambda>\delta$. 

%

\subsection{An equivalent LP}

We observe that the Program~\ref{eq:r_obj} includes some quadratic terms, specifically the products $\alpha_a \widetilde{q}(\iota'|a,s')$ and $\alpha_a p^a(s')$ in the objective function and the constraints represented by Equation~\ref{eq:r_honesty_ge}, Equation~\ref{eq:r_honesty_le} and Equation~\ref{eq:r_ic}.
In order to manage them, we introduce the variables $\gamma^a_{s'}$ (corresponding to $\alpha_a p^a(s')$), $\xi^{\iota'}_{a,s'}$ (corresponding to $\alpha_a \widetilde{q}(\iota'|a,s')$), and $\nu_a$ (corresponding to $\alpha_a$).
Then we provide the linear program $\mathcal{L}^\delta_{h,s,\iota}(M)$ defined as follows:
\begin{maxi!}
	{\scriptstyle \substack{ 
	\gamma^a_{s'} \in \mathbb{R}_{+} \\
	\xi^{\iota'}_{a,s'} \in \mathbb{R}_+ \\
	\nu_a \in \mathbb{R}_+
	}}
	{ \sum_{a \in \A} \sum_{s' \in \sset} P_h(s'|s,a)\left(\nu_a r_h(s,s') -\gamma^a_{s'} +\sum_{\iota' \in \mathcal{D}^{M,s}_\delta} \xi^{\iota'}_{a,s'}M(s',\iota') \right)}{}{}
	\label{eq:lp_obj}
	\addConstraint{ \sum_{a \in A} \sum_{s' \in \sset} P_h(s'|s,a) \left( \gamma^a_{s'} -\nu_a c_h(s,a) +\sum_{\iota' \in \mathcal{D}^{M,s}_\delta} \xi^{\iota'}_{a,s'} \iota' \right) \ge \iota - \delta}{}{}\label{eq:lp_honesty_ge}
	\addConstraint{ \sum_{a \in A} \sum_{s' \in \sset} P_h(s'|s,a) \left( \gamma^a_{s'} -\nu_a c_h(s,a) +\sum_{\iota' \in \mathcal{D}^{M,s}_\delta} \xi^{\iota'}_{a,s'} \iota' \right) \le \iota + \delta}{}{}\label{eq:lp_honesty_le}
	\addConstraint{ \sum_{s' \in \sset} P_h(s'|s,a)\left(\gamma^a_{s'} -\nu_a c_h(s,a) +\sum_{\iota' \in \mathcal{D}^{M,s}_\delta} \xi^{\iota'}_{a,s'} \iota' \right) \ge}{}{\nonumber}
	\addConstraint{\hspace{1cm} \sum_{s' \in \sset} P_h(s'|s,\widehat{a}) \left(\gamma^a_{s'} -\nu_a c_h(s,\widehat{a}) +\sum_{\iota' \in \mathcal{D}^{M,s}_\delta} \xi^{\iota'}_{a,s'} \iota' \right) \quad \forall a,\widehat{a} \in \A}\label{eq:lp_ic}
	\addConstraint{ \gamma^a_{s'} \le B\nu_a \quad \forall a \in \A, s' \in \sset}
	\addConstraint{ \sum_{\iota' \in \mathcal{D}^{M,s'}_\delta} \xi^{\iota'}_{a,s'} = \nu_a \quad \forall a \in A, s' \in \sset}
	\addConstraint{ \sum_{a \in A} \nu_a =1}
\end{maxi!}

Now we show that the problems $\mathcal{R}^\delta_{h,s,\iota}(M)$ and $\mathcal{L}^\delta_{h,s,\iota}(M)$ are equivalent, \emph{i.e.,} given a solution of $\mathcal{R}^\delta_{h,s,\iota}(M)$, one can compute a solution of $\mathcal{L}^\delta_{h,s,\iota}(M)$, and viceversa.
\begin{restatable}{lemma}{solRtoLP}
	\label{lem:sol_r_to_lp}
	Given a feasible solution $(\alpha,p,\widetilde{q})$ of $\mathcal{R}^\delta_{h,s,\iota}$, there exists a solution $(\nu,\gamma,\xi)$ of $\mathcal{L}^\delta_{h,s,\iota}$ with the same value.
\end{restatable}
\begin{proof}
	Consider the solution $(\nu,\gamma,\xi)$ defined as follows:
	\begin{align*}
		\nu_a &= \alpha_a &\forall a \in \A, \\
		\gamma^a_{s'} &= \alpha_a p^a(s') &\forall a \in \A, s' \in \sset, \\
		\xi^{\iota'}_{a,s'} &= \widetilde{q}(\iota'|,a,s') \alpha_a &\forall a \in A, s' \in \sset, \iota' \in D^{M,s}_\delta.
	\end{align*}
	Then one can show by direct calculation that $(\nu, \gamma, \xi)$ is a feasible solution of $\mathcal{L}^\delta_{h,s,\iota}$ and that its value is $G^\delta_{h,s,M}(\alpha,p,\widetilde{q})$.
\end{proof}

\begin{restatable}{lemma}{solLPtoR}
	\label{lem:sol_lp_to_r}
	Given a feasible solution $(\nu,\gamma,\xi)$ of $\mathcal{L}^\delta_{h,s,\iota}$, there exists a feasible solution $(\alpha,p,\widetilde{q})$ of $\mathcal{R}^\delta_{h,s,\iota}$ with at least the same value.
\end{restatable}
\begin{proof}
	Consider the solution $(\alpha,p,\widetilde{q})$ defined as follows:
	\begin{align*}
		\alpha_a &= \nu_a &&\forall a \in \A, \\
		p^a(s') &= \begin{cases}
			\frac{\gamma^a_{s'}}{\alpha_a} & \alpha_a >0 \\
			0 & \alpha_a=0
		\end{cases} &&\forall a \in \A, s' \in \sset, \\
		\widetilde{q}(\iota'|a,s') &= \begin{cases}
			\frac{\xi^{\iota'}_{a,s'}}{\alpha_a} & \alpha_a > 0 \\
			0 & \alpha_a = 0 
		\end{cases} &&\forall a \in \A, s' \in \sset, \iota' \in D^{M,s}_\delta.
	\end{align*}
	One can easily observe that the following holds:
	\begin{align*}
		\nu_a &= \alpha_a &&\forall a \in \A, \\
		\gamma^a_{s'} &= \alpha_a p^a(s') &&\forall a \in \A, s' \in \sset, \\
		\xi^{\iota'}_{a,s'} &= \widetilde{q}(\iota'|,a,s') \alpha_a &&\forall a \in A, s' \in \sset, \iota' \in D^{M,s}_\delta.
	\end{align*}
	Thus, by substituting the values of $(\nu,\gamma,\xi)$ in Constraint~\ref{eq:r_honesty_ge}, Constraint~\ref{eq:r_honesty_le}, and Constraint~\ref{eq:r_ic} of Program~\ref{eq:r_obj}, one can observe that they are equivalent to Constraint~\ref{eq:lp_honesty_ge}, Constraint~\ref{eq:lp_honesty_le}, and Constraint~\ref{eq:lp_ic} of Program~\ref{eq:lp_obj} respectively.
	Furthermore, the remaining constraints of Program~\ref{eq:lp_obj} make sure that the payments are correctly bounded by $B$ and $\alpha$ and $\widetilde{q}$ represent distributions.
	As a result, $(\alpha,p,\widetilde{q})$ is a feasible solution for Program~\ref{eq:r_obj}.
	Furthermore, direct calculations show that it achieves the same value of $(\nu, \gamma, \xi)$.
\end{proof}

By employing Lemma~\ref{lem:sol_lp_to_r} and Lemma~\ref{lem:sol_r_to_lp}, we can show that The problem $\mathcal{R}^\delta_{h,s,\iota}(M)$ can be solved in polynomial time, as stated by the following Lemma:
\begin{restatable}{lemma}{rPolyTime}
	\label{lem:lp_poly_time}
	The problem $\mathcal{R}^\delta_{h,s,\iota}(M)$ can be solved in time polynomial in $\nicefrac{1}{\delta}$ and the instance size.
\end{restatable}
\begin{proof}
	Consider an optimal feasible solution $(\nu,\gamma,\xi)$ of $\mathcal{L}^\delta_{h,s,\iota}(M)$.
	Such a solution can be computed in polynomial time, since $\mathcal{L}^\delta_{h,s,\iota}(M)$ is a linear program with a number of variables and constraints polynomial in $\nicefrac{1}{\delta}$ and the instance size.
	Now apply the transformation described in the proof of Lemma~\ref{lem:sol_lp_to_r} to obtain a feasible solution $(\alpha,p,\tilde{q})$ of $\mathcal{R}^\delta_{h,s,\iota}(M)$ with the same value.
	
	We show that such a solution is optimal.
	In order to do this, suppose by contradiction that there exists a different feasible solution $(\alpha',p',\tilde{q}')$ such that $G^\delta_{h,s,\iota}(\alpha',p',\tilde{q}') > G^\delta_{h,s,\iota}(\alpha,p,\tilde{q})$.
	Then apply the transformation provided in the proof of Lemma~\ref{lem:sol_r_to_lp} to obtain a feasible solution $(\nu',\gamma',\xi')$ of $\mathcal{L}^\delta_{h,s,\iota}(M)$ with the same value $G^\delta_{h,s,\iota}(\alpha',p',\tilde{q}')$.
	As such, according to Lemma~\ref{lem:sol_r_to_lp} the solution $(\nu',\gamma',\xi')$ achieves a value strictly larger than $(\nu,\gamma,\xi)$, which contradicts the hypothesis that $(\nu,\gamma,\xi)$ is an optimal solution.
\end{proof}

\subsection{Approximation oracle and concavity properties}
Lemma~\ref{lem:lp_poly_time} implies that Algorithm~\ref{alg:approx_oracle} can be executed in polynomial time.
Now we show that Algorithm~\ref{alg:approx_oracle} is an approximation oracle accordingly to Definition~\ref{def:approx_oracle}.
First, we need to introduce two additional technical lemmas.
In particular, we show that for any step $h \in \mathcal{H}$ and state $s \in \sset$, there exists a promise $\bar{\iota} \in \mathcal{D}_\delta$ such that $\mathcal{D}^{M^\delta_h,s}_\delta = [0,\bar{\iota}] \cap \mathcal{D}_\delta$.
Then, we prove that the function $x:[0,\bar{\iota}] \rightarrow \mathbb{R}$ defined as $x(\iota) \coloneqq \widehat{G}_{h,s,\iota}(M)$ is concave.

\begin{restatable}{lemma}{intervalDdelta}
	\label{lem:interval_d_delta}
	Suppose that Algorithm~\ref{alg:approx_oracle} is used as an approximate oracle $O^{h,s}_{\iota,\delta}$ during the entire execution of Algorithm~\ref{alg:approx_policy}.
	Then for every state $s \in \sset$ and step $h \in \mathcal{H}$, $0 \in \mathcal{D}^{M^\delta_h,s}_\delta$, and if a given $\iota \in \mathcal{D}_\delta$ does not belong to $\mathcal{D}^{M^\delta_h,s}_\delta$, then no $\iota'>\iota$ belongs to $\mathcal{D}^{M^\delta_h,s}_\delta$.
\end{restatable}
\begin{proof}
	In order to prove that $0 \in \mathcal{D}^{M^\delta_h,s}_\delta$, we reason by induction.
	We observe that for every state $s' \in \sset$, it holds that $M^\delta_{H+1}(s',0)=0$, while $M^\delta_{H+1}(s',\iota')=-\infty$ for every $\iota'>0$.
	Thus, the base case $h=H+1$ trivially satisfy the statement.
	Consider $h\le H$ and any state $s \in \sset$ and the promise $\iota=0$.
	We prove that the problem $\mathcal{R}^\delta_{h,s,0}(M)$ admits a solution.
	We recall that there is an action $\bar{a} \in \A$ such that $c_H(s,\bar{a})=0$.
	Thus, we consider the solution $(\alpha,p,\widetilde{q})$ defined as follows:
	\begin{align*}
		\alpha_a &= 
		\begin{cases}
			1, &a=\bar{a} \\
			0, &\text{otherwise}	
		\end{cases} &&\forall a \in \A \\
		p^a(s') &= 0 &&\forall a \in \A, s' \in \sset \\
		\widetilde{q}(\iota'|a,s') &=
		\begin{cases}
			1, &\iota'=0 \\
			0, &\text{otherwise}
		\end{cases} &&\forall a \in \A, s' \in \sset, \iota' \in \mathcal{D}_\delta.
	\end{align*}
	One can easily observe that this is a feasible solution to $\mathcal{R}^\delta_{h,s,0}(M)$.
	Consequently, $0 \in \mathcal{D}^{M^\delta_h,s}_\delta$.
	
	As a further step, we prove that if a given $\bar{\iota} \in \mathcal{D}_\delta$ does not belong to $\mathcal{D}^{M^\delta_h,s}_\delta$, then no $\widehat{\iota}>\bar{\iota}$ belongs to $\mathcal{D}^{M^\delta_h,s}_\delta$ (\emph{i.e.}, $\mathcal{R}^\delta_{h,s,\widehat{\iota}}(M)$ is unfeasible).
	Suppose, by contradiction, that there exists a a $\widehat{\iota}>\bar{\iota}$ such that $\widehat{\iota} \in \mathcal{D}^{M^\delta_h,s}_\delta$.
	We build an optimization problem $\mathcal{V}^\delta_{h,s,\widehat{\iota}}(M)$ as follows: take Program~\ref{eq:lp_obj} (\emph{i.e.},$\mathcal{L}^\delta_{h,s,\iota}(M)$) and let $\iota$ be a variable instead of a constant.
	Furthermore, add the constraint that $0 \le \iota \le \widehat{\iota}$.
	One can easily observe that $(\nu,\gamma,\xi,\widehat{\iota})$ is a feasible solution for $\mathcal{V}^\delta_{h,s,\widehat{\iota}}(M)$ if and only if $(\nu,\gamma,\xi)$ is a feasible solution for $\mathcal{L}^\delta_{h,s,\iota}(M)$.
	Since both $\mathcal{R}^\delta_{h,s,0}(M)$ and $\mathcal{R}^\delta_{h,s,\widehat{\iota}}(M)$ are feasible, by Lemma~\ref{lem:sol_r_to_lp} $\mathcal{L}^\delta_{h,s,0}(M)$ and $\mathcal{L}^\delta_{h,s,\widehat{\iota}}(M)$ admit feasible solutions too.
	Consequently, since  the feasibility space of $\mathcal{V}^\delta_{h,s,\widehat{\iota}}(M)$ is a polytope, it follows that there exist some $\alpha$, $p$ and $\widetilde{q}$ such that $(\alpha,p,\widetilde{q},\bar{\iota})$ is a feasible solution for $\mathcal{V}^\delta_{h,s,\widehat{\iota}}(M)$, as $0 \le \bar{\iota} \le \widehat{\iota}$.
	As a result, $\mathcal{L}^\delta_{h,s,\bar{\iota}}(M)$ and $\mathcal{R}^\delta_{h,s,\bar{\iota}}(M)$ are feasible, thus $\bar{\iota} \in \mathcal{D}^{M^\delta_h,s}_\delta$, which is a contradiction.
\end{proof}

\begin{restatable}{lemma}{concave}
	\label{lem:concave}
	Consider a fixed $h \in \mathcal{H}$, $s \in \mathcal{S}$ and let $M=M^\delta_{h+1}$.
	If $\mathcal{D}^{M^\delta_h,s}_\delta = \{\iota \in \mathcal{D}_\delta \mid \iota \le \bar{\iota}\}$ for some $\bar{\iota} \in \mathcal{D}_\delta$, then the function $x:[0,\bar{\iota}] \rightarrow \mathbb{R}$ defined as $x(\iota) \coloneqq \widehat{G}_{h,s,\iota}(M)$ is concave.
\end{restatable}
\begin{proof}
	Consider Program~\ref{eq:lp_obj} and $\widehat{L}_{h,s,\iota}(M)$ its optimal solution. 
	Thanks to Theorem 5.1 by~\cite{Bertsimas}, we have that the function $x'(\iota) = \widehat{L}^\delta_{h,s,\iota}(M)$ for any fixed step $h \in \mathcal{H}$ and state $s \in \sset$ is concave.
	Furthermore, the optimal solution of Program~\ref{eq:lp_obj} has the same value of the optimal solution of Problem~\ref{eq:r_obj} (see the proof of Lemma~\ref{lem:lp_poly_time}).
	Consequently:
	\begin{equation*}
		x(\iota) = \widehat{G}_{h,s,\iota}(M) = \widehat{L}_{h,s,\iota}(M) = x'(\iota),
	\end{equation*}
	thus $x$ is concave.
\end{proof}

\begin{restatable}{lemma}{approxOracle}
	\label{lem:approx_orcale}
	If Algorithm~\ref{alg:approx_oracle} is used as an approximate oracle $O^{h,s}_{\iota,\delta}$ during the entire execution of Algorithm~\ref{alg:approx_policy}, then it satisfies Definition~\ref{def:approx_oracle}.
\end{restatable}
\begin{proof}
	Consider a fixed state $s \in \sset$, time step $h \in \mathcal{H}$ and promise $\iota \in \mathcal{D}_\delta$.
	In the following we let $M \coloneqq M^\delta_{h+1}$ and $(\alpha,p,z,v) = O^{h,s}_{\iota,\delta}(M)$ computed accordingly to Algorithm~\ref{alg:approx_oracle}.
	
	As a first step, we observe that $M(s',z(a,s')) > -\infty$ for every $a \in \A$ and $s' \in \sset$.
	This can be proved by considering that Algorithm~\ref{alg:approx_oracle} is such that $z(a,s') \in \mathcal{D}^{M,s}_\delta$ (see Line~\ref{line:oracle_discretizaion}).
	
	Now we prove that $v \ge \widehat{F}_{h,s,\widetilde{\iota}}(M)$ for any $\widetilde{\iota} \in [\iota-\delta,\iota+\delta]$.
	If the problem $\mathcal{P}_{h,s,\widetilde{\iota}}(M)$ does not admit any feasible solution, then $\widehat{F}_{h,s,\widetilde{\iota}}(M) = -\infty$ and thus $v \ge \widehat{F}_{h,s,\iota}(M)$.
	Suppose instead that $\mathcal{P}_{h,s,\widetilde{\iota}}(M)$ admits an optimal feasible solution $(\alpha^\star,p^\star,z^\star)$.
	Then we build a feasible solution $(\alpha,p,\widetilde{q})$ for the problem $\mathcal{R}^\delta_{h,s,\iota}$ as follows:
	\begin{align*}
		\alpha_a &= \alpha^\star_a &&\forall a \in \A \\
		p^a(s') &= p_a^\star(s') &&\forall a \in \A, s' \in \sset \\
		\widetilde{q}(\iota'|a,s') &=
		\begin{cases}
			1, &\iota' = z^\star(a,s') \\
			0, & \text{otherwise}
		\end{cases}
		&&\forall a \in A, s' \in \sset, \iota' \in \mathcal{D}_\delta,
	\end{align*}
	where $p^\star_a$ denotes the optimal contract for action $a \in \A$.
	
	We observe that $q(a,s') = z^\star(a,s')$ for all $a \in A$ and $s' \in \sset$. 
	Thus, as $(\alpha^\star,p^\star,z^\star) \in \Psi^{h,s}_{\widetilde{\iota},0}$, it follows that $(\alpha,p,q) \in \Psi^{h,s}_{\widetilde{\iota},0} \subseteq \Psi^{h,s}_{\iota,\delta}$.
	Furthermore, by considering that $\widetilde{q} \in \mathcal{Q}^M_\delta$, we have that is $(\alpha,p,\widetilde{q})$ a feasible solution for $\mathcal{R}^\delta_{h,s,\iota}$.
	
	Given this feasible solution, we show that $v \ge \widehat{F}_{h,s,\widetilde{\iota}}(M)$.
	For every $a \in A$, and $s' \in \sset$, it holds that:
	\begin{equation*}
		y(a,s') = \sum_{\iota' \in \mathcal{D}^{M,s}_\delta} \widetilde{q}(\iota'|a,s')M(s',\iota') 
		=M(s',z^\star(a,s')).
	\end{equation*}
	Thus, by considering that $\alpha=\alpha^\star$ and $p=p^\star$, it follows that:
	\begin{align*}
		\widehat{F}_{h,s,\widetilde{\iota}}(M) &= F_{h,s,M}(\alpha^\star,p^\star,z^\star) \\
		&= G^\delta_{h,s,M}(\alpha,p,\widetilde{q}) \\
		&\le \widehat{G}^\delta_{h,s,\iota}(M) =v.
	\end{align*}
	
	As a further step, we show that $v \le F_{h,s,M}(\alpha,p,z)$.
	Let $(\alpha,p,\widetilde{q})$ be the optimal feasible solution of $\mathcal{R}^\delta_{h,s,\iota}(M)$, so that $z$ is the discretization of $q$ (Line~\ref{line:oracle_discretizaion} Algorithm~\ref{alg:approx_oracle}) computed according to Equation~\ref{eq:z_from_q}.
	In order to do prove that $v \le F_{h,s,M}(\alpha,p,z)$, we prove that:
	\begin{equation}\label{eq:m_larger_than_y}
		M(s',z(a,s')) \ge y(a,s') = \sum_{\iota' \in \mathcal{D}^{M,s'}_\delta} \widetilde{q}(\iota'|a,s')M(s',\iota'),
	\end{equation}
	for every $a \in \A$, and $s' \in \sset$.
	Such a result would imply that $v = G_{h,s,M}(\alpha,p,\widetilde{q}) \le F_{h,s,M}(\alpha,p,z)$, by simply taking the definitions of these two functions.
	
	Since Equation~\ref{eq:m_larger_than_y} is trivially satisfied for $h=H$, as $M(s',z(a,s')) = y(a,s') = 0$ we assume $h<H$.
	Let us consider a fixed action $a \in \mathcal{A}$ and future state $s'\in \sset$.
	Thanks to Lemma~\ref{lem:interval_d_delta}, there exists some $\bar{\iota} \in \mathcal{D}_\delta$ such that $\mathcal{D}^{M,s'}_\delta =\{\iota' \in \mathcal{D}_\delta \mid \iota' \le \bar{\iota}\}$, while $z(a,s') \in \mathcal{D}^{M,s'}_\delta$ by definition (Equation~\ref{eq:z_from_q}).
	Furthermore, let us define:
	\begin{equation*}
		\mathcal{M} \coloneqq \{(\iota',m) \mid \iota' \in \mathcal{D}^{M,s'}_\delta, m = M(s',\iota')\}.
	\end{equation*} 
	If we define the concave (Lemma~\ref{lem:concave}) function $x(\iota') \coloneqq \widehat{G}_{h+1,s',\iota'}(M^\delta_{h+2})$, then we have that $M(s',\iota') = x(\iota')$ for every $\iota' \in \mathcal{D}^{M,s'}_\delta$.
	As a result, the polygon $\text{co}(\mathcal{M})$ is the convex hull of the points $(\iota',x(\iota'))$, with $\iota' \in \mathcal{D}^{M,s'}_\delta$ and $x$ being a concave function.
	We further observe that both the points $(q(a,s'),y(a,s'))$ and $(z(a,s'),M(s',z(a,s'))) = (z(a,s'),x(z(a,s'))) \in \mathcal{M}$ belong to $\text{co}(\mathcal{M})$.

	Suppose that $q(a,s') \in \mathcal{D}^{M,s'}_\delta$. 
	Then $q(a,s') = z(a,s')$, as there are no other values in $\mathcal{D}^{M,s'}_\delta$ at distance less than $\delta$ from $q(a,s')$.
	As the points $(z(a,s'),y(a,s'))$ and $(z(a,s'),x(z(a,s')))$ both belong to $\text{co}(\mathcal{M})$ and $x$ is concave, it holds that $y(a,s') \le x(z(a,s')) = M(s',z(a,s'))$.
	
	On the other hand, if $q(a,s') \notin \mathcal{D}^{M,s'}_\delta$, then $0 < q(a,s') < \bar{\iota}$ and the discrete values$\lfloor q(a,s') \rfloor_\delta$ and $\lceil q(a,s') \rceil_\delta$ belong to $\mathcal{D}^{M,s'}_\delta$. 
	As the polygon $\text{co}(\mathcal{M})$ is the convex hull of the points $(\iota',x(\iota'))$, with $\iota' \in \mathcal{D}^{M,s'}_\delta$ and $x$ is a concave function, it follows that either $q(a,s') \le x(\lfloor q(a,s') \rfloor_\delta)$ or $q(a,s') \le x(\lceil q(a,s') \rceil_\delta)$.
	Consequently, by construction Line~\ref{line:oracle_discretizaion} Algorithm~\ref{alg:approx_oracle} selects $z(a,s') \in \{\lfloor q(a,s') \rfloor_\delta, \lceil q(a,s') \rceil_\delta\}$ such that $M(s',z(a,s')) = x(z(a,s')) \ge y(a,s')$.
	
	To conclude the proof, we show that if $v> -\infty$, then $(\alpha,p,z) \in \Psi^{h,s}_{\iota,\lambda}$, with $\lambda = 2\delta$.
	Again, let $(\alpha,p,\widetilde{q})$ be the optimal solution to $\mathcal{R}^\delta_{h,s,\iota}(M)$.
	We must prove that $(\alpha,p,z)$, obtained according to Line~\ref{line:oracle_discretizaion} Algorithm~\ref{alg:approx_oracle} (Equation~\ref{eq:z_from_q}), belongs to $\Psi^{h,s}_{\iota,\lambda}$, \emph{i.e}, it satisfies Equation~\ref{eq:opt_rel_prob_constr_honest_ge}, Equation~\ref{eq:opt_rel_prob_constr_honest_le} and Equation~\ref{eq:opt_rel_prob_constr_ic}.
	We observe that, according to Equation~\ref{eq:z_from_q}, for any future state $s' \in \sset$ and action $a \in \A$, it holds that $|z(a,s') - q(a,s')| \le \delta$.
	Thus, Equation~\ref{eq:opt_rel_prob_constr_honest_ge} holds as:
	\begin{align*}
		&\sum_{a \in A}\alpha_a \sum_{s' \in \sset} P_h(s'|s,a) \left(r^A_h(s,p^a,a,s') +z(a,s')\right) \\
		&\ge \sum_{a \in A}\alpha_a \sum_{s' \in \sset} P_h(s'|s,a) \left(r^A_h(s,p^a,a,s') + q(a,s') -\delta \right) \\
		&= \sum_{a \in A}\alpha_a \sum_{x' \in \sset} P_h(s'|s,a) \left(r^A_h(s,p^a,a,s') +q(a,s')\right) -\delta \\
		&\ge \iota -2\delta = \iota - \lambda,
	\end{align*}
	where the equality holds because $\sum_{a \in \mathcal{A}}\alpha_a \sum_{s' \in \sset} P_h(s'|s,a) = 1$, while  the last inequality holds because Constraint~\ref{eq:r_honesty_ge} is satisfied, as $(\alpha,p,q)$ is a feasible solution for Program~\ref{eq:opt_rel_prob_constr_honest_ge}.
	Similarly, we prove that Equation~\ref{eq:opt_rel_prob_constr_honest_le} is satisfied.
	
	Let us define $k_{a,s'} \coloneqq z(a,s')-q(a,s')$.
	Then, by applying Equation~\ref{eq:r_ic}, we can show that for every pair of actions $a,\widehat{a} \in \A$ such that $\alpha_a>0$:
	\begin{align*}
		&\sum_{s' \in \sset} P_h(s'|s,a) \left(r^A_h(s,p^a,a,s') +z(a,s')\right) \\
		&= \sum_{s' \in \sset} P_h(s'|s,a) \left(r^A_h(s,p^a,a,s') +q(a,s')\right) +k_{a,s'} \\
		&\ge \sum_{s' \in \sset} P_h(s'|s,\widehat{a}) \left(r^A_h(s,p^a,\widehat{a},s') +q(a,s') \right) +k_{a,s'} \\
		&=\sum_{s' \in \sset} P_h(s'|s,\widehat{a}) \left(r^A_h(s,p^a,\widehat{a},s') +z(a,s') \right),
	\end{align*}
	thus implying that Equation~\ref{eq:opt_rel_prob_constr_ic} is satisfied.
	As a result, $(\alpha,p,z) \in \Psi^{h,s}_{\iota,\lambda}$, with $\lambda = 2\delta$, concluding the proof.
\end{proof}

\end{document}